\documentclass[journal]{IEEEtran}
\usepackage{amsmath,graphicx}
\usepackage{amsmath,amssymb,epsfig,psfrag,cite,subfigure}
\include{macros}
\usepackage{graphicx}
\usepackage{epstopdf}

\usepackage{bm}

\usepackage{geometry}
\usepackage{tcolorbox}

\usepackage{accents}
\allowdisplaybreaks
\usepackage{soul}

\usepackage{amsthm}

\usepackage{nicefrac}

\usepackage{lipsum,multicol}

\usepackage{algorithm}
\usepackage{algpseudocode}
\algdef{SE}[SUBALG]{Indent}{EndIndent}{}{\algorithmicend\ }%
\algtext*{Indent}
\algtext*{EndIndent}

\usepackage{mwe}    
\usepackage{epsfig,psfrag}
\usepackage{subfigure}
\usepackage{color}
\usepackage{url}
\usepackage{mathtools,xparse}

\usepackage{gensymb}

\usepackage[multiple]{footmisc}

\usepackage{xr}
\makeatletter
\newcommand*{\addFileDependency}[1]{% argument=file name and extension
  \typeout{(#1)}
  \@addtofilelist{#1}
  \IfFileExists{#1}{}{\typeout{No file #1.}}
}
\makeatother

\newcommand*{\myexternaldocument}[1]{%
    \externaldocument{#1}%
    \addFileDependency{#1.tex}%
    \addFileDependency{#1.aux}%
}

\myexternaldocument{supp}

\usepackage{xcolor}

\theoremstyle{remark}

\newtheoremstyle{mytheoremstyle} % name
    {\topsep}                    % Space above
    {\topsep}                    % Space below
    {\upshape}                   % Body font
    {.5em}                           % Indent amount
    {\itshape}                   % Theorem head font
    {.}                          % Punctuation after theorem head
    {.5em}                       % Space after theorem head
    {}  % Theorem head spec (can be left empty, meaning ‘normal’)

\theoremstyle{plain}

\newtheoremstyle{iremark}
  {\topsep}   % ABOVESPACE
  {\topsep}   % BELOWSPACE
  {\upshape}  % BODYFONT
  {0.2in}       % INDENT (empty value is the same as 0pt)
  {\itshape}  % HEADFONT
  {.}         % HEADPUNCT
  {5pt plus 1pt minus 1pt} % HEADSPACE
  {\thmname{#1}\thmnumber{ \itshape#2}\thmnote{ (#3)}} % CUSTOM-HEAD-SPEC

\newtheorem{theorem}{Theorem}
\newtheorem{lemma}[theorem]{Lemma}

\newtheorem{proposition}{Proposition}
\theoremstyle{definition}
\newtheorem*{proof}{Proof}

\DeclareFontFamily{U}{mathx}{\hyphenchar\font45}
\DeclareFontShape{U}{mathx}{m}{n}{
	<5> <6> <7> <8> <9> <10>
	<10.95> <12> <14.4> <17.28> <20.74> <24.88>
	mathx10
}{}
\DeclareSymbolFont{mathx}{U}{mathx}{m}{n}
\DeclareFontSubstitution{U}{mathx}{m}{n}
\DeclareMathOperator*{\E}{\mathbb{E}}

\DeclarePairedDelimiter\abs{\lvert}{\rvert}%
\DeclarePairedDelimiter\absbig{\Big\lvert}{\Big\rvert}%

\makeatletter \renewcommand\d[1]{\ensuremath{%
		\;\mathrm{d}#1\@ifnextchar\d{\!}{}}}
\makeatother

\makeatletter
\newcommand*\rel@kern[1]{\kern#1\dimexpr\macc@kerna}
\newcommand*\widebar[1]{%
  \begingroup
  \def\mathaccent##1##2{%
    \rel@kern{0.8}%
    \overline{\rel@kern{-0.8}\macc@nucleus\rel@kern{0.2}}%
    \rel@kern{-0.2}%
  }%
  \macc@depth\@ne
  \let\math@bgroup\@empty \let\math@egroup\macc@set@skewchar
  \mathsurround\z@ \frozen@everymath{\mathgroup\macc@group\relax}%
  \macc@set@skewchar\relax
  \let\mathaccentV\macc@nested@a
  \macc@nested@a\relax111{#1}%
  \endgroup
}
\makeatother

\usepackage{scalerel,stackengine}
\stackMath
\newcommand\widecheck[1]{%
\savestack{\tmpbox}{\stretchto{%
  \scaleto{%
    \scalerel*[\widthof{\ensuremath{#1}}]{\kern-.6pt\bigwedge\kern-.6pt}%
    {\rule[-\textheight/2]{1ex}{\textheight}}%WIDTH-LIMITED BIG WEDGE
  }{\textheight}% 
}{0.5ex}}%
\stackon[1pt]{#1}{\scalebox{-1}{\tmpbox}}%
}

\newcommand{\rev}[1]{\textcolor{black}{#1}} % revised parts for main.tex
\newcommand{\norm}[1]{\left\lVert#1\right\rVert}
\newcommand{\normsmall}[1]{\big\lVert#1\big\rVert}

\newcommand{\thn}[1]{ {#1^{\rm{th} } } }

\newcommand{\snr}{{\rm{SNR}}}

\newcommand{\sigmaphi}{\sigma^2_{\phi}}
\newcommand{\rphii}{R_{\phi \phi}}
\newcommand{\rxii}{R_{\xi \xi}}
\newcommand{\deltat}{\Delta t}

\newcommand{\sigmaxi}{\sigma^2_{\xi}}

\newcommand{\Tsym}{ T_{\rm{sym}} }
\newcommand{\Ts}{ T_{\rm{s}} }
\newcommand{\Tcp}{ T_{\rm{cp}} }

\newcommand{\FF}{ \mathbf{F} }

\newcommand{\deltaf}{ \Delta f }

\newcommand{\fc}{ f_c }

\newcommand{\bxidelta}{ \bxi_{\Delta} }
\newcommand{\bxideltahat}{ \bxihat_{\Delta} }
\newcommand{\bomega}{ \boldsymbol{\omega} }
\newcommand{\bkappa}{ \boldsymbol{\kappa} }

\newcommand{\quot}[1]{``{#1}''}

\newcommand{\sm}{ s_{m} }

\newcommand{\llr}{\mathcal{L}}

\newcommand{\llrt}{\widetilde{\llr}}

\newcommand{\projrange}[1]{\boldsymbol{\Pi}_{#1}}
\newcommand{\projnull}[1]{\boldsymbol{\Pi}^{\perp}_{#1}}

\newcommand{\xnm}{ x_{n,m} }

\newcommand{\boldRhat}{ \widehat{\boldR} }

\newcommand{\rrhat}{ \widehat{\rr} }
\newcommand{\rhat}{ \widehat{r} }

\newcommand{\alphahat}{ \widehat{\alpha} }

\newcommand{\stilde}{ \widetilde{s} }
\newcommand{\ytilde}{ \widetilde{y} }

\newcommand{\sigmatau}{{ \sigma_{\xi}^2(\tau) }}

\newcommand{\mtN}{{\mathcal{N}}}

\newcommand{\fdb}{ f_{\rm{3dB}} }

\newcommand{\floop}{ f_{\rm{loop}} }

\newcommand{\taup}{  \tau_{\rm{p}} }

\newcommand{\Mod}[1]{\ (\mathrm{mod}\ #1)}

\newcommand{\rect}[1]{ { \rm{rect} }\left(#1\right) }

\newcommand{\mtCN}{{\mathcal{CN}}}

\newcommand{\trace}[1]{ {{{\rm{tr}}\bigg( #1 \bigg)}}  }

\newcommand{\realp}[1]{ \Re \left\{#1\right\}  }
\newcommand{\imp}[1]{ \Im \left\{#1\right\}  }
\newcommand{\vecc}[1]{ {\rm{vec}}\left(#1\right)  }
\newcommand{\veccs}[1]{ {\rm{vec}}\big(#1\big)  }

\newcommand{\veccinv}[1]{ {\rm{reshape}}_{N,M}\left(#1\right)  }

\newcommand{\epstau}{\epsilon_{\tau}}
\newcommand{\epsnu}{\epsilon_{\nu}}
\newcommand{\imax}{I_{\rm{max}}}

\newcommand{\diag}[1]{ {\rm{diag}}\left(#1\right)  }

\newcommand{\blkdiagg}[1]{ {\rm{blkdiag}}\big(#1\big)  }

\newcommand{\Imatrix}{{ \boldsymbol{\mathrm{I}} }}

\newcommand{\cc}{ \mathbf{c} }
\newcommand{\bb}{ \mathbf{b} }
\newcommand{\nn}{ \mathbf{n} }

\newcommand{\nuhat}{{ \widehat{\nu} }}
\newcommand{\tauhat}{{ \widehat{\tau} }}
\newcommand{\tauhatu}{{ \tauhat^{\rm{true}} }}

\newcommand{\boldzero}{{ {\boldsymbol{0}} }}
\newcommand{\boldone}{{ {\boldsymbol{1}} }}

\newcommand{\boldYfr}{\boldY^{\rm{PN-free}}}

\newcommand{\detm}[1]{ {{{\rm{det}}\left( #1 \right)}}  }

\newcommand{\boldY}{ \mathbf{Y} }

\newcommand{\boldX}{ \mathbf{X} }

\newcommand{\Wbhat}{ \widehat{\boldW} }

\newcommand{\boldYtilde}{ \widetilde{\boldY} }

\newcommand{\boldW}{ \mathbf{W} }

\newcommand{\boldZ}{ \mathbf{Z} }

\newcommand{\boldRcirc}{ \boldR^{{\rm{circ}}} }

\newcommand{\boldD}{ \boldGamma }

\newcommand{\boldA}{ \mathbf{A} }
\newcommand{\boldR}{ \mathbf{R} }

\newcommand{\qq}{ \mathbf{q} }
\newcommand{\yy}{ \mathbf{y} }

\newcommand{\xx}{ \mathbf{x} }

\newcommand{\ww}{ \mathbf{w} }
\newcommand{\zz}{ \mathbf{z} }

\newcommand{\rr}{ \mathbf{r} }

\newcommand{\boldSigma}{ \mathbf{\Sigma} }

\newcommand{\boldXi}{ \mathbf{\Xi} }
\newcommand{\boldXihat}{ \widehat{\boldXi} }

\newcommand{\upsb}{ \boldsymbol{\upsilon} }
\newcommand{\varpib}{ \boldsymbol{\varpi} }
\newcommand{\varsigmab}{ \boldsymbol{\varsigma} }

\newcommand{\bxi}{ \boldsymbol{\xi} }
\newcommand{\bxihat}{ \widehat{\bxi} }
\newcommand{\xihat}{ \widehat{\xi} }

\newcommand{\boldGamma}{ \mathbf{\Gamma} }

\newcommand{\transpose}[1]{ {#1}^{T} }

\newcommand{\complexset}[2]{ \mathbb{C}^{#1 \times #2}  }

\newcommand{\realset}[2]{ \mathbb{R}^{#1 \times #2}  }

\newcommand{\conj}{ {\ast} }

\newcommand{\etab}{ {\boldsymbol{\eta}} }

\newcommand{\Lambdab}{ {\boldsymbol{\Lambda}} }

\newcounter{relctr} %% <- counter for relations
\everydisplay\expandafter{\the\everydisplay\setcounter{relctr}{0}} %% <- reset every eq
\newcommand\labelrel[2]{%
  \begingroup
    \refstepcounter{relctr}%
    \stackrel{\textnormal{(\alph{relctr})}}{\mathstrut{#1}}%
    \originallabel{#2}%
  \endgroup
}
\AtBeginDocument{\let\originallabel\label} %% <- store original definition

\graphicspath{{./Figures/}}

\begin{document}
\bstctlcite{IEEEexample:BSTcontrol}

%%%%%%%%%%%%%%%%%% title page information %%%%%%%%%%%%%%%%%%
\title{\rev{Monostatic Sensing with OFDM under Phase Noise: From Mitigation to Exploitation}}
% \title{OFDM Joint Radar-Communications under Phase Noise: From Mitigation to Exploitation}
% \title{ICI-Aware Sensing with MIMO-OFDM Dual-Functional Radar-Communications: Parameter Estimation and Joint Beamforming Design}
% \title{OFDM Radar with Phase Noise}
% title alternatives:
%\title{Optimistic Sensing with OFDM Radar: Can We Exploit ICI Under Phase Noise and High-Mobility? }
%\title{Simultaneous Compensation and Exploitation of ICI in OFDM Radar Under Phase Noise and Mobility}
% \title{Phase Noise Enhanced OFDM Communication Radar: Joint Compensation and Exploitation of ICI}
% \title{Compensation and Exploitation of ICI for Phase-Noise-Enhanced OFDM Communication-Radar}
% \title{Phase Noise Enhanced OFDM Communication-Radar under High Mobility: Joint Compensation and Exploitation of ICI}
% \title{Phase Noise Enhanced Sensing with OFDM Joint Radar-Communications: Joint Mitigation and Exploitation of ICI}
% \title{Phase Noise Enhanced Sensing for OFDM Joint Radar-Communications}
% \title{Phase Noise Improves Ranging Performance: An OFDM RadCom Perspective}

\author{Musa Furkan Keskin, \textit{Member, IEEE}, Henk Wymeersch, \textit{Senior Member, IEEE}, and Visa Koivunen, \textit{Fellow, IEEE}\thanks{Musa Furkan Keskin and Henk Wymeersch are with the Department of Electrical Engineering, Chalmers University of Technology, SE 41296 Gothenburg, Sweden (e-mail: furkan@chalmers.se). Visa Koivunen is with the Department of Signal Processing and Acoustics, Aalto University, FI 00076 Aalto, Finland. This work is supported, in part, by MSCA-IF grant
888913 (OTFS-RADCOM) and the European Commission through the H2020 project Hexa-X (Grant Agreement no. 101015956).}}

% make the title area
\maketitle

\begin{abstract}
     %Phase noise (PN) can become a major bottleneck for OFDM joint radar-communications (JRC) systems in beyond 5G wireless networks. 
     We consider the problem of monostatic radar sensing with \rev{orthogonal frequency-division multiplexing} (OFDM) joint radar-communications (JRC) systems in the presence of phase noise (PN) caused by oscillator imperfections. We begin by providing a rigorous statistical characterization of PN in the radar receiver over multiple OFDM symbols for free-running oscillators (FROs) and phase-locked loops (PLLs). Based on the delay-dependent PN covariance matrix, we derive the hybrid maximum-likelihood (ML)/maximum a-posteriori (MAP) estimator of the deterministic delay-Doppler parameters and the random PN, resulting in a challenging high-dimensional nonlinear optimization problem. To circumvent the nonlinearity of PN, we then develop an iterated small angle approximation (ISAA) algorithm that progressively refines delay-Doppler-PN estimates via closed-form updates of PN as a function of delay-Doppler at each iteration. 
     % which serves to successively mitigate the impact of residual PN.
     Moreover, unlike existing approaches where PN is considered to be purely an impairment that has to be mitigated, we propose to exploit PN for \rev{resolving} range ambiguity by capitalizing on its delay-dependent statistics (i.e., the range correlation effect), through the formulation of a parametric Toeplitz-block Toeplitz covariance matrix reconstruction problem. Simulation results indicate quick convergence of ISAA to the hybrid Cram\'{e}r-Rao bound (CRB), as well as its remarkable performance gains over state-of-the-art benchmarks, for both FROs and PLLs under various operating conditions, while showing that the detrimental effect of PN can be turned into an advantage for sensing.
     
     %reveal the benefits of the PN exploitation approach and show that ISAA converges quickly to the hybrid Cram\'{e}r-Rao bound (CRB) and outperforms state-of-the-art benchmarks, while providing valuable insights into radar performance with FROs and PLLs.

	\textit{Index Terms--} OFDM, joint radar-communications, phase noise, exploitation, iterated small angle approximation.
	\vspace{-0.1in}
\end{abstract}

% \markboth{IEEE JSTSP, Special Issue on Joint Communication and Radar Sensing, \today}%
% {Shell \MakeLowercase{\textit{et al.}}: Bare Demo of IEEEtran.cls for IEEE Journals}

%%%%%%%%%%%%%%%%%%%%%%%%%%  body  %%%%%%%%%%%%%%%%%%%%%%%%%%
\section{Introduction}\label{sec_intro}
%Moving towards 6G wireless networks by late 2020s,  recently drawn significant attention from both academia and standardization bodies.  

%1st paragraph: general introduction to JRC/ISAC systems, benefits in beyond 5G/6G networks, mention OFDM in the last sentence as a commonly employed JRC waveform

%In general, the time-domain and frequency-domain approaches are shown to be equivalent \cite{JointPN_CE_OFDM_TSP_2017}.

%2nd paragraph: hardware impairments in OFDM, PN most severe one, how PN is handled in OFDM communications, how PN is mitigated in general JRC systems

Envisioned as one of the key enabling technologies in 6G wireless networks, the concept of joint radar-communications (JRC), also known as integrated sensing and communications (ISAC), has drawn significant attention in recent years \cite{SPM_JRC_2019,jointRadCom_review_TCOM,JRC_6G_Nokia_2021,Eldar_SPM_JRC_2020,JCAS_Survey_2022,Fan_ISAC_6G_JSAC_2022}. %with a dedicated initiative established at IEEE. 
Towards practical JRC implementation in real-world scenarios, two principal design paradigms have been pursued in the literature: radar-communications coexistence (RCC) \cite{SPM_Zheng_2019} and dual-functional radar-communications (DFRC) \cite{DFRC_SPM_2019,DFRC_Waveform_Design}. RCC deals with joint optimization of spectrally coexisting radar and communication systems deployed on separate platforms to mitigate mutual interference and enable efficient spectrum utilization \cite{RCC_Interf_TSP_2021}, while DFRC refers to integration of sensing and communication functionalities into a single hardware that employs a joint waveform to perform both tasks simultaneously \rev{for mutual benefit} \cite{Eldar_SPM_JRC_2020}. Sparking innovative use cases (e.g., sensing-assisted communications \rev{\cite{Fan_ISAC_6G_JSAC_2022}}) and bringing hardware/cost efficiency by piggybacking on the existing wireless infrastructure, DFRC holds great potential in \rev{emerging 6G} networks, where sensing will be an integral component \cite{Fan_ISAC_6G_JSAC_2022,6GVision_2020}. As a promising candidate for DFRC transmission, the orthogonal
frequency-division multiplexing (OFDM) waveform has been extensively studied thanks to its wide prevalence in mobile network standards and its satisfactory performance in radar operation \cite{RadCom_Proc_IEEE_2011,General_Multicarrier_Radar_TSP_2016,ICI_OFDM_TSP_2020,OFDM_DFRC_TSP_2021}.

As next-generation wireless systems are expected to operate at high frequencies, i.e., millimeter wave (mmWave) bands \cite{6GVision_2020}, hardware impairments (HWIs), such as phase noise (PN) \cite{PN_OFDM_TSP_2017}, power amplifier nonlinearity (PAN) and mutual coupling (MC), can become a major bottleneck for OFDM DFRC system performance in both radar and communications \cite{hexax_pimrc_2021,RF_JCS_2021}. In particular, the severity of PN, which is caused by oscillator imperfections\rev{\footnote{\label{fn_pn_osc}\rev{Being a time-varying impairment, PN constitutes a much more serious issue for DFRC systems than static impairments caused by oscillator non-idealities, such as carrier frequency offset (CFO) and I/Q imbalance \cite{OFDM_PN_HCRB_2014}. For instance, CFO has no effect on monostatic sensing systems as the same oscillator is employed for transmission and reception \cite{MIMO_OFDM_ICI_JSTSP_2021}.}}}, increases with operating frequency \cite{OFDM_PN_HCRB_2014,PN_CohBw_2021_TWC}. Due its rapidly time-varying nature, PN requires dynamic compensation accounting for subcarrier-level (frequency-domain) or sample-level (time-domain) mitigation processing, thereby posing a significant challenge in channel estimation and data detection for OFDM communications \cite{VI_PN_TSP_2007,OFDM_PN_HCRB_2014,PN_OFDM_Relay_TWC_2016,PN_mmWave_OFDM_TWC_2022}. \rev{To tackle the PN compensation problem for OFDM, various frequency-domain \cite{PN_OFDM_PLL_TCOM_2007,PN_CohBw_2021_TWC,PN_ICI_OFDM_COML_2005,OFDM_PN_Wiener_TCOM_2011,PN_OFDM_WLAB_COMML_2002} and time-domain \cite{OFDM_Joint_PHN_TSP_2006,PN_OFDM_Sayed_TSP_2007,PN_OFDM_TSP_2017} estimation approaches have been proposed.}

Despite the vast literature on PN estimation in \textit{OFDM communications}, very little effort has been devoted to studying the impact of PN on the performance of \textit{OFDM radar} (e.g., \cite{OFDM_JRC_PN_JLT_2022,SC_OFDM_PN}), let alone to developing algorithms to estimate and compensate for PN in radar sensing. \rev{In \cite{OFDM_JRC_PN_JLT_2022} and \cite{SC_OFDM_PN}}, the effect of PN on the range-velocity profile of \rev{an OFDM radar} is investigated, which shows that PN leads to an increase in noise floor and produces a ridge along the velocity axis. \rev{To support new use cases towards 6G networks, emerging mmWave sensing applications \cite{multiFunc_6G_CoMmag,IEEE80211_Auto_JSTSP_2021,Multibeam_JRC_CommLetter_2022,mmwave_radcom_JSTSP_2021} impose stringent requirements on range and velocity accuracy \cite{6G_HexaX_Access}, which necessitates the consideration of PN.}
% In \cite{OFDM_JRC_PN_JLT_2022}, an experimental study with an OFDM DFRC system is carried out to demonstrate the impact of PN on the range-velocity profile of the radar and on the error vector magnitude (EVM) at the communications receiver. In \cite{SC_OFDM_PN}, the effect of PN on the range-velocity profile of a stepped-carrier OFDM radar is investigated, which shows that PN leads to an increase in noise floor and produces a ridge along the velocity axis. 
From the perspective of radar receive processing, the existing OFDM radar algorithms (e.g., \cite{RadCom_Proc_IEEE_2011,Firat_OFDM_2012,OFDM_Radar_Phd_2014,OFDM_Radar_Corr_TAES_2020,ICI_OFDM_TSP_2020,MUSIC_OFDM_Radar_TSP_2021}) assume ideal oscillators and thus cannot provide satisfactory performance under the effect of PN, especially for large PN variances. In a nutshell, no systematic study has been performed to tackle the problem of radar sensing in OFDM DFRC systems with oscillator PN, and, accordingly, to derive an algorithm to jointly estimate delay, Doppler and PN.
% To the best of authors' knowledge, the problem of radar sensing with OFDM DFRC systems has not been investigated before.

%Mention \cite{RF_JCS_2021} and \cite{Bliss_PN_2016} as the studies are related to JCS, but not OFDM. . 

As a key distinction between radar and communications, the so-called \textit{range correlation} effect \cite{Range_Correlation_93,Bliss_PN_2016,SPM_PN_2019,Canan_SPM_2020} constitutes an essential peculiarity of PN in monostatic radar sensing compared to PN in a communications setup. Due to independent PN processes at the transmitter and the receiver, the PN statistics in communications systems 
%are typically known in advance and 
do not depend on unknown channel parameters (e.g., \cite{OFDM_Cons_PN_TCOM_2004,PN_BCRB_TSP_2013,OFDM_PN_HCRB_2014,PN_OFDM_TSP_2017}). Conversely, in a shared-oscillator radar transceiver, downconversion of the reflected signal from a target results in a differential (self-referenced/self-correlated) PN process \cite{Demir_PN_2006,PN_2006} corresponding to the difference between the original PN process and its version shifted in time by the round-trip delay of the target. This correlation effect renders the statistics of PN in the radar receiver range-dependent, leading to higher PN variance for farther targets \cite{Bliss_PN_2016}, which brings both challenges and opportunities specific to radar sensing. The main challenge pertains to delay and PN being coupled, making it difficult to disentangle the corresponding estimation tasks (as is commonly performed in joint channel/PN estimation in OFDM communications, e.g., \cite{OFDM_Joint_PHN_TSP_2006,OFDM_Joint_PHN_TVT_2009}). On the other hand, the main opportunity arises from the possibility to exploit the range-dependent PN statistics for enhancing range estimation performance. While \textit{PN exploitation} may offer promising performance gains, to the best of authors' knowledge, this topic remains surprisingly unexplored, both in the standard radar literature (i.e., frequency-modulated continuous wave (FMCW), multiple-input multiple-output (MIMO) or pulsed radars) and in OFDM DFRC research.
% however, this effect has not been exploited before, although it contains useful information that can facilitate range estimation.

%3rd paragraph: PN mitigation in radar systems, examples from pulsed radar and FMCW, no holistic study on PN in OFDM radar sensing, range correlation effect

% 4th paragraph: what is missing in the literature: OFDM sensing with PN not explored, different from OFDM communications, questions that we aim to answer in this paper

%5th paragraph: contributions

%%%%% Alternative 1
% In light of the existing literature, we identify several major gaps concerning the sensing functionality of OFDM DFRC systems in the presence of PN: \textit{(i)} lack of statistical characterization of PN in OFDM radar receiver for different types of oscillators, namely, free-running oscillators (FROs) and phase-locked loops (PLLs) \cite{Demir_PN_2006,PN_OFDM_PLL_TCOM_2007}, \textit{(ii)} lack of suitable algorithms for joint estimation of delay, Doppler and PN, which can mitigate the impact of PN on the sensing performance, \textit{(iii)} lack of PN exploitation methods that can take advantage of the range correlation effect.

%%%%% Alternative 2
In light of the existing literature, several fundamental questions arise concerning the sensing functionality of OFDM DFRC systems in the presence of PN:
\begin{itemize}
    %\item How does PN affect delay and Doppler estimation accuracy?
    % \item How do different types of oscillators, namely, free-running oscillators (FROs) and phase-locked loops (PLLs) \cite{Demir_PN_2006,PN_OFDM_PLL_TCOM_2007}, impact the sensing performance? Do they lead to similar levels of degradation in the accuracy of delay and Doppler estimation?
    \item What are the statistical properties of PN in the OFDM radar receiver for different types of oscillators, namely, free-running oscillators (FROs) and phase-locked loops (PLLs) \cite{Demir_PN_2006,PN_OFDM_PLL_TCOM_2007}? 
    \item How can we develop \rev{powerful} algorithms for joint estimation of delay, Doppler and PN, to \textit{mitigate} the impact of PN on the sensing performance? How do FROs and PLLs affect the performance of delay and Doppler estimation?
    %How can we effectively estimate and mitigate PN  Can we develop optimal delay-Doppler estimation and PN estimation/mitigation algorithms in monostatic OFDM sensing by effectively utilizing the delay-dependent PN statistics?
    \item Considering the range correlation effect, is it possible to \textit{exploit} PN to improve ranging performance beyond that achievable via PN-free, ideal oscillators?
\end{itemize}
In an attempt to answer these questions, this paper studies the problem of radar delay-Doppler estimation in OFDM DFRC systems in the face of oscillator PN. We begin by deriving the statistical characteristics of PN in the OFDM radar observations for both FROs and PLLs. Then, we propose a novel algorithm for joint estimation of delay, Doppler and PN, based on iterated small angle approximation of PN in the cost function of the hybrid maximum-likelihood (ML)/MAP estimator, which enables fast convergence to the corresponding theoretical bounds. Furthermore, we develop a PN exploitation approach that can effectively utilize the delay-dependent PN statistics to resolve range ambiguity. The main contributions of this paper can be summarized as follows:
\begin{itemize}
    \item \textbf{Problem Formulation for OFDM Radar Sensing under PN:} For the first time in the literature, we investigate the problem of monostatic radar sensing in OFDM DFRC systems under the effect of oscillator PN. To provide a rigorous problem formulation, we derive an explicit statistical characterization of PN in the OFDM radar receiver. \rev{We consider two commonly used oscillator models, namely, FRO and PLL. The derivations} %reveal the block-diagonal considering two realizations of an oscillator, namely, FRO and PLL, which 
    reveal the block-diagonal structure of the PN covariance matrix for the former and the more general Toeplitz-block Toeplitz structure for the latter.
    
    \item \textbf{Hybrid ML/MAP Estimator via Iterated Small Angle Approximation Algorithm:} We derive the hybrid ML/MAP estimator of the deterministic delay-Doppler parameters and the random PN over an OFDM frame with multiple symbols. \rev{The covariance matrix of PN in the backscattered signal} depends on the unknown delay. % where the PN covariance matrix depends on unknown delay. 
    To deal with the highly nonlinear nature of the resulting cost function, we propose a novel iterated small angle approximation (ISAA) approach that invokes the small PN approximation around the current PN estimate at each iteration, progressively refining delay-Doppler-PN estimates and minimizing the impact of residual PN. The proposed approach enables closed-form update of PN as a function of delay-Doppler and provides significant improvements in PN tracking accuracy through alternating iterations.
    
    \item \textbf{PN Exploitation to Resolve Range Ambiguity:} Relying on the key insight that PN conveys valuable information on delay through its delay-dependent statistics, we develop an algorithm for \rev{resolving} range ambiguity, that exploits the statistics of the PN estimates at the output of the proposed ISAA method. The PN exploitation approach formulates the range estimation as a parametric covariance matrix reconstruction problem by leveraging the Toeplitz-block Toeplitz structure and is capable of yielding unambiguous range estimates through the fact that PN covariance imposes no ambiguity in range (as opposed to the fundamental upper limit dictated by OFDM subcarrier spacing \cite{RadCom_Proc_IEEE_2011,OFDM_Radar_Phd_2014,OFDM_ICI_TVT_2017}).
    
    \item \textbf{Simulation Analysis:} %Extensive Monte Carlo simulations conducted under a broad variety of circumstances with regard to signal-to-noise ratio (SNR), oscillator type (FRO and PLL), oscillator quality ($3 \,\rm{dB}$ bandwidth and PLL loop bandwidth) and target range
    Extensive simulations conducted under a wide variety of operating conditions indicate that the proposed ISAA algorithm converges quickly to the corresponding hybrid Cram\'{e}r-Rao bounds (CRBs) \cite{hybrid_CRB_LSP_2008,Hybrid_ML_MAP_TSP} on delay-Doppler-PN estimation in few iterations and considerably outperforms the benchmark FFT method \cite{RadCom_Proc_IEEE_2011,OFDM_Radar_Phd_2014,OFDM_Radar_Corr_TAES_2020}. The accuracy gains \rev{are} more pronounced for higher SNRs, larger $3 \,\rm{dB}$ oscillator bandwidths, smaller loop bandwidths (for PLLs) and farther targets. In addition, PLLs are found to be more beneficial for Doppler estimation than FROs, due to the presence of slow-time PN correlation for PLLs.
    %in the Toeplitz-block Toeplitz structure. 
    Moreover, above a certain SNR level, the PN exploitation algorithm is shown to correctly identify the true range of range-ambiguous targets and achieve much higher ranging accuracy than the FFT method fed with PN-free observations, thereby turning PN into an advantage for sensing.\footnote{Notations: 
    %Uppercase (lowercase) boldface letters are used to denote matrices (vectors). 
    %$(\cdot)^{*}$, $(\cdot)^{T}$ and $(\cdot)^{H}$ represent conjugate, transpose and Hermitian transpose operators, respectively. 
    %The $\thn{n}$ entry of a vector $\xx$ is denoted as $\left[\xx\right]_n$, while the $\thn{(m,n)}$ element of a matrix $\boldX$ is $\left[ \boldX \right]_{m,n}$. 
    $\projrange{\boldX} = \boldX (\boldX^H \boldX)^{-1} \boldX^H$ represents the orthogonal projector onto the column space of $\boldX$ and $\projnull{\boldX} = \Imatrix - \projrange{\boldX}$. $\odot$ and $\otimes$ denote the Hadamard and Kronecker product, respectively. 
    %$\boldzero$ is an all-zeros vector and $\Imatrix$ denotes an identity matrix of appropriate size. 
    $\diag{\xx}$ outputs a diagonal matrix with the elements of a vector $\xx$ on the diagonals, $\diag{\boldX}$ represents a diagonal matrix with the diagonal elements of a square matrix $\boldX$ on the diagonals, $\vecc{\cdot}$ denotes matrix vectorization operator, and $\veccinv{\cdot}$ reshapes a vector into an $N \times M$ matrix. }
    %$\floorr{\cdot}$ is the floor function. 
    %For square matrices $\boldX$ and $\boldY$, $\boldX \succ \boldY$ means $\boldX - \boldY$ is positive definite.}.
\end{itemize}

%while $\ddiag{ \boldX }$ outputs a diagonal matrix whose diagonal elements correspond to those of $\boldX$.

\section{System Model and Problem Formulation}\label{sec_sys_mod}
Consider an OFDM JRC system consisting of a DFRC transceiver and a communications receiver (RX), \rev{as shown in Fig.~\ref{fig_scenario}}. Equipped with a radar-communications transmitter (TX) \rev{(i.e., a conventional OFDM TX)} and a radar RX on a single hardware platform, the DFRC transceiver sends data symbols to the communications RX and simultaneously performs monostatic radar sensing using the backscattered signals to accomplish various radar tasks (e.g., target detection, estimation, tracking and classification) \cite{RadCom_Proc_IEEE_2011,DFRC_SPM_2019,jointRadCom_review_TCOM}. To enable full-duplex operation without self-interference to the radar RX, we assume sufficient isolation and decoupling of TX/RX antennas at the DFRC transceiver \cite{RadCom_Proc_IEEE_2011,SI_5G_2015,OFDM_Radar_Phd_2014,80211_Radar_TVT_2018,OFDM_FD_LTE_2019,Fan_ISAC_6G_JSAC_2022}. 
%\cite{Interference_MIMO_OFDM_Radar_2018,RadCom_Proc_IEEE_2011,SI_5G_2015,InBandFD,OFDM_Radar_Phd_2014,80211_Radar_TVT_2018,Adaptive_RadCom_2019,OFDM_FD_LTE_2019,Fan_ISAC_6G_JSAC_2022}
\rev{At the communications RX, conventional OFDM receive operations (e.g., channel estimation, frequency synchronization, data detection \cite{PN_mmWave_OFDM_TWC_2022}) are performed ordinarily without any constraints from the sensing functionality.} Moreover, the oscillator of the DFRC transceiver, which is shared between the TX and radar RX on the co-designed joint hardware platform, is assumed to be non-ideal and impaired by PN due to imperfections \cite{Demir_PN_2000,OFDM_Cons_PN_TCOM_2004,PN_OFDM_PLL_TCOM_2007,OFDM_PN_HCRB_2014,PN_CohBw_2021_TWC}. In this section, we derive OFDM transmit and radar receive signal models in the presence of PN, %provide a statistical characterization of PN in the radar receiver,
and formulate the resulting OFDM radar sensing problem. \rev{We note that the paper will focus on radar sensing under PN while the communications RX is assumed to compensate for PN via well-established approaches, e.g., \cite{PN_OFDM_Sayed_TSP_2007,VI_PN_TSP_2007,PN_Spectral_ICI_TSP_2010,PN_OFDM_TSP_2017}.}

\begin{figure}
    \centering
    \includegraphics[width=1\columnwidth]{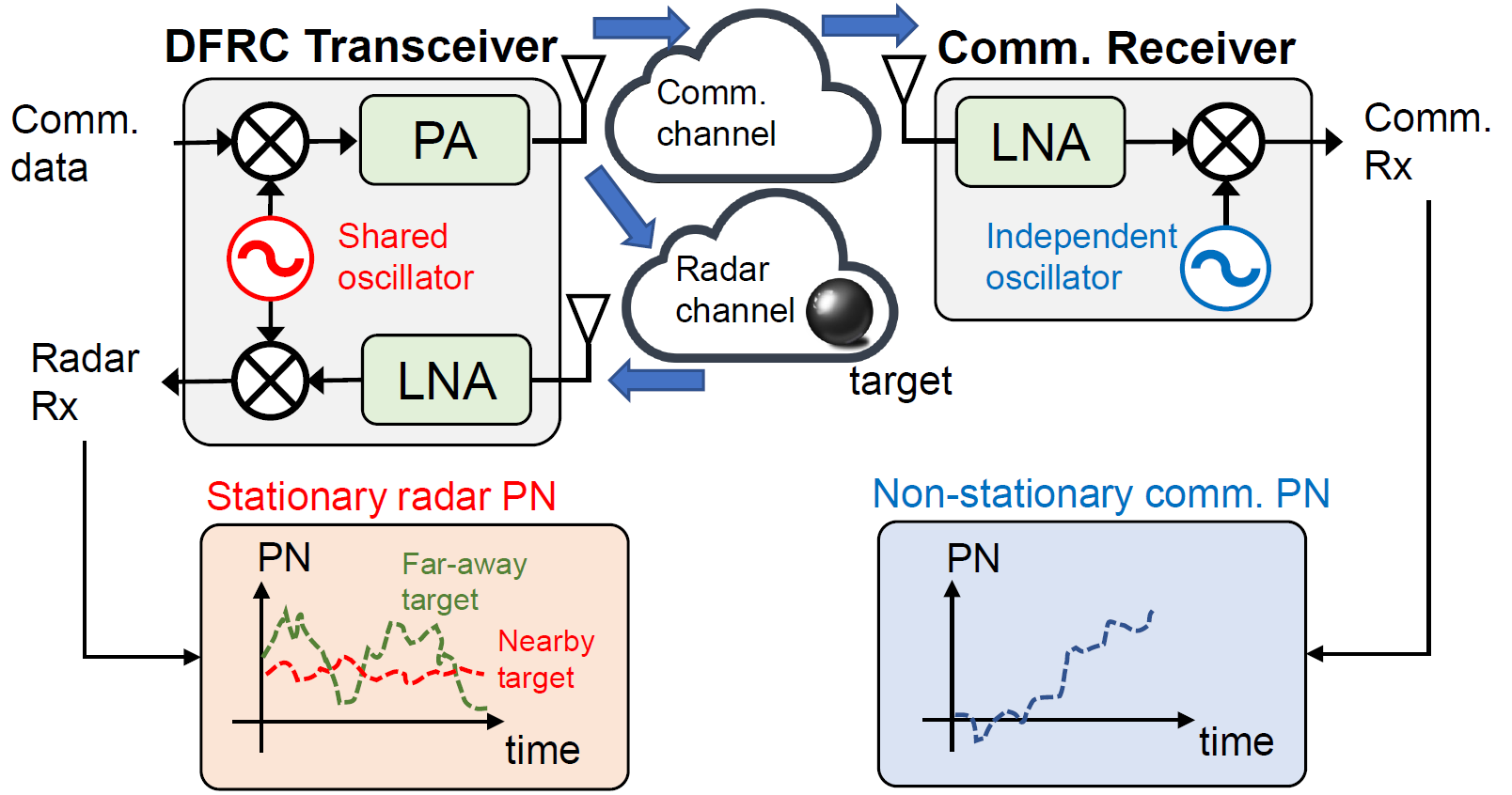}
    \caption{\rev{OFDM JRC system comprising a DFRC transceiver that sends data symbols to a communications RX and simultaneously performs monostatic sensing using the backscattered signals under the impact of PN. Downconversion of the radar return signal using the shared oscillator results in a self-correlated PN process, leading to delay-dependent PN statistics in radar sensing, while PN statistics in the communications RX, which employs an independent oscillator, have no relation to channel parameters.}}
    \label{fig_scenario}
    \vspace{-0.2in}
\end{figure}

\vspace{-0.1in}
\subsection{Transmit Signal Model}\label{sec_transmit}
We consider an OFDM communication frame with $M$ symbols and $N$ subcarriers. The total duration of a symbol is given by $\Tsym = \Tcp + T$, where $\Tcp$ and $T$ denote, respectively, the cyclic prefix (CP) and the elementary symbol durations \cite{RadCom_Proc_IEEE_2011}. In addition, $\deltaf = 1/T$ is the subcarrier spacing, leading to a total bandwidth of $N \deltaf = B$. The complex baseband OFDM transmit signal can be expressed as \cite{General_Multicarrier_Radar_TSP_2016}
\begin{align}
    s(t) = \sum_{m=0}^{M-1} \sm(t)~,
\end{align}
where
\begin{equation}\label{eq_ofdm_baseband}
\sm(t) = \frac{1}{\sqrt{N}} \sum_{n = 0}^{N-1}  \xnm \, e^{j 2 \pi n \deltaf t} \rect{\frac{t - m\Tsym}{\Tsym}} 
\end{equation} 
is the OFDM signal for the $\thn{m}$ symbol, $\xnm$ denotes the complex data symbol on the $\thn{n}$ subcarrier for the $\thn{m}$ symbol, and $\rect{t}$ is a rectangular pulse that takes the value $1$ for $t \in \left[0, 1 \right]$ and $0$ otherwise. In the presence of PN in the oscillator, the upconverted transmit signal over the block of $M$ symbols for $t \in \left[0, M \Tsym \right]$ can be written as \cite{PN_2006}
\begin{equation}\label{eq_passband_st}
\stilde(t) = \Re \left\{  s(t) e^{j \left[ 2 \pi \fc t + \phi(t)\right]} \right\} ~,
\end{equation}
where $\fc$ is the carrier frequency and $\phi(t)$ denotes the PN process in the oscillator.

\subsection{Receive Signal Model}\label{sec_radar_rec}
In radar sensing, we assume the existence of a point target in the far-field, with round-trip delay $\tau = 2 R/c$, normalized Doppler shift $\nu = 2 v/c$ and complex channel gain $\alpha$, which includes path loss and radar cross section effects. Here, $R$, $v$ and $c$ denote the distance, radial velocity and speed of propagation, respectively. Given the transmit signal model in \eqref{eq_passband_st}, the passband backscattered signal at the radar RX can be given as
\begin{align}\label{eq_rec_passband}
    \ytilde(t) = \Re \left\{ \alpha \, s(t-\tau(t)) e^{j \left[ 2 \pi \fc (t-\tau(t)) + \phi(t-\tau(t))\right]} \right\} ~,
\end{align}
where $\tau(t) = \tau - \nu t$ is the time-varying delay due to Doppler shift. After downconverting the passband signal in \eqref{eq_rec_passband} through the noisy oscillator, which corresponds to multiplication by $e^{-j (2 \pi \fc t + \phi(t))}$ \cite{PN_SI_TSP_2017}, the equivalent complex baseband signal can be written as \cite{SPM_PN_2019}
\begin{align}\nonumber
    y(t) &= \alpha \, s(t-\tau(t)) e^{j \left[ 2 \pi \fc (t-\tau(t)) + \phi(t-\tau(t)) \right]} e^{-j \left[ 2 \pi \fc t + \phi(t)\right]} \\ \label{eq_rec_baseband}
    &= \alpha \, s(t-\tau(t)) e^{-j 2 \pi \fc \tau} e^{j 2 \pi \fc \nu t }  e^{j \left[ \phi(t-\tau(t)) - \phi(t)\right] } ~.
\end{align}

Focusing primarily on vehicular JRC scenarios, we assume that the Doppler shifts satisfy $ \lvert \nu \rvert \ll 1/N$ \cite{Firat_OFDM_2012,ICI_OFDM_TSP_2020,MIMO_OFDM_ICI_JSTSP_2021}, where for typical vehicular OFDM JRC systems, \rev{$\nu \ll 10^{-6}$} (corresponding to \rev{$v \ll 540 \, \rm{km/h}$}), while $N$ is on the order of $10^3$.
%make the following widely used assumptions in the OFDM radar literature \cite{MIMO_OFDM_ICI_JSTSP_2021}.
% \begin{itemize}
%     \item $\Tcp \geq \tau$, i.e., the CP duration is larger than the maximum round-trip target delay \cite{Firat_OFDM_2012,OFDM_Radar_Phd_2014,SPM_JRC_2019}.
 This allows us to approximate the PN term in \eqref{eq_rec_baseband} as $\phi(t-\tau(t)) \approx \phi(t-\tau)$. In addition, the time-bandwidth product $B M \Tsym$ is small enough to justify (together with $ \lvert \nu \rvert \ll 1/N$) the narrowband approximation $s(t - \tau(t)) \approx s(t - \tau)$ \cite{OFDM_ICI_TVT_2017}.
%\end{itemize}
Under this setting, the received signal in \eqref{eq_rec_baseband} becomes
\begin{align} \label{eq_rec_baseband2}
    y(t) &= \alpha \, s(t-\tau) e^{-j 2 \pi \fc \tau} e^{j 2 \pi \fc \nu t }  w(t, \tau) ~,
\end{align}
where the multiplicative PN process is represented by  
\begin{align} \label{eq_pn_def}
    w(t, \tau) \triangleq e^{j \left[ \phi(t-\tau) - \phi(t)\right] } ~.
\end{align}
The statistical properties of the PN process $\phi(t-\tau) - \phi(t)$ in \eqref{eq_pn_def} will be derived in Sec.~\ref{sec_pn_statistics}.

%In Appendix~\ref{app_pn_stat}, we provide a statistical characterization of the PN related process $w(t, \tau)$.

\subsection{Fast-Time/Slow-Time Representation with Phase Noise}\label{sec_fast_slow}
For the $\thn{m}$ symbol, we remove the CP and sample $y(t)$ in \eqref{eq_rec_baseband2} at $t = m\Tsym + \Tcp + \ell T / N$ for $\ell = 0, \ldots, N-1$. Making the standard OFDM radar assumptions $\Tcp \geq \tau$ \cite{Firat_OFDM_2012,OFDM_Radar_Phd_2014,SPM_JRC_2019} (CP duration is set to be longer than the round-trip delay of the furthermost target) and $\fc T \nu \ll 1$ \cite{Passive_OFDM_2010,OFDM_Passive_Res_2017_TSP,OTFS_RadCom_TWC_2020,OFDM_DFRC_TSP_2021} (Doppler shift $\fc \nu$ is small compared to subcarrier spacing $\deltaf$), and ignoring constant phase terms, the received signal for the $\thn{m}$ symbol can be written as \cite{ICI_OFDM_TSP_2020,MIMO_OFDM_ICI_JSTSP_2021}
\begin{align}\label{eq_rec_bb2}
    y_{\ell,m} &= \alpha  \, e^{j 2 \pi \fc m \Tsym \nu  }  w_{\ell,m}(\tau) \\ \nonumber &~~\times \frac{1}{\sqrt{N}}  \sum_{n = 0}^{N-1}  \xnm \, e^{j 2 \pi n \frac{\ell}{N}} e^{-j 2 \pi n \deltaf \tau} ~,
\end{align}
where $w_{\ell,m}(\tau)$ is the sampled version of the PN term $w(t, \tau)$ in \eqref{eq_rec_baseband2} for $t = m\Tsym + \Tcp + \ell T / N$. Let
\begin{align} \label{eq_steer_delay}
	\bb(\tau) & \triangleq  \transpose{ \left[ 1, e^{-j 2 \pi \deltaf \tau}, \ldots,  e^{-j 2 \pi (N-1) \deltaf  \tau} \right] } ~, \\ \label{eq_steer_doppler}
	\cc(\nu) & \triangleq \transpose{ \left[ 1, e^{-j 2 \pi f_c \Tsym \nu }, \ldots,  e^{-j 2 \pi f_c (M-1) \Tsym \nu } \right] } ~, 
\end{align}
represent the frequency-domain and temporal (slow-time) steering vectors, respectively. 

Aggregating the observations in \eqref{eq_rec_bb2} over fast-time $\ell$ and slow-time $m$, and taking into account the presence of additive sensor noise, the fast-time/slow-time observation matrix in the presence of PN is obtained as \cite{MIMO_OFDM_ICI_JSTSP_2021}
\begin{align} \label{eq_ym_all_multi}
    \boldY = \alpha \, \boldW \odot \FF_N^{H} \Big(\boldX \odot \bb(\tau) \cc^{H}(\nu) \Big)  + \boldZ ~,
\end{align}
where $\boldW \in \complexset{N}{M}$ with $\left[ \boldW \right]_{\ell,m} \triangleq w_{\ell,m}(\tau)$ is the \textit{multiplicative PN matrix}\footnote{For notational convenience, we drop the dependence of $\boldW$ on $\tau$.} consisting of fast-time/slow-time samples from the PN process in \eqref{eq_pn_def}, $\FF_N \in \complexset{N}{N}$ is the unitary DFT matrix with $\left[ \FF_N \right]_{\ell,n} = \frac{1}{\sqrt{N}} e^{- j 2 \pi n \frac{\ell}{N}} $, $\boldX \in \complexset{N}{M}$ contains the complex data symbols with $\left[ \boldX \right]_{n,m} \triangleq \xnm$, $\boldY \in \complexset{N}{M}$ with $\left[ \boldY \right]_{\ell,m} \triangleq y_{\ell,m} $, and $\boldZ \in \complexset{N}{M}$ is additive white Gaussian noise (AWGN) with $\vecc{\boldZ} \sim \mtCN(\boldzero, \allowbreak 2\sigma^2 \Imatrix ) $. As observed from \eqref{eq_ym_all_multi}, the PN component $\boldW$ introduces intercarrier interference (ICI) in OFDM radar \cite{OFDM_ICI_TVT_2017,ICI_OFDM_TSP_2020,MIMO_OFDM_ICI_JSTSP_2021} (similar to its effect in OFDM communications \cite{PN_OFDM_PLL_TCOM_2007,OFDM_PN_HCRB_2014,PN_OFDM_TSP_2017,PN_CohBw_2021_TWC}) and might severely degrade the performance of delay-Doppler estimation.

\vspace{-0.1in}
\subsection{Special Case: Ideal Oscillator}\label{sec_ideal}
To relate the derived signal model in \eqref{eq_ym_all_multi} to the commonly used ones in the literature, we investigate the special case of an ideal oscillator where the PN process is not present, i.e., $\phi(t) = 0 \, , \forall t$, which yields $w(t,\tau) = 1\, , \forall t, \tau$, and $\boldW$ becomes an all-ones matrix, i.e., $\boldW = \boldone_{N \times M}$. In this case, \eqref{eq_ym_all_multi} reverts to the PN-free model
\begin{align} \label{eq_ym_special}
    \boldYfr = \alpha \, \FF_N^{H} \Big(\boldX \odot \bb(\tau) \cc^{H}(\nu) \Big)   + \boldZ ~.
\end{align}
Following the traditional processing chain for OFDM radar receivers \cite{OFDM_Radar_Corr_TAES_2020,RadCom_Proc_IEEE_2011,OFDM_Radar_Phd_2014}, we take the DFT of the columns of $\boldYfr$ in \eqref{eq_ym_special} to switch from fast-time/slow-time to frequency/slow-time domain and obtain the standard OFDM radar observations \cite{OFDM_DFRC_TSP_2021,OFDM_Passive_Res_2017_TSP,Passive_OFDM_2010,RadCom_Proc_IEEE_2011,OFDM_Radar_Corr_TAES_2020,OFDM_Radar_Phd_2014}: 
\begin{align} \label{eq_ym_special2}
    \boldYtilde = \FF_N \boldYfr = \alpha \,   \boldX \odot    \bb(\tau) \cc^{H}(\nu)   + \FF_N \boldZ ~,
\end{align}
where $\vecc{\FF_N \boldZ} \sim \mtCN(\boldzero, \allowbreak 2\sigma^2 \Imatrix ) $. Clearly, \eqref{eq_ym_special2} does not involve any ICI effect, and is therefore amenable to conventional delay-Doppler estimation algorithms (after removing the effect of $\boldX$), such as 2-D DFT over time and frequency domains \cite{RadCom_Proc_IEEE_2011,OFDM_Radar_Phd_2014,OFDM_Radar_Corr_TAES_2020} and super-resolution methods \cite{Passive_OFDM_2010,OFDM_Passive_Res_2017_TSP}.

% perform matched filtering with the transmit symbols $\boldX$ and take 2-D DFT to obtain the delay-Doppler spectrum:
% \begin{align}\label{eq_mf}
%     \chi(\tauhat, \nuhat) = \bb^H(\tauhat) \big(\FF_N \boldY \odot \boldX^{*}\big) \cc(\nuhat)  ~,
% \end{align}
% where $\bb^H(\tauhat)$ and $\cc(\nuhat)$ correspond to the rows of the IDFT matrix and the columns of the DFT matrix, respectively, for a uniform delay-Doppler grid $(\tauhat, \nuhat)$ sampled at integer multiples of delay-Doppler resolutions.

\subsection{Problem Statement for OFDM Radar under Phase Noise}\label{sec_formulation}
Given the transmit data symbols\rev{\footnote{\label{fn_coloc}\rev{Being co-located on a shared platform with the JRC transmitter, the radar receiver has the knowledge of transmit data symbols $\boldX$ \cite{RadCom_Proc_IEEE_2011,ICI_OFDM_TSP_2020,MIMO_OFDM_ICI_JSTSP_2021}.}}} $\boldX$ and the fast-time/slow-time observations $\boldY$ in \eqref{eq_ym_all_multi}, the problem of interest for OFDM radar sensing in the presence of PN is to estimate the target parameters $\alpha$, $\tau$ and $\nu$, which inherently involves estimating the PN matrix $\boldW$ and compensating for its effect on $\boldY$. To tackle this problem, we first derive the statistical properties of the PN process in Sec.~\ref{sec_pn_statistics}, which will then be utilized in Sec.~\ref{sec_alg_est} to propose a novel algorithm for joint estimation of delay, Doppler and PN. In Sec.~\ref{sec_extension}, we take a step further by exploiting PN as something beneficial for radar sensing.

%%%%%%%%%%%%%%%%%%%%%%%%%%%%%%%%%%%%%%%%%%%%%%%%%%%%%%%%
%%%%%%%%%%%%%%%%%%%%%%%%%%%%%%%%%%%%%%%%%%%%%%%%%%%%%%%%
\section{Phase Noise Statistics}\label{sec_pn_statistics}
This section provides a statistical characterization of the PN process $\phi(t-\tau) - \phi(t)$ in \eqref{eq_pn_def}, gives expressions of the PN variance for different types of oscillators and derives the structure of the PN covariance matrix.

\subsection{Statistics of Differential Phase Noise Process}\label{sec_stat_DPN}
Let $\phi(t)$ be a zero-mean Gaussian random process with variance $\sigmaphi(t)$ \cite{PN_2006,PN_OFDM_Sayed_TSP_2007}, i.e.,
\begin{equation}\label{eq_pn_stat}
\phi(t) \sim \mtN(0, \sigmaphi(t))~,
\end{equation}
where the form of $\sigmaphi(t)$ depends on the type of oscillator. Let us define the \textit{differential} PN (DPN) process \cite{DPN_93} (also called \textit{self-referenced} PN, \textit{increment} PN process \cite{Demir_PN_2006}, or PN \textit{variation} \cite{PN_2006}) as
%with variance $\sigmaphit$ \cite{PN_2006}, i.e.,
%\begin{equation}\label{eq_phiego}
%\phi(t) \sim \mtN(0, \sigmaphit)~.
%\end{equation}
\begin{equation}\label{eq_dpn}
\xi(t, \tau) \triangleq \phi(t) - \phi(t-\tau)~.
\end{equation} 
Since the DPN process\footnote{While $\xi(t, \tau)$ is called the DPN here to distinguish it from the actual PN process $\phi(t)$, we will mostly refer to $\xi(t, \tau)$ as PN in the remainder of the text for ease of exposition.} is stationary, its statistics depend only on the increment value (target delay) $\tau$ \cite[Sec.~IV]{Demir_PN_2006}. Hence, the DPN process can be statistically characterized as \cite{PN_2006}
\begin{equation}\label{eq_dpn_stat}
\xi(t, \tau) \sim \mtN(0, \sigmaxi(\tau))~,
\end{equation}
where $\sigmaxi(\tau)$ is the delay-dependent variance of $\xi(t, \tau) $. The following lemma provides the second-order statistics of $\xi(t, \tau) $.
\begin{lemma}\label{lemma_corr_dpn}
    The correlation function of the DPN process $\xi(t, \tau) $ in \eqref{eq_dpn} is given by
    \begin{align}\label{eq_exp_xi2}
     \rxii(\deltat, \tau) &= \rxii(t_1, t_2, \tau) \triangleq \E \left[  \xi(t_1, \tau)  \xi(t_2, \tau) \right] \\ \nonumber &= \frac{ \sigmaxi(\tau+\deltat) + \sigmaxi(\tau-\deltat) }{2} - \sigmaxi(\deltat) ~,
\end{align}
where $\deltat \triangleq t_1 - t_2$ is the time difference.
\end{lemma}
\begin{proof}
    Please see Appendix~\ref{app_corr_dpn}.
\end{proof}
\vspace{-0.25in}

%%%%%%%%%%%%%%%%%%%%%%%%%%%%%%%%%%%%%%%%%%%%%%%%%%%%%%%%
\subsection{PN Variance for Different Oscillator Types}
Lemma~\ref{lemma_corr_dpn} allows us to compute the correlation of $\xi(t, \tau) $ in terms of its delay-dependent variance function $\sigmaxi(\tau)$. We now provide expressions for $\sigmaxi(\tau)$ by distinguishing between two realizations of an oscillator, namely, free-running oscillator (FRO) and phase-locked loop (PLL) synthesizer \cite{PN_OFDM_PLL_TCOM_2007,Demir_PN_2006}.

\subsubsection{Free-Running Oscillators (FROs)}
For FROs, the variance of the DPN $\xi(t, \tau)$ in \eqref{eq_dpn_stat} is given by \cite[Sec.~V]{Demir_PN_2006}, \cite[Sec.~V,\,VI]{PN_2006}, \cite[Sec.~III-A]{PN_OFDM_PLL_TCOM_2007}
\begin{align} \label{eq_sigmatauwn}
    \sigmaxi(\tau) &= 4 \pi \fdb \abs{\tau} ~,
\end{align}
where $\fdb$ is the $3 \, \rm{dB}$ bandwidth of the Lorentzian oscillator spectrum.

\subsubsection{Phase-Locked Loop (PLL) Synthesizers}
For PLL architectures, the delay-dependent variance in \eqref{eq_dpn_stat} can be expressed as \cite[Sec.~VII-A]{Demir_PN_2006}, \cite[Sec.~III]{PN_2006}
\begin{align}\label{eq_sigmatauwn_pll}
    \sigmaxi(\tau) &= \frac{2\fdb}{\floop} \left(1 - e^{-2 \pi \floop \abs{\tau}} \right) ~, 
\end{align}
where $\floop$ denotes the loop bandwidth of PLL.

We note that PLL degenerates to FRO with decreasing $\floop$, i.e., \eqref{eq_sigmatauwn_pll} converges to \eqref{eq_sigmatauwn} as $\floop \to 0$. With the expressions \eqref{eq_sigmatauwn} and \eqref{eq_sigmatauwn_pll}, a complete statistical characterization of $\xi(t, \tau)$ can be obtained via \eqref{eq_dpn_stat} and \eqref{eq_exp_xi2} for FROs and PLLs, which, in turn, yields the covariance matrix of the fast-time/slow-time PN samples in $\boldW$ in \eqref{eq_ym_all_multi}. As an example, Fig.~\ref{fig_PN_covariance_PLL_FRO} plots the covariance of $\xi(t, \tau)$ at different target delays $\tau$ for both PLLs and FROs. The figure illustrates the \textit{delay-dependency of the PN statistics} in OFDM radar sensing (as opposed to OFDM communications, e.g., \cite{OFDM_Cons_PN_TCOM_2004,PN_BCRB_TSP_2013,OFDM_PN_HCRB_2014}) and provides insights into the effect of PLL based control on the correlation behavior of PN.

%%%%%%%%%%%%%%%%%%%%%%%%%%%%%%%%%%%%%%%%%%
% PN covariance for PLL and FRO
\begin{figure}
	\centering
    \vspace{-0.1in}
	\includegraphics[width=0.95\linewidth]{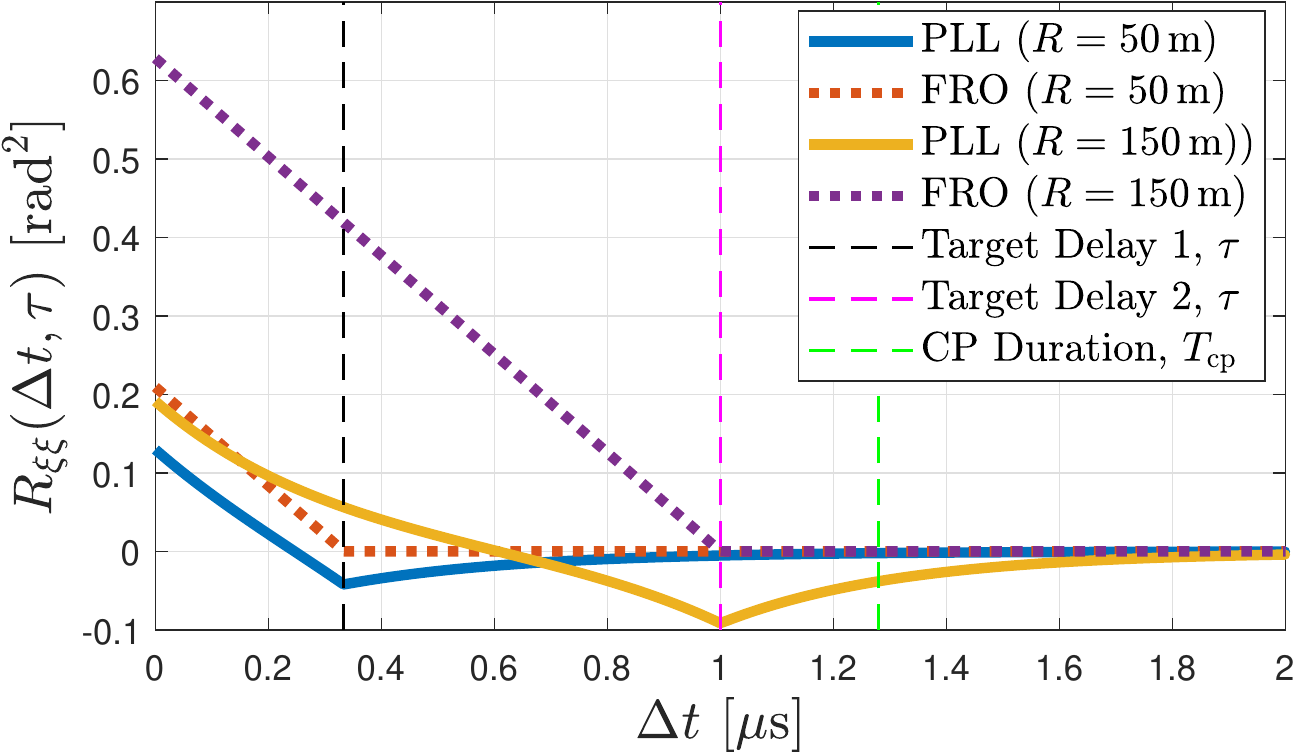}
	\vspace{-0.1in}
	\caption{Covariance of the DPN process $\xi(t, \tau)$ in \eqref{eq_dpn}, computed via the expression in \eqref{eq_exp_xi2}, for PLL and FRO architectures at two different target ranges $R = 50 \, \rm{m}$ and $R = 150 \, \rm{m}$, where the $3 \, \rm{dB}$ bandwidth is $\fdb = 50 \, \rm{kHz}$ and the PLL loop bandwidth is $\floop = 500 \, \rm{kHz}$. \rev{Contrary to OFDM communications, the PN statistics in OFDM radar depend on target delay, as also illustrated in Fig.~\ref{fig_scenario}}.}
	\label{fig_PN_covariance_PLL_FRO}
	\vspace{-0.1in}
\end{figure}
%%%%%%%%%%%%%%%%%%%%%%%%%%%%%%%%%%%%%%%%%%

%%%%%%%%%%%%%%%%%%%%%%%%%%%%%%%%%%%%%%%%%%%%%%%%%%%%%%%%
\subsection{Delay-Dependent PN Covariance Matrix}\label{sec_delay_PN_mat}
Let $\ww \triangleq \vecc{\boldW} \in \complexset{NM}{1}$ and $\bxi \in \realset{NM}{1}$ be the sampled version of the PN process in \eqref{eq_dpn} over the entire OFDM frame, i.e., $\ww = e^{-j \bxi}$ from \eqref{eq_pn_def}. Then, $\bxi$ can be statistically characterized as\rev{\footnote{\label{fn_mult_ext}\rev{The derivation of PN statistics can be straightforwardly extended to the multi-target case by computing cross-correlation of PN vectors associated with different targets at different delays, following the arguments in Lemma~\ref{lemma_corr_dpn}.}}}
\begin{align}\label{eq_bxi_stat}
    \bxi \sim \mtN(\boldzero, \boldR(\tau))~,
\end{align}
where $\boldR(\tau) \in \realset{NM}{NM}$ is the delay-dependent positive definite covariance matrix of $\bxi$. Using \eqref{eq_dpn_stat} and \eqref{eq_exp_xi2}, the $\thn{(i_1,i_2)}$ entry of $\boldR(\tau)$ can be written as
\begin{align}\label{eq_rtau_entries}
    \left[ \boldR(\tau) \right]_{i_1,i_2} = \rxii(\deltat_{i_1 i_2}, \tau) ~,
\end{align}
where
\begin{align}\label{eq_deltat}
    \deltat_{i_1 i_2} &\triangleq (i_1-i_2) \Ts +(m_1-m_2)\Tcp
\end{align}
for $(i_1, i_2) = (n_1 + m_1 N, n_2 + m_2 N)$, with $0 \leq n_1, n_2 \leq N-1$ and $0 \leq m_1, m_2 \leq M-1$ denoting the fast-time and slow-time sample indices, respectively, and $\Ts = T/N$ the sampling interval. We note from \eqref{eq_deltat} that due to CP removal, $\deltat_{i_1 i_2}$ depends not only on $i_1-i_2$, but also on the difference between symbol (slow-time) indices for PN samples belonging to different symbols.

From \eqref{eq_rtau_entries} and \eqref{eq_deltat}, it is straightforward to see that $\boldR(\tau)$ is a symmetric \textit{Toeplitz-block Toeplitz} matrix \cite{blockToeplitzInv_83} consisting of $M \times M$ blocks of size $N \times N$:
\begin{align}\label{eq_toep_block}
    \boldR(\tau) = \begin{bmatrix} \boldR_0(\tau) & \boldR_1(\tau) & \ldots  & \boldR_{M-1}(\tau) \\
    \boldR_1^T(\tau) & \ddots & \ddots &   \\
    & \ddots & \ddots & \boldR_1(\tau) \\
    \boldR_{M-1}^T(\tau) & &  \boldR_1^T(\tau) & \boldR_0(\tau) 
    \end{bmatrix} ~,
\end{align}
where the $\thn{m}$ Toeplitz block $\boldR_m(\tau) \in \realset{N}{N}$ is given by
\begin{align} \label{eq_ind_blocks}
    [\boldR_m(\tau)]_{n_1,n_2} = \rxii(\deltat^{(m)}_{n_1 n_2}, \tau)
\end{align}
with $\deltat^{(m)}_{n_1 n_2} \triangleq (n_1-n_2)\Ts -m \Tsym$. While PLLs exhibit the generic Toeplitz-block Toeplitz structure in \eqref{eq_toep_block}, FROs lead to a more special covariance structure, as pointed out in the following lemma.

\begin{lemma}\label{lemma_fro_blkdiag}
    The PN covariance matrix $\boldR(\tau)$ in \eqref{eq_toep_block} is block-diagonal for FROs with $\tau \leq \Tcp$, i.e.,
    \begin{align} \label{eq_R_blkdiag}
    \boldR(\tau) = \blkdiagg{\boldR_0(\tau),\ldots,\boldR_0(\tau)} ~.
\end{align}
\end{lemma}
\begin{proof}
    Please see Appendix~\ref{app_blk_diag}.
\end{proof}

An intuitive interpretation of Lemma~\ref{lemma_fro_blkdiag} can be provided as follows. Since PN samples in different symbols are separated in time by at least $\Tcp$, the PN process in \eqref{eq_dpn} with $\tau \leq \Tcp$ becomes uncorrelated from one symbol to another. In other words, the time intervals for PN accumulation \cite{Demir_PN_2006} corresponding to $\xi(t, \tau)$ and $\xi(t+\Tcp, \tau)$ are non-overlapping, leading to uncorrelated PN in the absence of a control loop. This result is also corroborated by Fig.~\ref{fig_PN_covariance_PLL_FRO}, where the correlation in case of FRO is zero for $\deltat \geq \tau$, while the correlation for PLL can have non-zero values for $\deltat \geq \Tcp$.

% threshold heuristics: derivative of CRB with respect to SNR --> check

%%%%%%%%%%%%%%%%%%%%%%%%%%%%%%%%%%%%%%%%%%%%%%%%%%%%%%%%
%%%%%%%%%%%%%%%%%%%%%%%%%%%%%%%%%%%%%%%%%%%%%%%%%%%%%%%%
\section{Proposed Algorithm for Delay-Doppler Estimation under Phase Noise}\label{sec_alg_est}
%We begin this section by discussing key dissimilarities between OFDM radar and communications in PN estimation. 
In this section, using the statistical characterization of PN derived in \eqref{eq_bxi_stat} and \eqref{eq_toep_block}, we formulate the sensing problem stated in Sec.~\ref{sec_formulation} using a hybrid ML/MAP estimation approach and propose a novel iterated small angle approximation (ISAA) algorithm to jointly estimate delay, Doppler and PN.

\subsection{Hybrid ML/MAP Estimator}\label{sec_hyb_est}
To derive the hybrid ML/MAP estimator of delay, Doppler and PN, we first rewrite the observations in \eqref{eq_ym_all_multi} as
\begin{align}\label{eq_y_mult_doppler}
    \yy = \alpha \, \boldXi \qq(\tau, \nu) + \zz ~,
\end{align}
where $\yy \triangleq \vecc{\boldY} \in \complexset{NM}{1}$, $\zz \triangleq \vecc{\boldZ} \in \complexset{NM}{1}$, 
\begin{align} \label{eq_xi}
\boldXi &\triangleq \diag{e^{-j \bxi}} \in \complexset{NM}{NM} ~,
\\
\label{eq_qtaunu}
    \qq(\tau, \nu) &\triangleq \vecc{\FF_N^{H} \rev{\Big[}\boldX \odot \bb(\tau) \cc^{H}(\nu) \rev{\Big]} } \in \complexset{NM}{1} ~.
\end{align}
Our goal herein is to estimate from \eqref{eq_y_mult_doppler} the unknown parameter vector $\etab = \left[ \tau, \nu, \alpha, \bxi^T  \right]^T$,
%\begin{align}\label{eq_eta_unk}
%\end{align}
consisting of both random ($\bxi$) and deterministic ($\tau, \nu, \alpha$) parameters\rev{\footnote{\label{fn_pn_tv}\rev{Estimation of $\etab$ should be performed for each new OFDM frame due to the time-varying nature of the PN and possible movement of the target in the delay-Doppler plane. Although the PN $\bxi$ has time-invariant statistics specified in \eqref{eq_bxi_stat}, it will have a different (time-varying) realization in each frame and thus needs to be re-estimated by re-executing the proposed method in Sec.~\ref{sec_alg_est} as new observations in the form of \eqref{eq_y_mult_doppler} arrive.}}}. Then, the hybrid ML/MAP estimator of $\etab$ can be written as \cite{Hybrid_ML_MAP_TSP}
\begin{align}\label{eq_hybrid_ml_map}
    (\tauhat, \nuhat, \alphahat, \bxihat  ) = \arg \max_{\tau, \nu, \alpha, \bxi} f_{\yy, \bxi} (\yy, \bxi; \tau, \nu , \alpha) ,
\end{align}
where $f_{\yy, \bxi} (\yy, \bxi; \tau, \nu , \alpha)$ is the joint PDF of $\yy$  and $\bxi$ as a function of the deterministic parameters $\tau$, $\nu$ and $\alpha$ 
\begin{align}\label{eq_pdf_all}
    f_{\yy, \bxi} (\yy, \bxi; \tau, \rev{\nu}, \alpha) = f_{\yy \lvert \bxi} (\yy \lvert \bxi; \tau, \rev{\nu}, \alpha) f_{\bxi}(\bxi; \tau) ~,
\end{align}
$f_{\yy \lvert \bxi} (\yy \lvert \bxi; \tau, \nu, \alpha)$ is the conditional PDF of $\yy$ given $\bxi$, and $f_{\bxi}(\bxi; \tau)$ is the \textit{a-priori} PDF of $\bxi$. It follows from \eqref{eq_bxi_stat} and \eqref{eq_y_mult_doppler} that
\begin{align}
    &f_{\yy \lvert \bxi} (\yy \lvert \bxi; \tau, \nu, \alpha) \\ \nonumber &= \frac{1}{(2\pi \sigma^2)^{NM} } \exp\left\{ - \frac{ \norm{\yy - \alpha \, \boldXi \qq(\tau, \nu) }^2 }{2 \sigma^2} \right\} ~, \\ \label{eq_pdf_xi}
    &f_{\bxi}(\bxi; \tau)  \\ \nonumber &= \frac{1}{\sqrt{(2 \pi)^{NM} \detm{\boldR(\tau)}  } } \exp\left\{ - \frac{ \bxi^T \boldR(\tau)^{-1} \bxi}{2 } \right\} ~.
\end{align}
Plugging \eqref{eq_pdf_all}--\eqref{eq_pdf_xi} into \eqref{eq_hybrid_ml_map} yields 
\begin{align}\label{eq_hybrid_ml_map2}
    ( \tauhat, \nuhat, \alphahat, \bxihat ) &= \arg \min_{\tau, \nu, \alpha, \bxi} \Bigg\{ \frac{ \norm{\yy - \alpha \, \boldXi \qq(\tau, \nu) }^2 }{2 \sigma^2} \\ \nonumber
    & ~~~~~~+ \frac{ \bxi^T \boldR(\tau)^{-1} \bxi + \log \det \boldR(\tau) }{2 }     \Bigg\} ~.
\end{align}
For given $\tau$, $\nu$ and $\bxi$, the optimal estimate of $\alpha$ in \eqref{eq_hybrid_ml_map2} is given by
\begin{align}\label{eq_alpha_hat}
    \alphahat = \frac{\qq^H(\tau, \nu) \boldXi^H \yy  }{ \qq^H(\tau, \nu) \boldXi^H \boldXi \qq(\tau, \nu)  } = \frac{\qq^H(\tau, \nu) \boldXi^H \yy   }{ \norm{\boldX}_F^2  } ~,
\end{align}
where the last equality follows from \eqref{eq_steer_delay}, \eqref{eq_steer_doppler} and \eqref{eq_qtaunu}:
\begin{align}\label{eq_qq_norm}
    \norm{\qq(\tau, \nu) }^2 &= \norm{ \FF_N^{H} \Big(\boldX \odot \bb(\tau) \cc^{H}(\nu) \Big) }_F^2 
    \\ \nonumber
    &= \norm{  \boldX \odot \bb(\tau) \cc^{H}(\nu)  }_F^2
    = \norm{\boldX}_F^2 ~.
\end{align}

Inserting \eqref{eq_alpha_hat} into \eqref{eq_hybrid_ml_map2}, the hybrid ML/MAP problem becomes
\begin{align}\label{eq_hybrid_ml_map_doppler2}
    (\tauhat, \nuhat, \bxihat) &= \arg \min_{ \tau, \nu, \bxi} ~ \frac{ \yy^H \projnull{\boldXi \qq(\tau, \nu)} \yy }{ \sigma^2} +  \bxi^T \boldR(\tau)^{-1} \bxi
    \\ \nonumber &~~+ \log \det \boldR(\tau)     ~,
\end{align}
where
\begin{align}\label{eq_projnull2}
    \projnull{\boldXi \qq(\tau, \nu)} &= \Imatrix - \frac{ \boldXi \qq(\tau, \nu) \qq^H(\tau, \nu)  \boldXi^H  }{ \norm{\qq(\tau, \nu) }^2} 
    \\ \nonumber
    &= \boldXi \boldXi^H -\frac{ \boldXi \qq(\tau, \nu) \qq^H(\tau, \nu)  \boldXi^H  }{\norm{\qq(\tau, \nu) }^2}  
    \\ \nonumber
    &= \boldXi \left( \Imatrix - \frac{ \qq(\tau, \nu) \qq^H(\tau, \nu)   }{ \norm{\qq(\tau, \nu) }^2}   \right) \boldXi^H
    \\ \nonumber
    &= \boldXi \projnull{ \qq(\tau, \nu) }  \boldXi^H ~,
\end{align}
with $\boldXi \boldXi^H = \Imatrix$ resulting from \eqref{eq_xi}. Substituting \eqref{eq_projnull2} into \eqref{eq_hybrid_ml_map_doppler2}, we obtain
\begin{align}\label{eq_hybrid_ml_map_doppler3}
    (\tauhat, \nuhat, \bxihat )  = \arg \min_{ \tau, \nu, \bxi} ~ \llr(\tau, \nu, \bxi) ~,
\end{align}
where 
\begin{align}\nonumber
    \llr(\tau, \nu, \bxi) &= \frac{ \yy^H \boldXi \projnull{ \qq(\tau, \nu) }  \boldXi^H  \yy }{ \sigma^2} +  \bxi^T \boldR(\tau)^{-1} \bxi \\ \label{eq_hybrid_ml_map_doppler4}  &~~+ \log \det \boldR(\tau)     ~.
\end{align}

\subsection{Iterated Small Angle Approximation}
The hybrid ML/MAP optimization problem derived in \eqref{eq_hybrid_ml_map_doppler3} seems quite challenging to solve, primarily due to highly non-linear behavior of the high-dimensional PN vector $\bxi$ in the objective \eqref{eq_hybrid_ml_map_doppler4} (see \eqref{eq_xi}), leading to many local optima. A possible remedy to overcome such non-linearity is to employ \textit{small angle approximation} (SAA) $e^{j \theta} \approx 1 + j \theta$ \cite{OFDM_Joint_PHN_TSP_2006} for small $\theta$. However, while this approach can work well for PN estimation in communications, it may lead to large errors in our sensing problem of interest. This is especially the case for distant targets since the PN variance increases with target delay, as seen from \eqref{eq_sigmatauwn} and \eqref{eq_sigmatauwn_pll}, invalidating the assumption of small $\theta$.

To circumvent the non-linearity of PN in \eqref{eq_hybrid_ml_map_doppler4} and deal with large PN variances, we propose an \textit{iterated small angle approximation} (ISAA) approach that invokes SAA around a given estimate of PN at each iteration, starting from an all-zeros estimate $\bxihat = \boldzero$. More specifically, suppose we have an estimate of PN vector at the $\thn{i}$ iteration, denoted by $\bxihat^{(i)} \in \realset{NM}{1}$, and we wish to approximate the exponential PN term $e^{-j\bxi}$ in \eqref{eq_hybrid_ml_map_doppler4} around $\bxihat^{(i)}$. 
%This is because $\bxi$ might not be small enough to provide good approximation, but $\bxi-\bxihat$ is residual PN and thus leads to more accurate small angle approximation (SAA).
To this end, we define the \textit{residual PN} 
\begin{align}\label{eq_bxidelta}
    \bxidelta = \bxi - \bxihat^{(i)}
\end{align}
for a given PN estimate $\bxihat^{(i)}$. Re-writing \eqref{eq_hybrid_ml_map_doppler4} as a function of $\bxidelta$, we have 
\begin{align}\nonumber
    &\llr(\tau, \nu, \bxidelta \rev{+ \bxihat^{(i)}}) = \frac{ (\upsb^{\rev{(i)}})^T \diag{\yy}^H \projnull{ \qq(\tau, \nu)}  \diag{\yy}   (\upsb^{\rev{(i)}})^\conj  }{ \sigma^2}    
    \\  \label{eq_nll_original_saa}  &~~+  (\bxidelta + \bxihat^{(i)})^T \boldR(\tau)^{-1} (\bxidelta + \bxihat^{(i)}) + \log \det \boldR(\tau) ~,
\end{align}
where $\upsb^{\rev{(i)}} = e^{-j (\bxidelta + \bxihat^{(i)})} \in \complexset{NM}{1}$. In order to solve \eqref{eq_hybrid_ml_map_doppler3} for $\bxi$ in an iterative fashion, we propose to solve the following minimization problem at the $\thn{i}$ iteration for the residual PN $\bxidelta$, given a PN estimate $\bxihat^{(i)}$ at hand: 
\begin{align}\label{eq_hybrid_ml_map_doppler_res}
    (\tauhat, \nuhat, \bxideltahat )  = \arg \min_{ \tau, \nu, \bxidelta} ~ \llr(\tau, \nu, \bxidelta \rev{+ \bxihat^{(i)}}) ~.
\end{align}

To tackle \eqref{eq_hybrid_ml_map_doppler_res}, we invoke SAA for $\upsb^{\rev{(i)}}$ in \eqref{eq_nll_original_saa} around $\bxidelta$ to obtain
\begin{align}\label{eq_saa_delta}
    \upsb^{\rev{(i)}} = e^{-j \bxihat^{(i)}} \odot e^{-j\bxidelta} \approx e^{-j \bxihat^{(i)}} \odot (\boldone - j \bxidelta) ~.
\end{align}
Plugging \eqref{eq_saa_delta} into the first term in \eqref{eq_nll_original_saa}, we obtain the approximation
\begin{align} \nonumber
    &(\upsb^{\rev{(i)}})^T \diag{\yy}^H \projnull{ \qq(\tau, \nu)}  \diag{\yy}  (\upsb^{\rev{(i)}})^\conj
    \\ \nonumber
    &\approx \left[e^{-j \bxihat^{(i)}} \odot (\boldone - j \bxidelta)\right]^T \diag{\yy}^H \projnull{ \qq(\tau, \nu)}  \diag{\yy} 
    \\ \nonumber
    &~~\times \left[e^{j \bxihat^{(i)}} \odot (\boldone + j \bxidelta) \right] 
    \\ \label{eq_approx_saa}
    &= (\boldone - j \bxidelta)^T \boldD^{\rev{(i)}}(\tau, \nu)   (\boldone + j \bxidelta)~,
\end{align}
where $\boldD^{\rev{(i)}}(\tau, \nu) \in \complexset{NM}{NM}$ is defined as\footnote{Note that \eqref{eq_approx_saa} is guaranteed to be real since $\boldD(\tau, \nu)$ in \eqref{eq_gammanutau} is Hermitian.}
\begin{align}\label{eq_gammanutau}
    \boldD^{\rev{(i)}}(\tau, \nu)   &\triangleq \rev{\Big[} \diag{\yy}^H \projnull{ \qq(\tau, \nu)}  \diag{\yy} \rev{\Big]} \odot \rev{\Big[} e^{-j \bxihat^{(i)}} (e^{j \bxihat^{(i)}})^T \rev{\Big]} ~.
\end{align}
Now, substituting the approximation \eqref{eq_approx_saa} into \eqref{eq_nll_original_saa}, we have
\begin{align} 
    &\llr(\tau, \nu, \bxidelta\rev{+ \bxihat^{(i)}}) \\ \label{eq_nll_original_saa2}
    &\approx \frac{ 1}{ \sigma^2} (\boldone - j \bxidelta)^T \boldD^{\rev{(i)}}(\tau, \nu)   (\boldone + j \bxidelta)
    \\ \nonumber
    &~~+ (\bxihat^{(i)})^T \boldR(\tau)^{-1} \bxihat^{(i)} + \bxidelta^T \boldR(\tau)^{-1} \bxidelta + 2 \bxidelta^T \boldR(\tau)^{-1} \bxihat^{(i)}
    \\ \nonumber
    &~~ + \log \det \boldR(\tau) 
    \\  \label{eq_nll_original_saa2_2}
    &= \frac{ 1}{ \sigma^2} \Big[ \boldone^T \realp{\boldD^{\rev{(i)}}(\tau, \nu)} \boldone + \bxidelta^T \realp{\boldD^{\rev{(i)}}(\tau, \nu)} \bxidelta
    \\ \nonumber &~~+ 2 \bxidelta^T \imp{\boldD^{\rev{(i)}}(\tau, \nu} \boldone  \Big] + \bxidelta^T \boldR(\tau)^{-1} \bxidelta + 2 \bxidelta^T \boldR(\tau)^{-1} \bxihat^{(i)}
    \\ \nonumber &~~+ (\bxihat^{(i)})^T \boldR(\tau)^{-1} \bxihat^{(i)} +  \log \det \boldR(\tau) 
    \\ \label{eq_nll_original_saa3}
    &= \bxidelta^T \left( \frac{1}{\sigma^2} \realp{\boldD^{\rev{(i)}}(\tau, \nu)} +  \boldR(\tau)^{-1} \right) \bxidelta 
    \\ \nonumber &~~+ 2 \bxidelta^T \left( \frac{1}{\sigma^2} \imp{\boldD^{\rev{(i)}}(\tau, \nu)} \boldone  +  \boldR(\tau)^{-1} \bxihat^{(i)}\right)
    \\ \nonumber &~~+ \frac{ 1}{ \sigma^2} \boldone^T \realp{\boldD^{\rev{(i)}}(\tau, \nu)} \boldone + (\bxihat^{(i)})^T \boldR(\tau)^{-1} \bxihat^{(i)} +  \log \det \boldR(\tau) ~,
\end{align}
where \eqref{eq_nll_original_saa2_2} is due to $\bxidelta^{\rev{(i)}}$ being a real vector and $\boldD(\tau, \nu)$ being a Hermitian matrix (please see Appendix~\ref{app_obtain} for details).

Observing that \eqref{eq_nll_original_saa3} is quadratic in $\bxidelta$, the optimal estimate of $\bxidelta$ that minimizes the approximated version of $\llr(\tau, \nu, \bxidelta\rev{+ \bxihat^{(i)}})$ in \eqref{eq_nll_original_saa3} can be written in closed-form for a given delay-Doppler pair $(\tau, \nu)$ as follows:
\begin{align} \nonumber
    \bxideltahat(\tau, \nu)
    &= - \boldR(\tau) \Big(  \realp{\boldD^{\rev{(i)}}(\tau, \nu)} \boldR(\tau) + \sigma^2  \Imatrix \Big)^{-1} 
    \\ \label{eq_pn_res_est_doppler}
    &~~~~~\times \left( \imp{\boldD^{\rev{(i)}}(\tau, \nu)} \boldone  + \sigma^2 \boldR(\tau)^{-1} \bxihat^{(i)} \right) ~.
\end{align}
% \begin{align} \nonumber
%     \bxideltahat(\tau, \nu)
%     &= - \Big(  \realp{\boldD(\tau, \nu)} + \sigma^2  \boldR(\tau)^{-1} \Big)^{-1} 
%     \\ \label{eq_pn_res_est_doppler}
%     &~~~~~\times \left( \imp{\boldD(\tau, \nu)} \boldone  + \sigma^2 \boldR(\tau)^{-1} \bxihat^{(i)} \right) ~.
% \end{align}
Finally, using the residual estimate in \eqref{eq_pn_res_est_doppler} and the definition in \eqref{eq_bxidelta}, the PN estimate can be updated as
\begin{align} \label{eq_bxi_update}
    \bxihat^{(i+1)} = \bxihat^{(i)} + \bxideltahat(\tau, \nu) ~.
\end{align}

%%%%%%%%%%%%%%%%%%%%%%%%%%%%%%%%%%%%%%%%%%%%%%%
\subsection{Alternating Optimization to Solve \eqref{eq_hybrid_ml_map_doppler3}}\label{sec_alt_opt}
The iterative procedure developed in \eqref{eq_pn_res_est_doppler} and \eqref{eq_bxi_update} for updating the PN estimate $\bxihat$ as a function of delay and Doppler motivates an alternating optimization method to solve the original hybrid ML/MAP optimization problem \eqref{eq_hybrid_ml_map_doppler3}. Hence, we propose to estimate delay, Doppler and PN in \eqref{eq_hybrid_ml_map_doppler3} using an \textit{iterative refinement} approach that alternates between \textit{PN estimation} and \textit{delay-Doppler estimation} as follows:
\begin{itemize}
    \item \textbf{Update $\bxihat$ for Fixed $(\tau, \nu)$:} For a given delay-Doppler pair $(\tau, \nu)$, we compute the residual PN via \eqref{eq_pn_res_est_doppler} and update $\bxihat$ via \eqref{eq_bxi_update}.
    \item \textbf{Update $(\tau, \nu)$ for Fixed $\bxihat$:} For a given PN estimate $\bxihat$, we find the optimal delay-Doppler pair by solving \eqref{eq_hybrid_ml_map_doppler3}:
    \begin{align}\label{eq_hybrid_ml_map_doppler3_fixed}
    (\tauhat, \nuhat)  = \arg \min_{ \tau, \nu} ~ \llr(\tau, \nu, \bxihat) ~,
\end{align}
which is equivalent to (please see Appendix~\ref{app_delayDoppler} for details)
\begin{align}\label{eq_llrmax}
    (\tauhat, \nuhat)  &= \arg \max_{ \tau, \nu} ~ \llrt(\tau, \nu, \bxihat)  ~,
\end{align}
where
\begin{align}\nonumber
    \llrt(\tau, \nu, \bxihat) &= \frac{ \absbig{ \bb^H(\tau) \rev{\Big[} \boldX^\conj \odot \FF_N \big( \Wbhat^\conj \odot \boldY \big)  \rev{\Big]} \cc(\nu) }^2 }{ \sigma^2 \norm{\boldX}_F^2  }  \\ \label{eq_llrtilde} &~~ - \bxihat^T \boldR(\tau)^{-1} \bxihat - \log \det \boldR(\tau) ~,
\end{align}
and $\Wbhat \triangleq \veccinv{ e^{- j \bxihat}}$.
%is employed for PN compensation. 
\end{itemize}
%We observe that the first term in \eqref{eq_llrtilde} applies the standard 2-D FFT method\footnote{From \eqref{eq_steer_delay} and \eqref{eq_steer_doppler}, it is noted that $\bb(\tau)$ and $\cc(\nu)$ coincide with DFT matrix columns over a uniformly sampled delay-Doppler grid.} \cite{RadCom_Proc_IEEE_2011,OFDM_Radar_Phd_2014,OFDM_Radar_Corr_TAES_2020,MIMO_OFDM_ICI_JSTSP_2021} (discussed in Sec.~\ref{sec_ideal}) on the PN-compensated observation $ \Wbhat^\conj \odot \boldY$. 
The entire algorithm to solve \eqref{eq_hybrid_ml_map_doppler3} is summarized in Algorithm~\ref{alg_map_isaa} and referred to as MAP-ISAA\rev{\footnote{\label{fn_data}\rev{Algorithm~\ref{alg_map_isaa} is agnostic to data symbols $\boldX$ in the sense that it can work well with arbitrary $\boldX$ and imposes no constraints on data symbols generated by the communications system.}}}.
%\footnote{\label{fn_ambiguity_doppler}\rev{Please refer to Sec.~\ref{sec_phase_Doppler_ambiguity} in the supplementary material for a discussion on the issue of phase and Doppler ambiguity at high SNRs and a regularization based solution to tackle this problem.}}
\rev{Under certain conditions, Algorithm~\ref{alg_map_isaa} converges to a stationary point of \eqref{eq_hybrid_ml_map_doppler3} (please see Appendix~\ref{sec_conv_analysis} for a detailed convergence analysis).} Employing the conjugate gradient (CG) method \cite{OFDM_Joint_PHN_TSP_2006} to evaluate \eqref{eq_pn_res_est_doppler}, the per-iteration complexity of Algorithm~\ref{alg_map_isaa} can be obtained as (please see Appendix~\ref{app_comp_alg_map_isaa} for details)
\begin{align}
    \mathcal{O}\big( MN  \left( \log M + (I M_0 +1) \log N \right) \big) ~,
\end{align}
where $I$ is the number of CG iterations and $M_0$ is the number of dominant blocks of $\boldR(\tau)$ in \eqref{eq_toep_block}, with $M_0 = 1$ for FROs due to Lemma~\ref{lemma_fro_blkdiag} and $1 \leq M_0 \ll M$ for PLLs (typically, $M_0 \leq 3$).

%%%%%%%%%%%%%%%%%%%%%%%%%%%%%%%%%%%%%%%%%%%%%%%%%%%%%%
% MAP-ISAA Algorithm
\begin{algorithm}[t]
	\caption{Joint Delay, Doppler and PN Estimation with MAP Criterion via Iterated Small Angle Approximation (MAP-ISAA).}
	\label{alg_map_isaa}
% 	\footnotesize
	\begin{algorithmic}[1]
	    \State \textbf{Input:} Fast-time/slow-time observations $\boldY$ in \eqref{eq_ym_all_multi}, convergence thresholds $\epstau$ and $\epsnu$ for delay and Doppler, and maximum number of iterations $\imax$.
	    \State \textbf{Output:} Estimates of delay, Doppler and PN $\{ \tauhat, \nuhat, \bxihat \}$. 
	    \State \textbf{Initialization:} Set $i=0$.
	    \begin{enumerate}%[label=(\alph*)]
	        \item Initialize the PN estimate to be the all-zeros vector, i.e., $\bxihat^{(0)} = \boldzero$, in accordance with \eqref{eq_bxi_stat}.
	        \item Initialize the delay-Doppler pair to be the output of the standard 2-D FFT method, i.e.,
	        \begin{align}\nonumber
                (\tauhat^{(0)}, \nuhat^{(0)})  &= \arg \max_{ \tau, \nu} ~  \absbig{ \bb^H(\tau) \left( \boldX^\conj \odot \FF_N \boldY   \right) \cc(\nu) }^2 ~.
            \end{align}
	    \end{enumerate}
	    \State \textbf{Iterated Approximation Steps:} 
	    \State \textbf{while} $i < \imax$
	    \Indent
	        \State \label{alg_pn_step} Update PN estimate via \eqref{eq_pn_res_est_doppler} and \eqref{eq_bxi_update}:
	        \begin{align*} 
                \bxihat^{(i+1)} = \bxihat^{(i)} + \bxideltahat(\tauhat^{(i)}, \nuhat^{(i)}) ~.
            \end{align*}
            \State \label{alg_dd_step} Update delay-Doppler estimate via \eqref{eq_llrmax} and \eqref{eq_llrtilde}:
            \begin{align*}
    (\tauhat^{(i+1)}, \nuhat^{(i+1)})  &= \arg \max_{ \tau, \nu} ~ \llrt(\tau, \nu, \bxihat^{(i+1)})  ~.
            \end{align*}
            \State Set $i = i+1$.
            \State \textbf{if} $\abs{\tauhat^{(i)}-\tauhat^{(i-1)}} \leq \epstau$ and $\abs{\nuhat^{(i)}-\nuhat^{(i-1)}} \leq \epsnu$
	    \Indent
	    \State \textbf{break}
	    \EndIndent
	    \State \textbf{end if}
	    \EndIndent
	    \State \textbf{end while}
	\end{algorithmic}
	\normalsize
\end{algorithm}
%%%%%%%%%%%%%%%%%%%%%%%%%%%%%%%%%%%%%%%%%%%%%%%%%%%%%%

% %%%%%%%%%%%%%%%%%%%%%%%%%%%%%%%%%%%%%%%%%%%%%%%
% \subsection{Low-Complexity Implementation}

% %%%%%%%%%%%%%%%%%%%%%%%%%%%%%%%%%%%%%%%%%%%%%%%
% \section{Hybrid CRB Analysis}\label{sec_crb}

%%%%%%%%%%%%%%%%%%%%%%%%%%%%%%%%%%%%%%%%%%%%%%%%%%%%%%%%
%%%%%%%%%%%%%%%%%%%%%%%%%%%%%%%%%%%%%%%%%%%%%%%%%%%%%%%%
\section{Extension to Range-Ambiguous Targets: From Mitigation to Exploitation}\label{sec_extension}
We have up to now focused on how to estimate and \textit{mitigate} PN in the OFDM sensing problem of interest. In this section, we extend the proposed estimation framework to the case of range-ambiguous targets and devise a PN \textit{exploitation} approach to resolve range ambiguity.

% \subsection{Turning PN from Foe to Friend: Radar vs. Communications}
\subsection{PN Exploitation: Concept Description}
\label{sec_pn_foe_friend}
An important peculiarity of PN in OFDM monostatic \textit{sensing} is that it provides additional source of information on target delays that can be exploited to improve delay estimation performance, whereas PN in OFDM \textit{communications} always degrades performance (e.g., \cite{OFDM_Cons_PN_TCOM_2004,PN_BCRB_TSP_2013,OFDM_PN_HCRB_2014,PN_OFDM_TSP_2017,PN_CohBw_2021_TWC}). More precisely, it is seen from \eqref{eq_steer_delay}, \eqref{eq_bxi_stat} and \eqref{eq_qtaunu} that both $\bb(\tau)$ and $\boldR(\tau)$, which are functions of the unknown delay $\tau$, have an impact on the observation \eqref{eq_y_mult_doppler} and the hybrid ML/MAP cost function \eqref{eq_hybrid_ml_map_doppler4}. Hence, in addition to the standard frequency-domain phase rotations in $\bb(\tau)$, the covariance matrix $\boldR(\tau)$ of the PN vector conveys information on target delay. As noticed from the elements of $\boldR(\tau)$ in \eqref{eq_exp_xi2}--\eqref{eq_rtau_entries}, there is no maximum unambiguous delay imposed by $\boldR(\tau)$ in estimating $\tau$. This implies that we can detect the true ranges of the targets with $\tau > T$ by exploiting the information in $\boldR(\tau)$. However, for the ideal case without PN, we can only extract the delay information from $\bb(\tau)$ (e.g., \cite{Passive_OFDM_2010,RadCom_Proc_IEEE_2011,OFDM_Radar_Corr_TAES_2020,OFDM_Radar_Phd_2014}), which leads to range ambiguity for targets with $\tau > T$ due to the periodicity of the complex exponential terms in \eqref{eq_steer_delay}.

\subsection{PN Exploitation: \rev{Proposed} Algorithm}\label{sec_pn_exp_alg}
%\subsection{Range Ambiguity Resolution via PN Exploitation}
Let us consider a range-ambiguous target with true (unambiguous) delay $\tau = \taup + k T$ for some integer $k \geq 1$ denoting the ambiguity index, where $0 < \taup \leq T$ is the principal delay of the target, i.e., $\taup = \tau \Mod{T}$. Based on the observations in Sec.~\ref{sec_pn_foe_friend}, we propose to resolve the range ambiguity using the PN statistics represented by $\boldR(\tau)$. The key insight here is that the observation-related (first) term in \eqref{eq_hybrid_ml_map_doppler4} assumes the same value for $\tau$ and $\taup$ due to inherent range ambiguity in $\bb(\cdot)$, while the statistics of the PN estimate $\bxihat$ at the output of Algorithm~\ref{alg_map_isaa} would match only with $\boldR(\tau)$ (not with $\boldR(\taup)$), enabling range ambiguity resolution\footnote{Please see Fig.~\ref{fig_covarianceProfilePNExploitation_SNR_20dB} in Sec.~\ref{sec_pn_exp} for an illustration of this phenomenon.}. Inspired by this insight, we formulate a parametric covariance matrix reconstruction problem\rev{\footnote{\label{fn_exip}\rev{The covariance matching problem formulated in \eqref{eq_parametric_Frob} has a theoretical justification based on the extended invariance principle (EXIP) \cite{stoica1989reparametrization,ottersten1998covariance} and the ML estimator of delay from $\bxihat$. Please refer to Appendix~\ref{sec_cov_match} for details.}}} that minimizes the Frobenius distance between the sample matrix $\boldRhat = \bxihat \bxihat^T$ and the parametric matrix $\boldR(\tau)$ in \eqref{eq_toep_block}:
\begin{align}\label{eq_parametric_Frob}
    \tauhatu = \arg \min_{ \tau} ~ \normsmall{ \boldR(\tau) - \boldRhat }_F^2 ~,
\end{align}
potentially yielding the \textit{unambiguous (true) delay estimate} $\tauhatu$. Exploiting the Toeplitz-block Toeplitz structure in \eqref{eq_toep_block}, the problem \eqref{eq_parametric_Frob} reduces to \cite{toeplitz_radar_2020}
\begin{align}\label{eq_parametric_Frob2}
    \tauhatu = \arg \min_{ \tau} ~  \normsmall{ [\boldR(\tau)]_{0,:} - \rrhat }^2 ~,
\end{align}
where $[\boldR(\tau)]_{0,:}$ is the first row of $\boldR(\tau)$ and
\begin{align}\label{eq_rhat_row}
    \rhat_i = \frac{1}{NM-i} \sum_{k=0}^{NM-i-1} \xihat_k \, \xihat_{k+i} ~,
\end{align}
with $\rrhat = [\rhat_0, \ldots, \rhat_{NM-1} ]^T$ and $\bxihat = [\xihat_0 , \ldots, \xihat_{NM-1}]^T$.

In Algorithm~\ref{alg_resolve_amb}, we present the proposed algorithm for range ambiguity resolution via PN exploitation. We restrict the search interval for range in Algorithm~\ref{alg_map_isaa} such that the resulting delay estimate $\tauhat$ is ambiguous. As seen from Algorithm~\ref{alg_resolve_amb}, the prior information extracted from the PN estimates is exploited to find the range ambiguity index of the target (i.e., ambiguity resolution), while the principal range estimate of the target comes from Algorithm~\ref{alg_map_isaa}.

%%%%%%%%%%%%%%%%%%%%%%%%%%%%%%%%%%%%%%%%%%%%%%%%%%%%%%
% Range Ambiguity Resolution Algorithm
\begin{algorithm}[t]
	\caption{PN Exploitation to Resolve Range Ambiguity.}
	\label{alg_resolve_amb}
% 	\footnotesize
	\begin{algorithmic}[1]
	    \State \textbf{Input:} Estimates of (ambiguous) delay and PN at the output of Algorithm~\ref{alg_map_isaa} $\{\tauhat, \bxihat \} $.
	    \State \textbf{Output:} Unambiguous (true) estimate of delay $\tauhatu$. 
	    \begin{enumerate}%[label=(\alph*)]
	        \item Compute an estimate $\rrhat = [\rhat_0, \ldots, \rhat_{NM-1} ]^T$ of the first row of the PN covariance matrix $\boldR(\tau)$ in \eqref{eq_toep_block} by using the PN estimate $\bxihat$:
	        \begin{align} \label{eq_rhat}
                \rhat_i = \frac{1}{NM-i} \sum_{k=0}^{NM-i-1} \xihat_k \, \xihat_{k+i} ~,
            \end{align}
            where $\bxihat = [\xihat_0 , \ldots, \xihat_{NM-1}]^T$.
            \item Construct the set of possible delay values by unfolding the ambiguous delay estimate $\tauhat$ up to some maximum ambiguity index $K$:
            \begin{align}\label{eq_set_delay}
                \mathcal{T} = \{ \tau ~ \lvert ~ \tau = \tauhat + kT, \, k = 0, \ldots, K \}  ~.
            \end{align}
	        \item Find the unambiguous (true) estimate of delay by solving the parametric Toeplitz-block Toeplitz PN matrix reconstruction problem:
	        \begin{align}\label{eq_parametric_Frob3}
                \tauhatu = \arg \min_{ \tau \in \mathcal{T}} ~  \normsmall{ [\boldR(\tau)]_{0,:} - \rrhat }^2 ~.
            \end{align}
	    \end{enumerate}
	\end{algorithmic}
	\normalsize
\end{algorithm}
%%%%%%%%%%%%%%%%%%%%%%%%%%%%%%%%%%%%%%%%%%%%%%%%%%%%%%

%%%%%%%%%%%%%%%%%%%%%%%%%%%%%%%%%%%%%%%%%%%%%%%%%%%%%%%%
%%%%%%%%%%%%%%%%%%%%%%%%%%%%%%%%%%%%%%%%%%%%%%%%%%%%%%%%
\section{Simulation Results}\label{sec_sim}
In this section, we assess the performance of the proposed delay-Doppler estimation algorithm for PN-impaired OFDM radar sensing using the mmWave setting in Table~\ref{tab_parameters}. In the observation model \eqref{eq_ym_all_multi}, the complex data symbols $\boldX$ are chosen randomly from the QPSK constellation. To evaluate the RMSE performances, we generate a total of $2500$ Monte Carlo realizations consisting of every combination of $50$ independent realizations of PN vector $\bxi$ and AWGN vector $\zz$ in \eqref{eq_y_mult_doppler}. Unless otherwise stated, we consider a target with $R = 30 \, \rm{m}$ and $v = 20 \, \rm{m/s}$, and an oscillator with $\fdb = 200 \, \rm{kHz}$ and $\floop = 1 \, \rm{MHz}$ (in the case of PLL) \cite{Demir_PN_2006}. Moreover, SNR is defined as $\snr = \abs{\alpha}^2/(2 \sigma^2)$ according to the model in \eqref{eq_ym_all_multi}. To provide comparative performance analysis, we consider the following benchmark processing schemes: 
\begin{itemize}
    \item \textit{MAP-ISAA}: The proposed hybrid ML/MAP estimation algorithm based on ISAA, described in Algorithm~\ref{alg_map_isaa}.
    
    \item \textit{2-D FFT}: The standard 2-D FFT method used in OFDM radar processing \cite{OFDM_Passive_Res_2017_TSP,Passive_OFDM_2010,RadCom_Proc_IEEE_2011,OFDM_Radar_Corr_TAES_2020,OFDM_Radar_Phd_2014}, which corresponds to the optimal estimator \rev{in the ML sense} in the absence of PN\rev{\footnote{\label{fn_2d_fft}\rev{The ML estimator in the absence of PN can be obtained as a special case of the hybrid ML/MAP estimator in \eqref{eq_hybrid_ml_map_doppler3} for known PN matrix with all elements being equal to $1$, i.e., $\boldW = \boldone_{N \times M}$, as mentioned in Sec.~\ref{sec_ideal}. This special case has already been derived in \eqref{eq_hybrid_ml_map_doppler3_fixed}--\eqref{eq_llrtilde} for a given PN estimate $\Wbhat$. Inserting $\Wbhat = \boldone_{N \times M}$ and $\bxihat = \boldzero_{NM}$ into \eqref{eq_llrtilde}, one can readily derive the ML estimator in \eqref{eq_fft_benc}, which can be implemented via 2-D FFT as $\bb(\tau)$ in \eqref{eq_steer_delay} and $\cc(\nu)$ in \eqref{eq_steer_doppler} are DFT matrix columns for a uniformly sampled delay-Doppler grid.}}}:
    \begin{align} \label{eq_fft_benc}
        (\tauhat, \nuhat) = \arg \rev{\max_{\tau, \nu}} ~ \absbig{ \bb^H(\tau) \left( \boldX^\conj \odot \FF_N \boldY   \right) \cc(\nu) }^2~.
    \end{align}
    
    \item \textit{2-D FFT (PN-free)}: The 2-D FFT method applied on the PN-free version of the radar observations, $\boldYfr$, in \eqref{eq_ym_special}: 
    \rev{\begin{align} \label{eq_fft_benc_pn_free}
        (\tauhat, \nuhat) = \arg \max_{\tau, \nu} ~ \absbig{ \bb^H(\tau) \left( \boldX^\conj \odot \FF_N \boldYfr   \right) \cc(\nu) }^2~.
    \end{align}}
    This benchmark will provide insights into PN-induced performance losses and gains in delay-Doppler estimation.
\end{itemize}
Besides the above schemes, we also plot the hybrid CRBs \cite{hybrid_CRB_LSP_2008,Hybrid_ML_MAP_TSP} to quantify the theoretical performance bounds. The hybrid CRB in the presence and absence of PN will be denoted as \quot{\textit{CRB}} and \quot{\textit{CRB (PN-free)}}, respectively, which theoretically lower-bound the RMSE of \quot{\textit{MAP-ISAA}} and \quot{\textit{2-D FFT (PN-free)}} algorithms.

%%%%% simulation parameters table %%%%%%%%%%
\begin{table}\footnotesize
\caption{OFDM Simulation Parameters}
\vspace{0.05in}
\centering
    \begin{tabular}{|l|l|}
        \hline
        \textbf{Parameter} & \textbf{Value} \\ \hline
        Carrier Frequency, $\fc$  & $28 \, \rm{GHz}$ \\ \hline
        Total Bandwidth, $B$ & $50 \, \rm{MHz}$ \\ \hline
        Number of Subcarriers, $N$ & $256$ \\ \hline
        Number of Symbols, $M$ & $10$ \\ \hline
        Subcarrier Spacing, $\deltaf$ & $195.31 \, \rm{kHz}$  \\ \hline
        Symbol Duration, $T$ & $5.12 \, \mu \rm{s}$  \\ \hline
        Cyclic Prefix Duration, $\Tcp$ & $1.28 \, \mu \rm{s}$  \\ \hline
        Total Symbol Duration, $\Tsym$ & $6.40 \, \rm{\mu s}$  \\ \hline
    \end{tabular}
    \label{tab_parameters}
    \vspace{-0.1in}
\end{table}
%%%%%%%%%%%%%%%%%%%%%%%%%%%%%%%%%%%%%%%%%%%%%%%%%%%%%%%%

% explore fundamental limits on delay-Doppler estimation under the impact of PN and 

In what follows, we first evaluate the RMSE performances of the considered processing schemes under various operating conditions with regard to SNR, oscillator quality and target range. Then, we illustrate the convergence behavior of the proposed algorithm in Algorithm~\ref{alg_map_isaa}. Finally, we demonstrate the PN exploitation capability of the proposed approach in Algorithm~\ref{alg_resolve_amb}.

% %%%%%%%%%%%%%%%%%%%%%%%%%%%%%%%%%%%%%%%%%%%%%%%%%%%%%%%%
% \subsection{Analysis of Theoretical Limits for FROs and PLLs}

%%%%%%%%%%%%%%%%%%%%%%%%%%%%%%%%%%%%%%%%%%%%%%%%%%%%%%%%
\subsection{Performance with respect to SNR}\label{sec_perf_snr}
We first assess the performance of the proposed MAP-ISAA algorithm, along with the FFT-based benchmark methods, with respect to SNR. Fig.~\ref{fig_rmse_r30_v20_f3dB_200_FRO} and Fig.~\ref{fig_rmse_r30_v20_f3dB_200_PLL} show (for FRO and PLL architectures, respectively) the range and velocity RMSEs of the considered schemes\rev{\footnote{\label{fn_practical}\rev{Due to discretization of search space and its confinement to a finite interval for practical implementation of the estimators, the RMSE might slightly fall below the CRB in certain rare scenarios.}}}. 
%as well as the RMSE in PN estimation achieved by the MAP-ISAA algorithm. 
In terms of ranging performance in Fig.~\ref{fig_range_rmse_r30_v20_f3dB_200_FRO} and Fig.~\ref{fig_range_rmse_r30_v20_f3dB_200_PLL}, the proposed algorithm achieves the CRB and significantly outperforms the 2-D FFT benchmark for both FRO and PLL architectures, especially at medium and high SNRs, where an order-of-magnitude improvement in ranging accuracy can be observed. It is seen that the performance of the standard FFT method saturates above a certain SNR level, while the proposed approach can effectively utilize the prior information on PN to compensate for its impact on the observations and avoid such saturation behavior by attaining its theoretical lower bound (which decreases with increasing SNR). Moreover, in compliance with the theoretical bounds, MAP-ISAA exhibits ranging performance very close to that achieved in the absence of PN, which evidences its remarkable PN compensation capability.

In contrast to their ranging performances, FRO and PLL display different trends in velocity estimation, depicted in Fig.~\ref{fig_vel_rmse_r30_v20_f3dB_200_FRO} and Fig.~\ref{fig_vel_rmse_r30_v20_f3dB_200_PLL}. For PLL, the MAP-ISAA algorithm can provide noticeable improvements in velocity RMSE over the FFT benchmark, leading to gains on the order of several $10 \, \rm{cm/s}$, while in case of FRO performance gains seem negligible with respect to the FFT method. The same observation is also valid when comparing the CRB to the RMSE of the FFT method since MAP-ISAA can get very close to the CRB for both types of oscillators. This difference between FRO and PLL in velocity estimation can be attributed to the following two facts: \textit{(i)} velocity information in \eqref{eq_ym_all_multi} is extracted from \textit{slow-time (symbol-to-symbol) phase rotations} represented by $\cc(\nu)$ in \eqref{eq_steer_doppler}, and \textit{(ii)} PN samples in different symbols are \textit{uncorrelated} for FRO, while they can be \textit{correlated} for PLL, as shown in \eqref{eq_toep_block}, Lemma~\ref{lemma_fro_blkdiag} and Fig.~\ref{fig_PN_covariance_PLL_FRO}. Therefore, correlation of PN across symbols can be exploited in PLL for accurate PN estimation and compensation in slow-time, leading to better velocity estimates compared to FRO. In this respect, similar ranging performances for FRO and PLL can be explained by pointing out high correlation of PN in \textit{fast-time} for both oscillator types, represented by the block-diagonals $\boldR_0(\tau)$ in \eqref{eq_toep_block} and \eqref{eq_R_blkdiag}. 
%Finally, we reveal the remarkable PN estimation performance of Algorithm~\ref{alg_map_isaa}, as seen from Fig.~\ref{fig_pn_rmse_r30_v20_f3dB_200_FRO} and Fig.~\ref{fig_pn_rmse_r30_v20_f3dB_200_PLL}.

Comparing the asymptotic trends of the MAP-ISAA algorithm, as well as the corresponding CRB, between range and velocity estimation in Fig.~\ref{fig_rmse_r30_v20_f3dB_200_FRO} and Fig.~\ref{fig_rmse_r30_v20_f3dB_200_PLL}, we observe plateau in velocity estimation performance as opposed to monotonically decreasing errors for range estimation. In connection with this, PN causes only slight degradation of ranging performance, whereas velocity estimation can be severely degraded by PN. This is due to the fact that range information is gathered from frequency-domain (or, equivalently, fast-time) phase shifts across $\bb(\tau)$ in \eqref{eq_steer_delay}, while velocity information comes from slow-time phase shifts in \eqref{eq_steer_doppler}. Since PN enjoys much higher correlation in fast-time than in slow-time, such correlation can be utilized to cancel out its effect in fast-time and accordingly provide accurate range estimates.

%%%%%%%%%%%%%%%%%%%%%%%%%%%%%%%%%%%%%%%%%%
% range, vel and PN RMSE for FRO
\begin{figure}[t]
        \begin{center}
        %\vspace{-0.22in}
        \subfigure[]{
			 \label{fig_range_rmse_r30_v20_f3dB_200_FRO}
			 \includegraphics[width=0.42\textwidth]{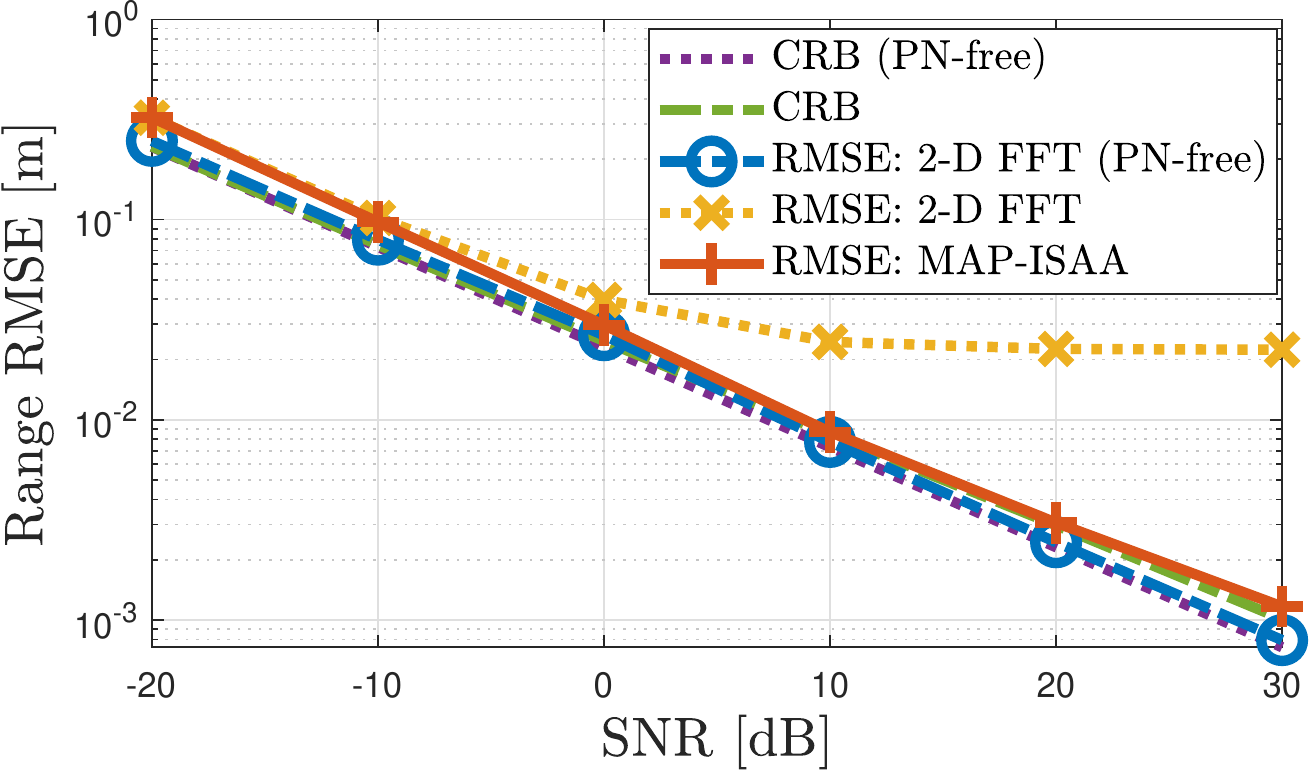}
			 %\vspace{-1in}
			 %\vspace{-5em}
			 %\hspace{-0.5in}
		}
        %\hfill 
        \subfigure[]{
			 \label{fig_vel_rmse_r30_v20_f3dB_200_FRO}
			 \includegraphics[width=0.42\textwidth]{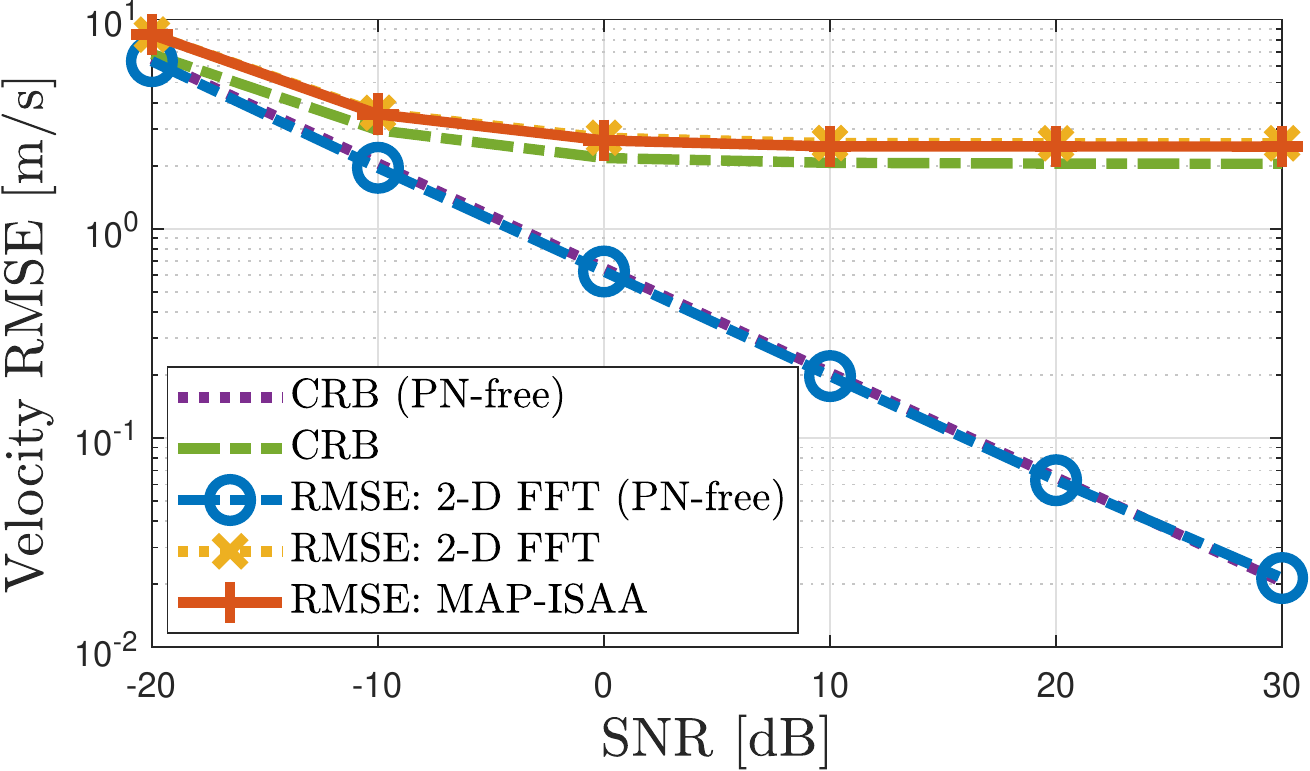}
		}
		
% 		\subfigure[]{
% 			 \label{fig_pn_rmse_r30_v20_f3dB_200_FRO}
% 			 \includegraphics[width=0.4\textwidth]{Figures/pn_rmse_r30_v20_f3dB_200_FRO-eps-converted-to.pdf}
% 		}
		
		\end{center}
		\vspace{-0.2in}
        \caption{\subref{fig_range_rmse_r30_v20_f3dB_200_FRO} Range and \subref{fig_vel_rmse_r30_v20_f3dB_200_FRO} velocity RMSE with respect to SNR for FRO with $\fdb = 200 \, \rm{kHz}$.} 
        \label{fig_rmse_r30_v20_f3dB_200_FRO}
        %\vspace{-0.12in}
        \vspace{-0.1in}
\end{figure}
%%%%%%%%%%%%%%%%%%%%%%%%%%%%%%%%%%%%%%%%%%

%%%%%%%%%%%%%%%%%%%%%%%%%%%%%%%%%%%%%%%%%%
% range, vel and PN RMSE for PLL
\begin{figure}[t]
        \begin{center}
        %\vspace{-0.22in}
        \subfigure[]{
			 \label{fig_range_rmse_r30_v20_f3dB_200_PLL}
			 \includegraphics[width=0.42\textwidth]{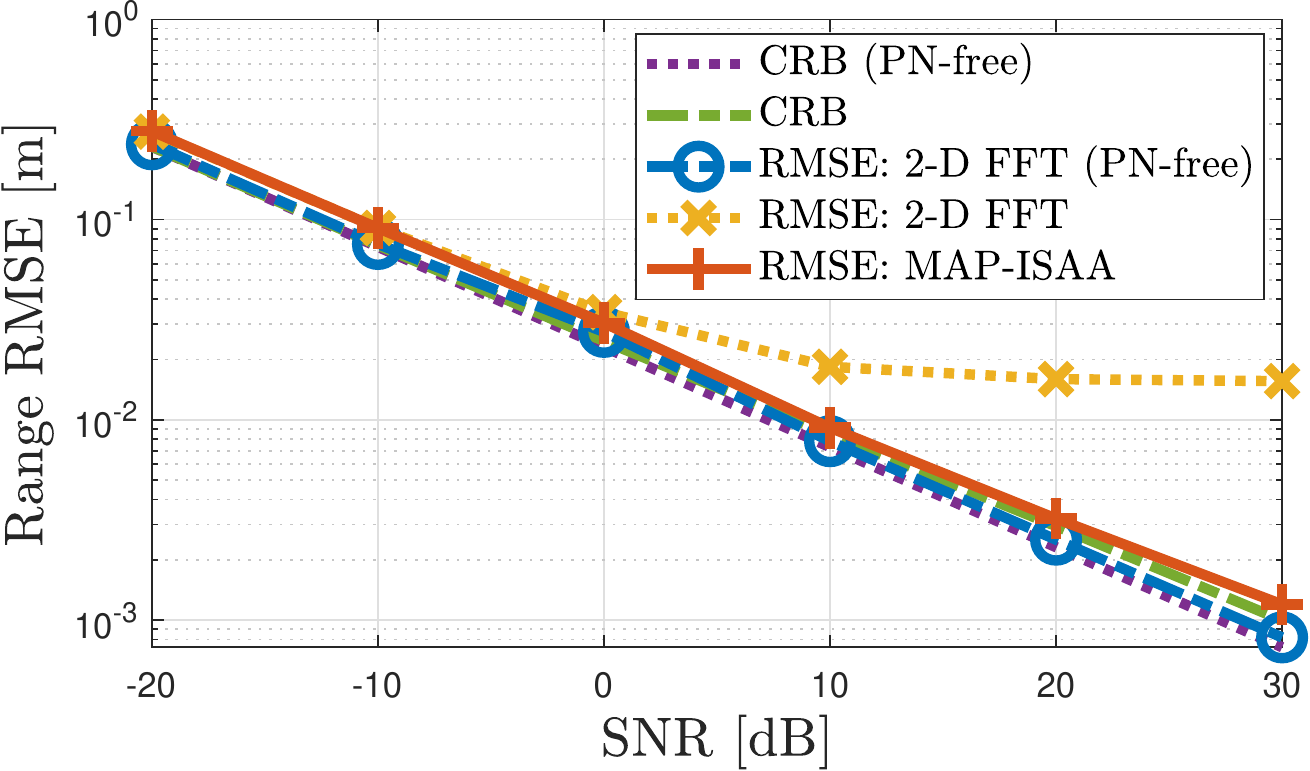}
			 %\vspace{-1in}
			 %\vspace{-5em}
			 %\hspace{-0.5in}
		}
        %\hfill 
        \subfigure[]{
			 \label{fig_vel_rmse_r30_v20_f3dB_200_PLL}
			 \includegraphics[width=0.42\textwidth]{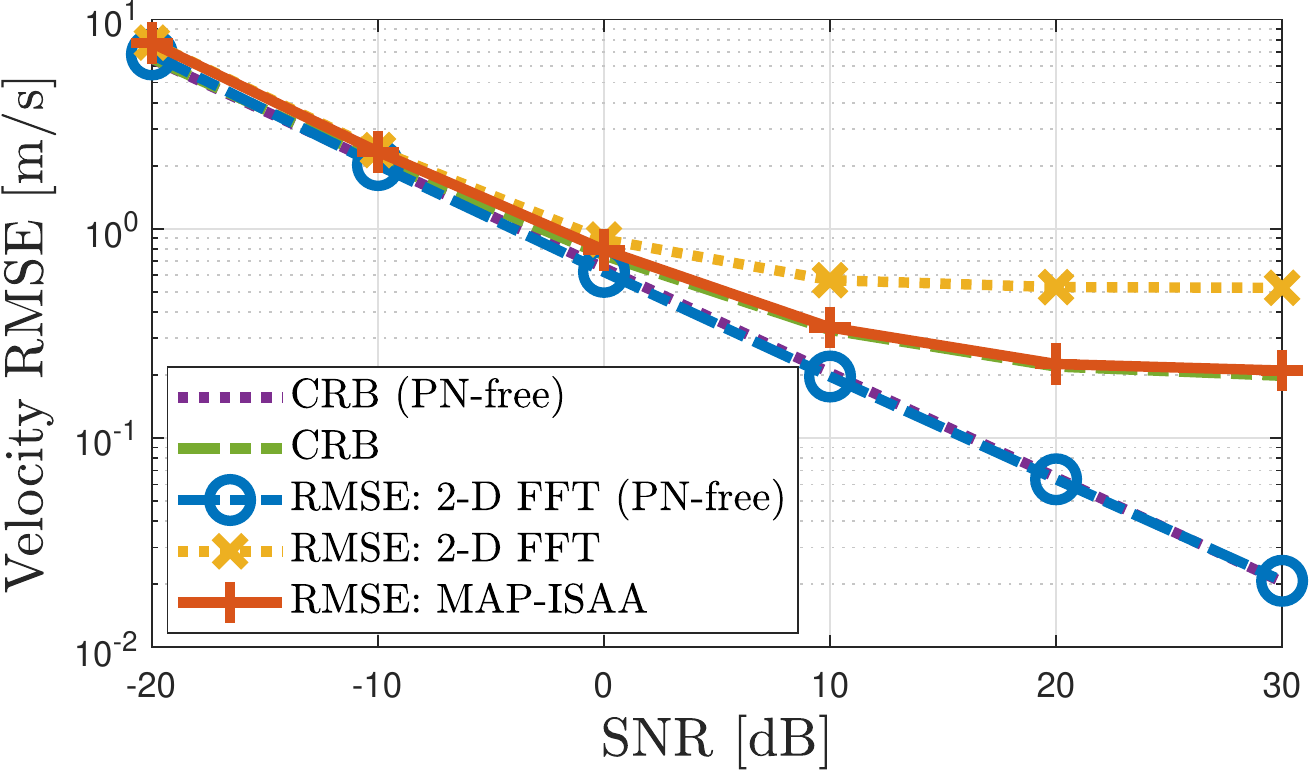}
		}
		
% 		\subfigure[]{
% 			 \label{fig_pn_rmse_r30_v20_f3dB_200_PLL}
% 			 \includegraphics[width=0.4\textwidth]{Figures/pn_rmse_r30_v20_f3dB_200_PLL-eps-converted-to.pdf}
% 		}
		
		\end{center}
		\vspace{-0.2in}
        \caption{\subref{fig_range_rmse_r30_v20_f3dB_200_PLL} Range and \subref{fig_vel_rmse_r30_v20_f3dB_200_PLL} velocity RMSE with respect to SNR for PLL with $\fdb = 200 \, \rm{kHz}$ and $\floop = 1 \, \rm{MHz}$.} 
        \label{fig_rmse_r30_v20_f3dB_200_PLL}
        %\vspace{-0.12in}
        \vspace{-0.1in}
\end{figure}
%%%%%%%%%%%%%%%%%%%%%%%%%%%%%%%%%%%%%%%%%%

%%%%%%%%%%%%%%%%%%%%%%%%%%%%%%%%%%%%%%%%%%%%%%%%%%%%%%%%
\subsection{Performance with respect to Oscillator Quality}\label{sec_perf_osc_qual}
We now investigate the performance of the considered schemes with respect to oscillator quality for both FRO and PLL implementations. To this end, in Fig.~\ref{fig_rmse_wrt_f3dB_FRO} we report accuracy in range, velocity and PN estimation\rev{\footnote{\label{fn_pn_rmse}\rev{The PN RMSE is defined as $\big[\sum_{c=1}^{C} \normsmall{\bxihat_c - \bxi_c}^2/ (NMC)\big]^{1/2}$, where $\bxihat_c \in \realset{NM}{1}$ and $\bxi_c \in \realset{NM}{1}$ denote, respectively, the estimated and true PN at the $\thn{c}$ Monte Carlo realization and $C = 2500$. For the 2-D FFT method, we set $\bxihat_c = \boldzero$.}}} against the $3 \, \rm{dB}$ bandwidth $\fdb$ of FRO at $\snr = 20 \, \rm{dB}$. It can be observed that the ranging performance of the proposed MAP-ISAA algorithm is highly robust in the face of worsening oscillator quality (i.e., increasing $\fdb$), maintaining almost the same RMSE level and attaining the corresponding bounds over a wide range of $\fdb$ values. On the other hand, the FFT benchmark suffers from a considerable performance degradation as $\fdb$ increases due to increasing PN variance, shown in Fig.~\ref{fig_pn_rmse_wrt_f3dB_FRO}. Additionally, range accuracy obtained by MAP-ISAA is very close to that achievable in the absence of PN for all $\fdb$ values, which demonstrates the effectiveness of the proposed PN estimation/compensation approach in Algorithm~\ref{alg_map_isaa}. Regarding velocity RMSE, similar trends to those in Sec.~\ref{sec_perf_snr} can be observed for both the CRB and the RMSE of MAP-ISAA, due to lack of PN correlation in slow-time (i.e., the block-diagonal structure of PN covariance matrix for FRO, as specified in Lemma~\ref{lemma_fro_blkdiag}).

%%%%%%%%%%%%%%%%%%%%%%%%%%%%%%%%%%%%%%%%%%
% range and vel RMSE for FRO with respect to f_3dB
\begin{figure}[t]
        \begin{center}
        %\vspace{-0.22in}
        \subfigure[]{
			 \label{fig_range_rmse_wrt_f3dB_FRO}
			 \includegraphics[width=0.42\textwidth]{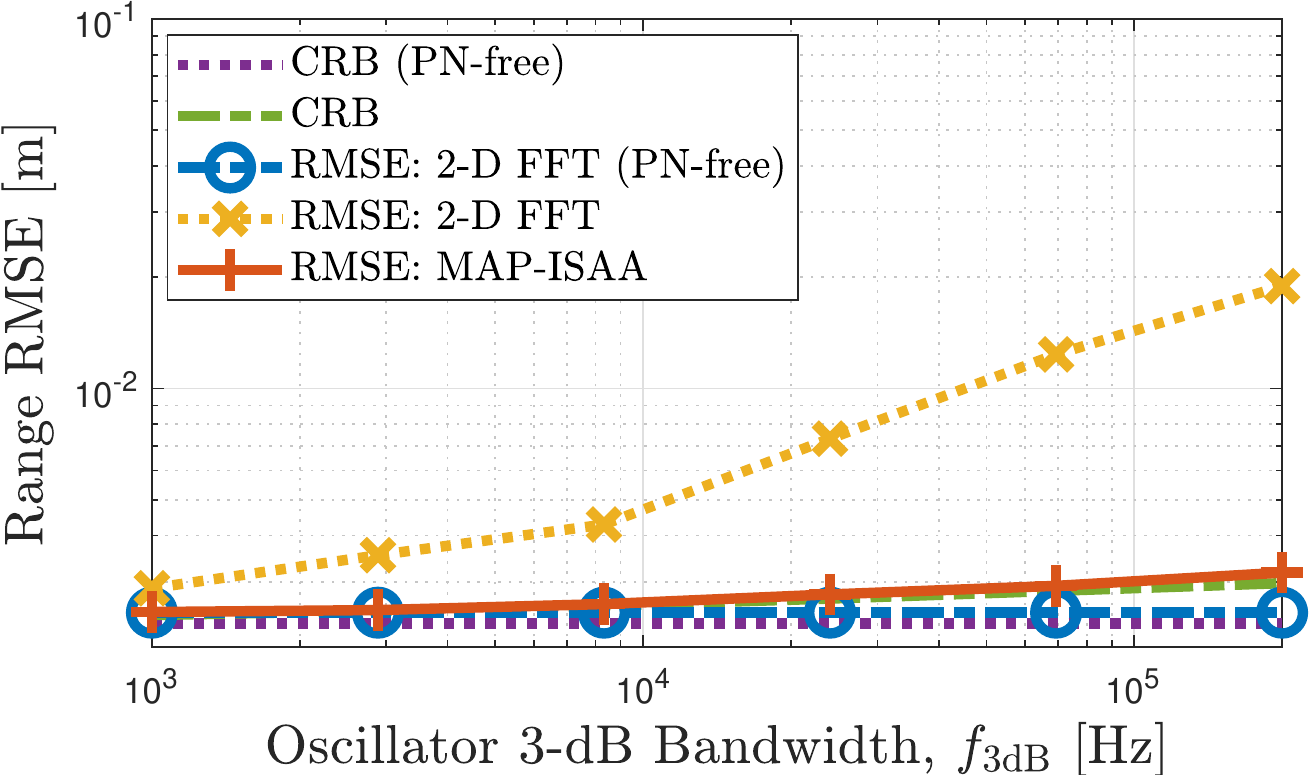}
			 %\vspace{-1in}
			 %\vspace{-5em}
			 %\hspace{-0.5in}
		}
        %\hfill 
        \subfigure[]{
			 \label{fig_vel_rmse_wrt_f3dB_FRO}
			 \includegraphics[width=0.42\textwidth]{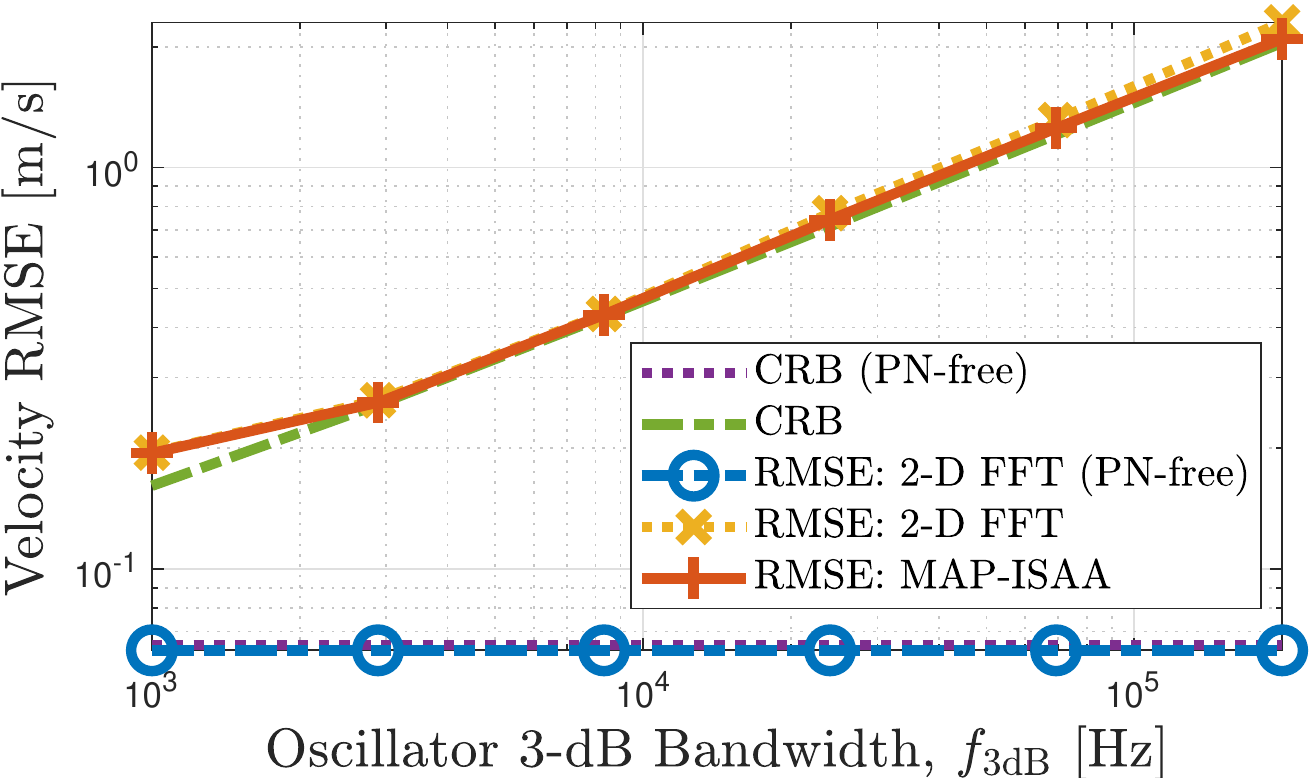}
		}
		
		\subfigure[]{
			 \label{fig_pn_rmse_wrt_f3dB_FRO}
			 \includegraphics[width=0.41\textwidth]{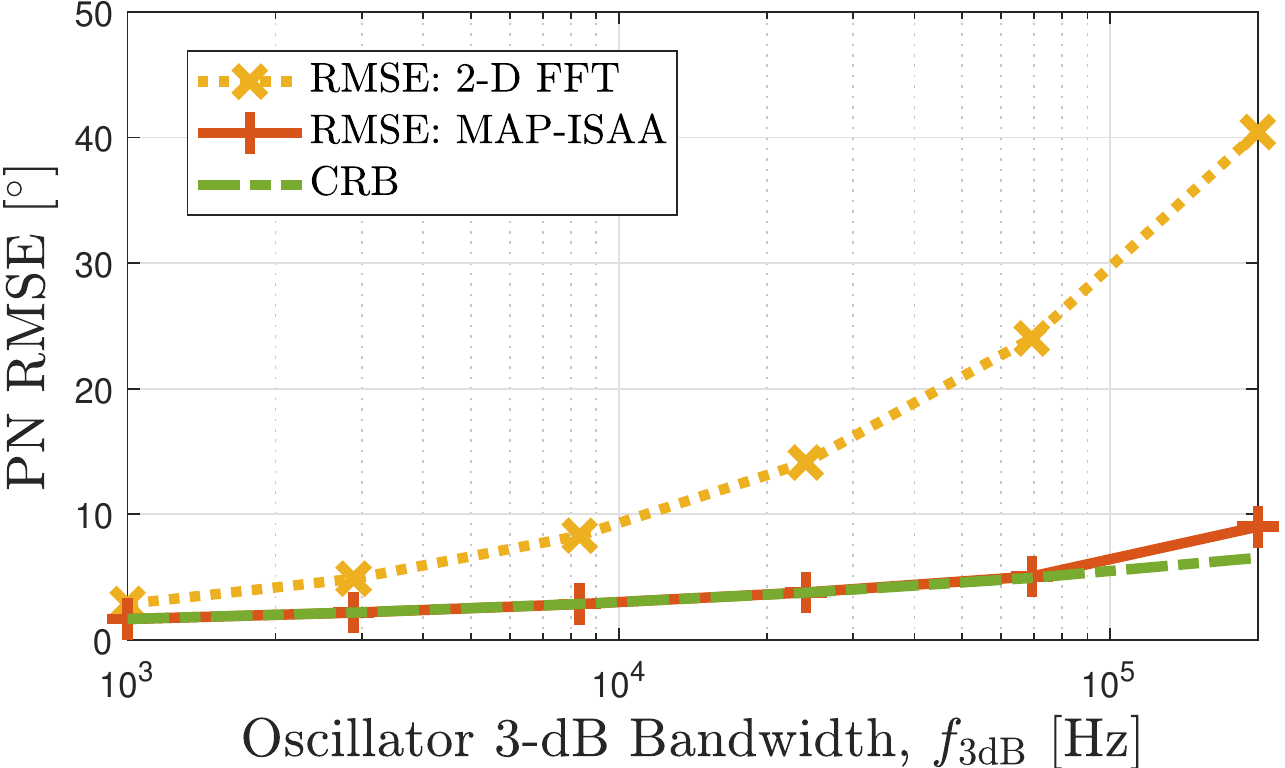}
		}
		
		\end{center}
		\vspace{-0.2in}
        \caption{\subref{fig_range_rmse_wrt_f3dB_FRO} Range, \subref{fig_vel_rmse_wrt_f3dB_FRO} velocity and \subref{fig_pn_rmse_wrt_f3dB_FRO} PN RMSE with respect to $3 \, \rm{dB}$ bandwidth $\fdb$ of FRO at $\snr = 20 \, \rm{dB}$.} 
        \label{fig_rmse_wrt_f3dB_FRO}
        %\vspace{-0.12in}
        \vspace{-0.1in}
\end{figure}
%%%%%%%%%%%%%%%%%%%%%%%%%%%%%%%%%%%%%%%%%%

Next, we consider the RMSE performances for PLL synthesizers in Fig.~\ref{fig_rmse_wrt_floop_PLL}, which shows accuracies against the PLL loop bandwidth $\floop$. In terms of both range and velocity RMSEs, the PN-ignorant FFT method performs worse with decreasing $\floop$ since PN variance increases as $\floop$ decreases, as seen from \eqref{eq_sigmatauwn_pll} and Fig.~\ref{fig_pn_rmse_wrt_floop_PLL}. On the other hand, the proposed algorithm attains the same ranging accuracy for the entire interval of $\floop$ values from $10 \, \rm{kHz}$ to $10 \, \rm{MHz}$, proving its robustness against the quality of PLL synthesizer. Hence, the performance gain of MAP-ISAA in range estimation with respect to the FFT benchmark becomes more pronounced with decreasing PLL quality (i.e., decreasing $\floop$). Moreover, 
%in agreement with the previous analysis of ranging accuracy with respect to SNR and $\fdb$, 
the ranging performance of the proposed approach under the impact of PN is near that of the FFT benchmark that uses PN-free radar observations over a broad range of $\floop$ values, which indicates almost perfect cancellation of the effect of PN. Furthermore, the velocity RMSE curves in Fig.~\ref{fig_vel_rmse_wrt_floop_PLL} reveal that MAP-ISAA performs very close to CRB and can provide gains on the order of $10 \, \rm{cm/s}$ over the FFT benchmark. Finally, from Fig.~\ref{fig_rmse_wrt_f3dB_FRO} and Fig.~\ref{fig_rmse_wrt_floop_PLL}, we note that MAP-ISAA can achieve the CRBs on range, velocity and PN estimation for both FROs and PLLs under various levels of oscillator quality.

%%%%%%%%%%%%%%%%%%%%%%%%%%%%%%%%%%%%%%%%%%
% range and vel RMSE for PLL with respect to f_loop
\begin{figure}[t]
        \begin{center}
        %\vspace{-0.22in}
        \subfigure[]{
			 \label{fig_range_rmse_wrt_floop_PLL}
			 \includegraphics[width=0.42\textwidth]{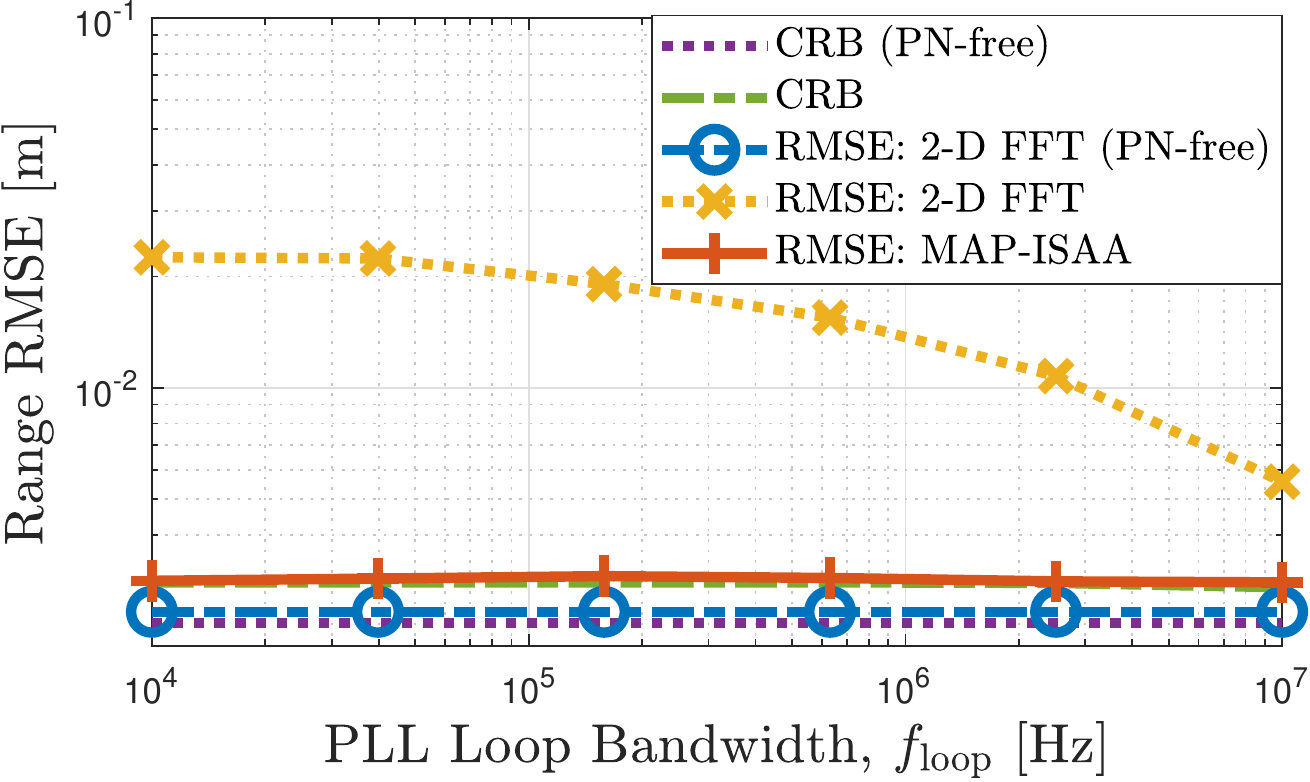}
			 %\vspace{-1in}
			 %\vspace{-5em}
			 %\hspace{-0.5in}
		}
        %\hfill 
        \subfigure[]{
			 \label{fig_vel_rmse_wrt_floop_PLL}
			 \includegraphics[width=0.42\textwidth]{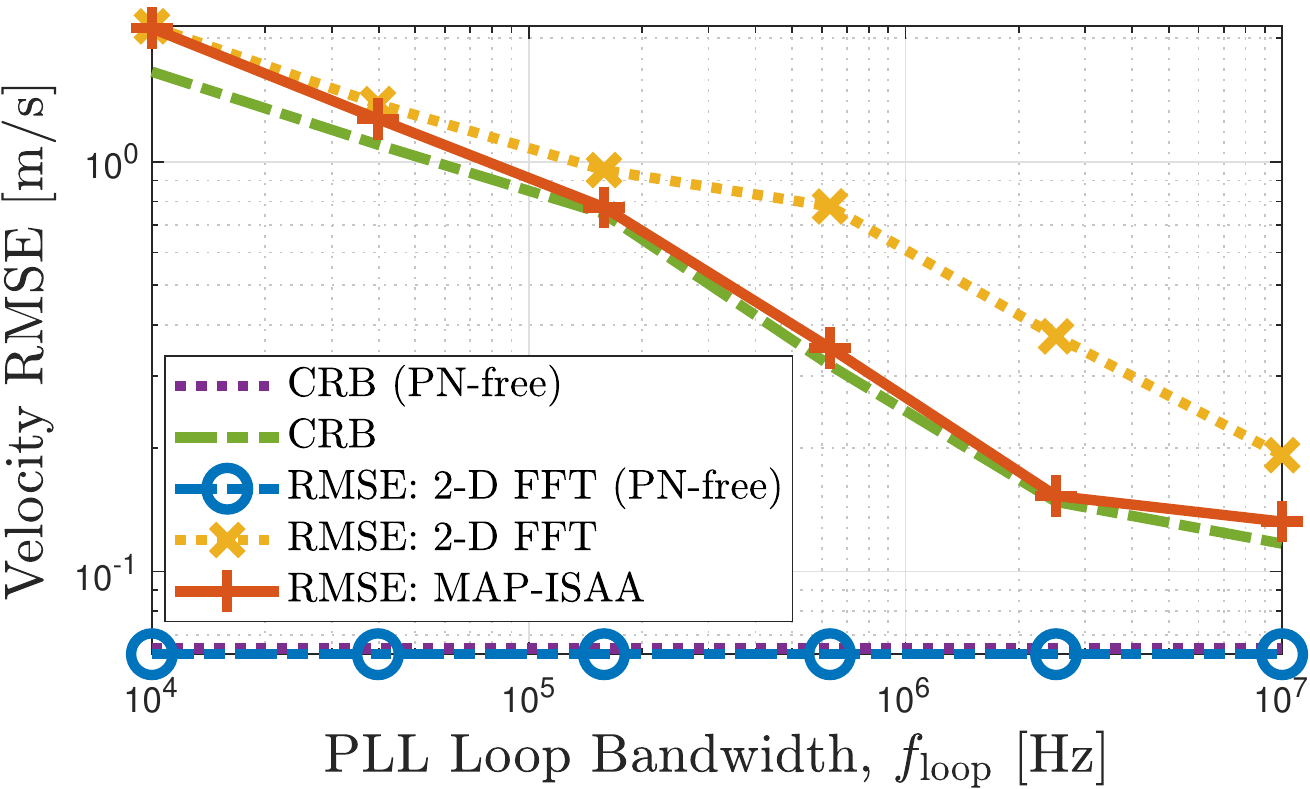}
		}
		
		\subfigure[]{
			 \label{fig_pn_rmse_wrt_floop_PLL}
			 \includegraphics[width=0.41\textwidth]{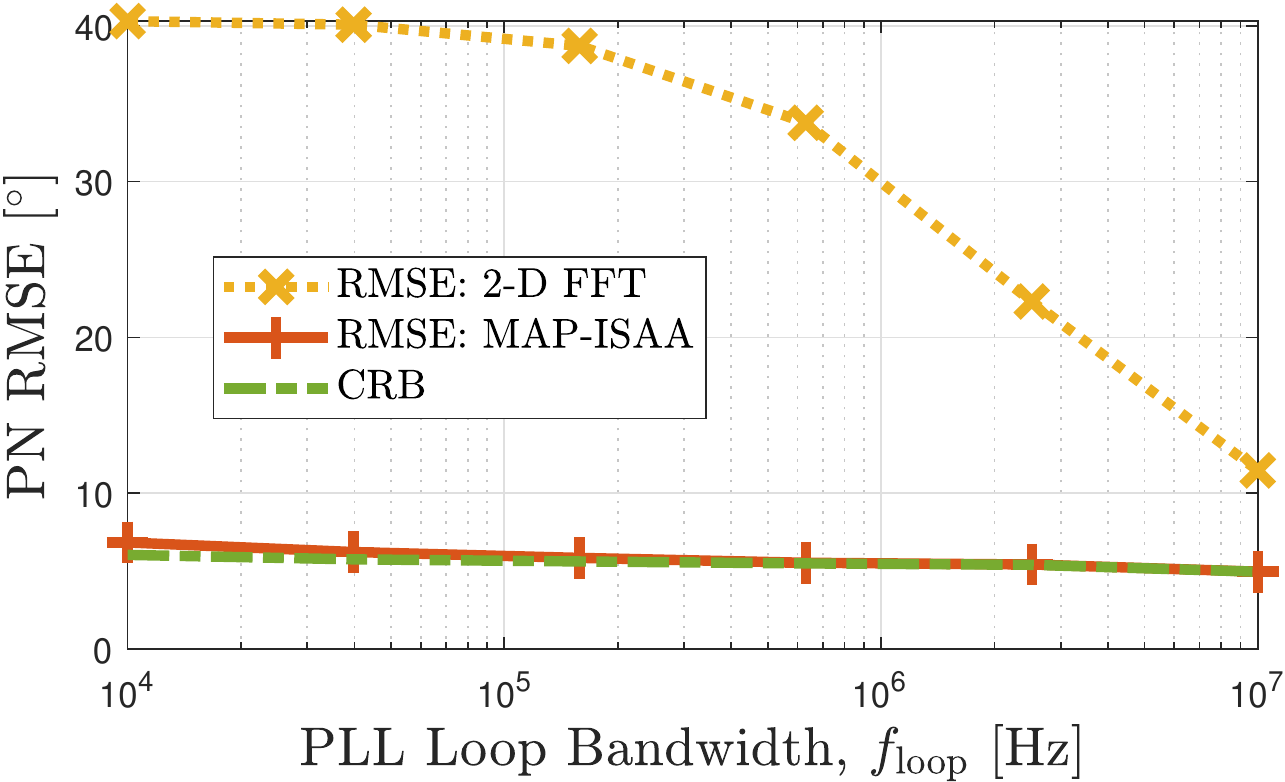}
		}
		
		\end{center}
		\vspace{-0.2in}
        \caption{\subref{fig_range_rmse_wrt_floop_PLL} Range, \subref{fig_vel_rmse_wrt_floop_PLL} velocity and \subref{fig_pn_rmse_wrt_floop_PLL} PN RMSE with respect to loop bandwidth $\floop$ of PLL with $\fdb = 200 \, \rm{kHz}$ at $\snr = 20 \, \rm{dB}$.
        } 
        \label{fig_rmse_wrt_floop_PLL}
        %\vspace{-0.12in}
        \vspace{-0.15in}
\end{figure}
%%%%%%%%%%%%%%%%%%%%%%%%%%%%%%%%%%%%%%%%%%

%%%%%%%%%%%%%%%%%%%%%%%%%%%%%%%%%%%%%%%%%%%%%%%%%%%%%%%%
\subsection{Performance with respect to Target Range}\label{sec_perf_range}
To further highlight the benefits of the proposed algorithm, we explore the impact of target range on the performance of the processing schemes under consideration, recalling the delay-dependent statistics of PN derived in Sec.~\ref{sec_pn_statistics}. In Fig.~\ref{fig_rmse_wrt_range_FRO} and Fig.~\ref{fig_rmse_wrt_range_PLL}, we show the range-velocity RMSEs versus target range for FRO and PLL, respectively. As can be noticed, the ranging performance of the proposed approach is robust against increasing target range for both FRO and PLL implementations, which suggest that Algorithm~\ref{alg_map_isaa} can successfully utilize the knowledge of delay-dependent PN covariance $\boldR(\tau)$ to jointly estimate the coupled delay and PN parameters. However, the FFT method experiences substantial loss in ranging accuracy as target moves further away from radar, in compliance with monotonically increasing variance of PN as a function of delay in \eqref{eq_sigmatauwn} and \eqref{eq_sigmatauwn_pll}. Moreover, comparing MAP-ISAA and CRB against the PN-free benchmark, PN-induced performance degradation is only marginal, which agrees with the observations in Sec.~\ref{sec_perf_snr} and Sec.~\ref{sec_perf_osc_qual}. Furthermore, looking at the velocity RMSEs in Fig.~\ref{fig_vel_rmse_wrt_range_FRO} and Fig.~\ref{fig_vel_rmse_wrt_range_PLL}, we observe that improvements on the order of $\rm{m/s}$ can be provided by the proposed algorithm over the FFT benchmark for PLL, whereas no noticeable gains occur in case of FRO. This further corroborates the above-mentioned insights on the impact of different PN correlation characteristics of FRO and PLL (specified in \eqref{eq_toep_block} and \eqref{eq_R_blkdiag}) onto velocity accuracy.

%%%%%%%%%%%%%%%%%%%%%%%%%%%%%%%%%%%%%%%%%%
% range and vel RMSE for FRO with respect to range
\begin{figure}[t]
        \begin{center}
        %\vspace{-0.22in}
        \subfigure[]{
			 \label{fig_range_rmse_wrt_range_FRO}
			 \includegraphics[width=0.42\textwidth]{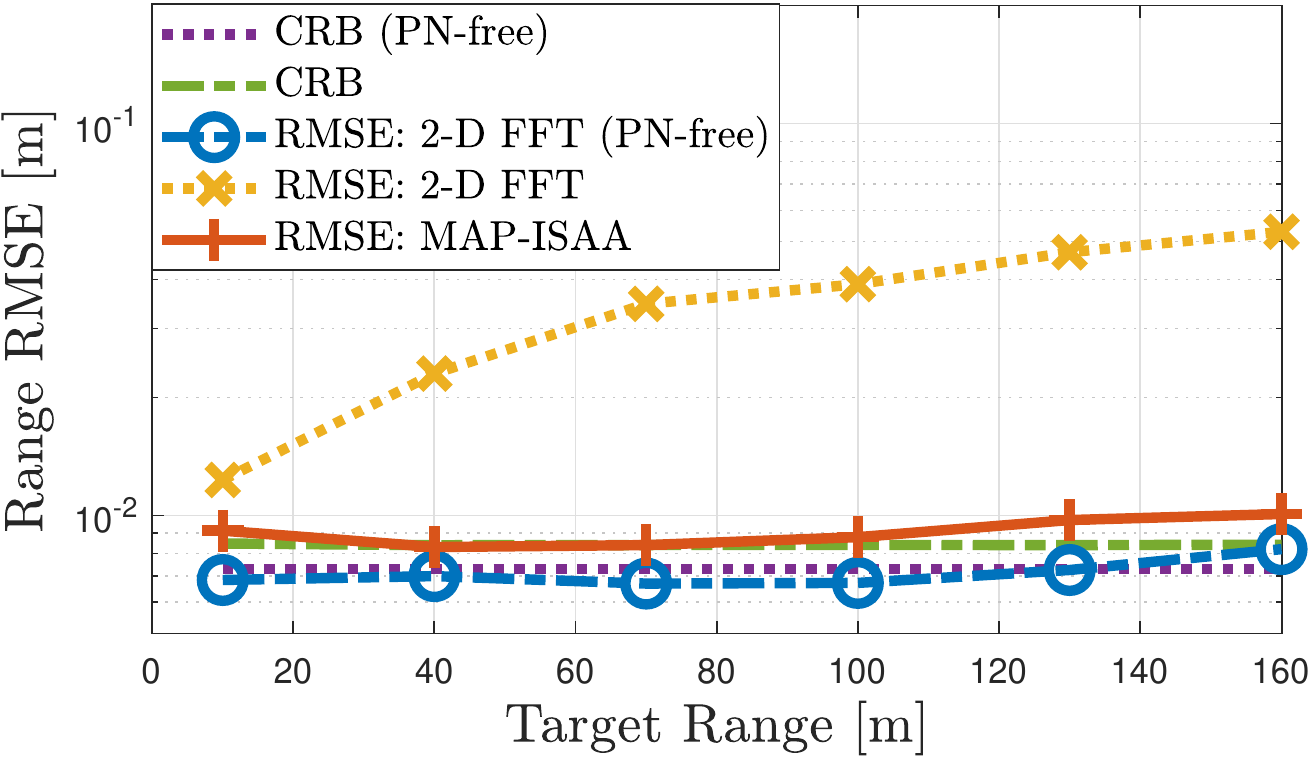}
			 %\vspace{-1in}
			 %\vspace{-5em}
			 %\hspace{-0.5in}
		}
        %\hfill 
        \subfigure[]{
			 \label{fig_vel_rmse_wrt_range_FRO}
			 \includegraphics[width=0.42\textwidth]{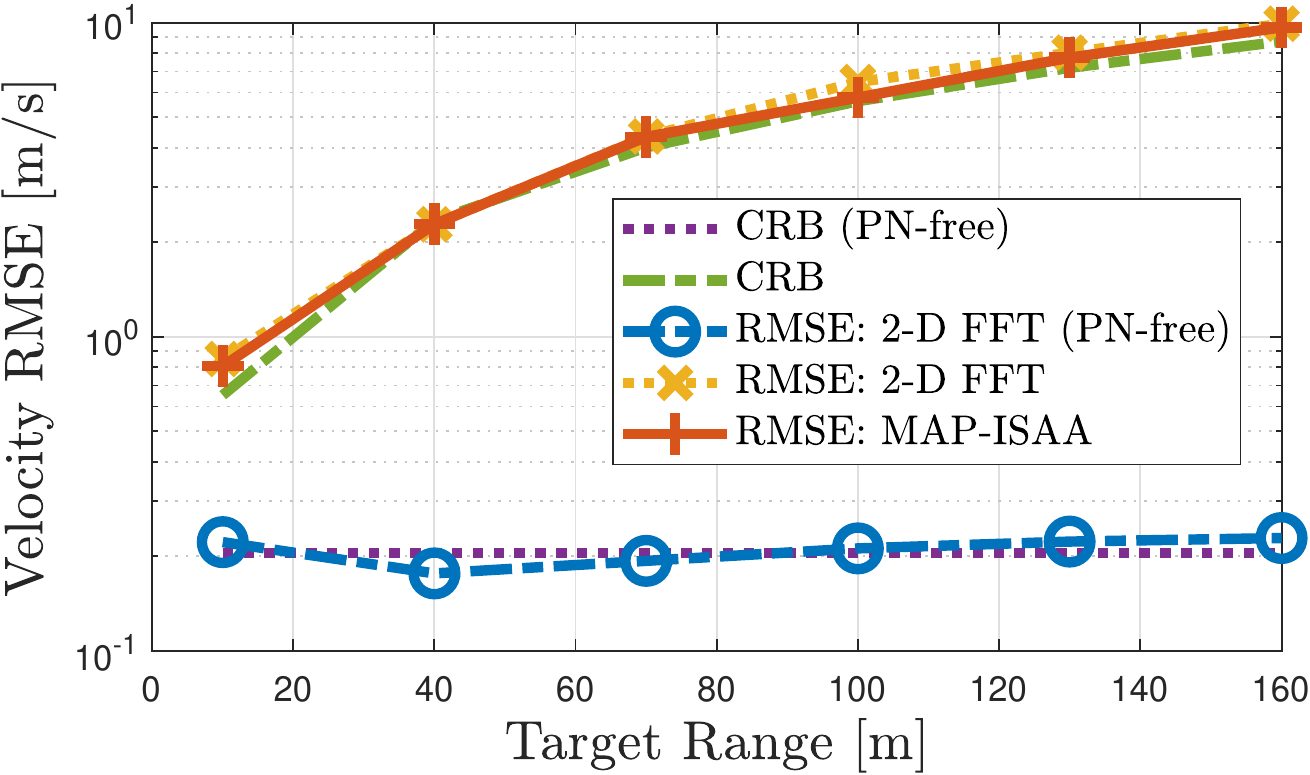}
		}

		\end{center}
		\vspace{-0.2in}
        \caption{\subref{fig_range_rmse_wrt_range_FRO} Range, and \subref{fig_vel_rmse_wrt_range_FRO} velocity RMSE with respect to target range for FRO with $\fdb = 150 \, \rm{kHz}$ at $\snr = 10 \, \rm{dB}$.} 
        \label{fig_rmse_wrt_range_FRO}
        %\vspace{-0.12in}
        \vspace{-0.1in}
\end{figure}
%%%%%%%%%%%%%%%%%%%%%%%%%%%%%%%%%%%%%%%%%%

%%%%%%%%%%%%%%%%%%%%%%%%%%%%%%%%%%%%%%%%%%
% range and vel RMSE for PLL with respect to range
\begin{figure}[t]
        \begin{center}
        %\vspace{-0.22in}
        \subfigure[]{
			 \label{fig_range_rmse_wrt_range_PLL}
			 \includegraphics[width=0.42\textwidth]{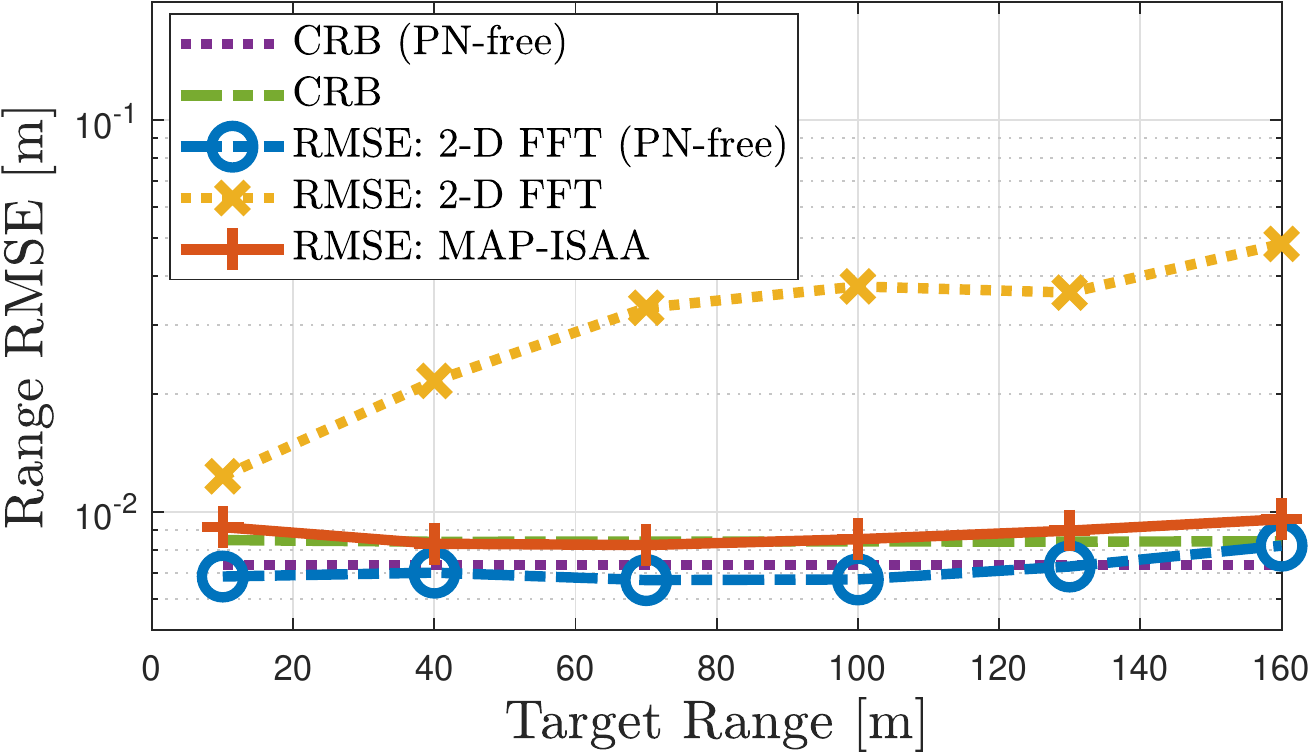}
			 %\vspace{-1in}
			 %\vspace{-5em}
			 %\hspace{-0.5in}
		}
        %\hfill 
        \subfigure[]{
			 \label{fig_vel_rmse_wrt_range_PLL}
			 \includegraphics[width=0.42\textwidth]{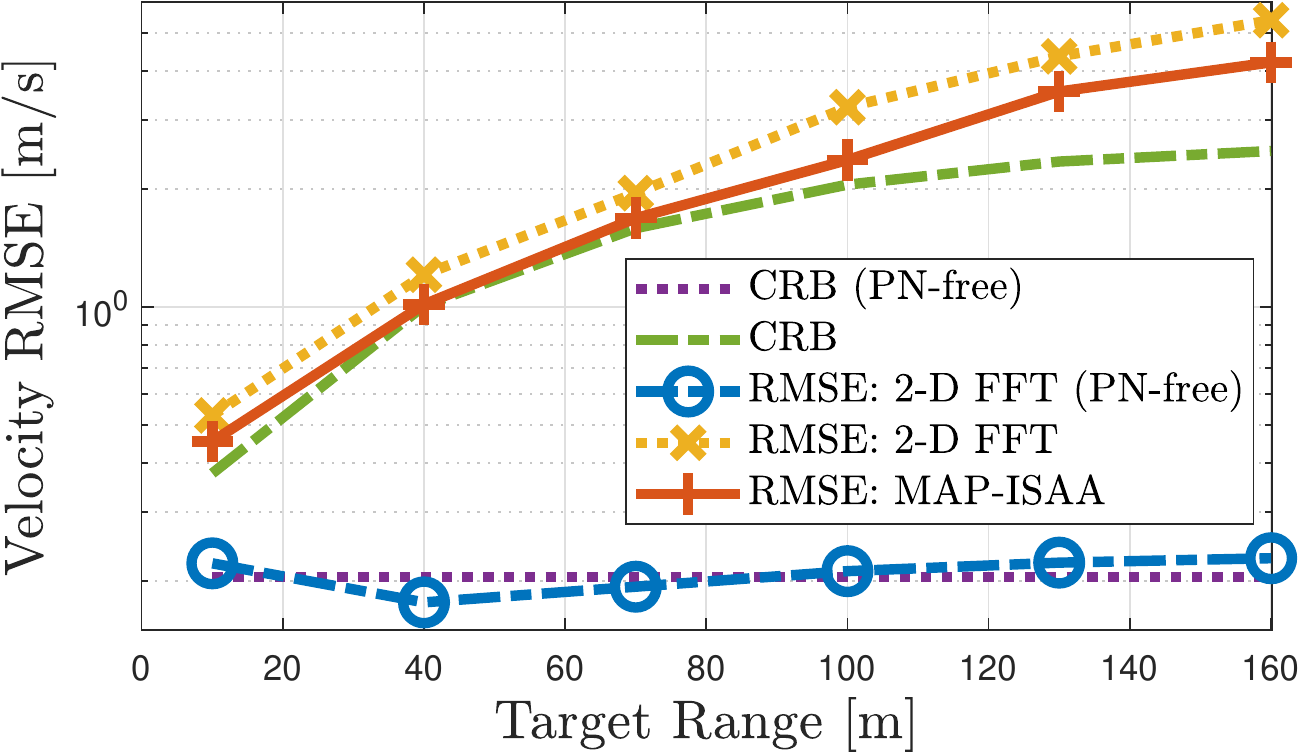}
		}

		\end{center}
		\vspace{-0.2in}
        \caption{\subref{fig_range_rmse_wrt_range_PLL} Range, and \subref{fig_vel_rmse_wrt_range_PLL} velocity RMSE with respect to target range for PLL with $\fdb = 150 \, \rm{kHz}$ and $\floop = 100 \, \rm{kHz}$ at $\snr = 10 \, \rm{dB}$.} 
        \label{fig_rmse_wrt_range_PLL}
        %\vspace{-0.12in}
        \vspace{-0.1in}
\end{figure}
%%%%%%%%%%%%%%%%%%%%%%%%%%%%%%%%%%%%%%%%%%

%%%%%%%%%%%%%%%%%%%%%%%%%%%%%%%%%%%%%%%%%%%%%%%%%%%%%%%%
\subsection{Convergence Behavior of the Proposed Algorithm}\label{sec_conv_alg}
To illustrate the convergence behavior of the proposed MAP-ISAA algorithm in Algorithm~\ref{alg_map_isaa} through iterated approximations, Fig.~\ref{fig_iterations_rmse_PLL_10dB} demonstrates the evolution of range, velocity and PN RMSEs over consecutive iterations for PLL. It is observed that starting from the output of the FFT benchmark, Algorithm~1 monotonically converges to the corresponding CRBs on range, velocity and PN estimation within few iterations, which proves the effectiveness of the proposed iterated SAA approach. To explore the effect of iterated approximations on PN tracking performance, in Fig.~\ref{fig_pn_tracking} we depict instances from the PN process along with the PN estimates at the first and last iteration of Algorithm~\ref{alg_map_isaa} for FRO and PLL architectures\rev{\footnote{\label{fn_wrapping}\rev{Phase wrapping in PN trajectories (i.e., phases outside $[-\pi, \ \pi]$) occurs very rarely and does not affect the resulting RMSE performances.}}}. Comparing the results of the first and last iterations, we see that by applying iterated approximations in Algorithm~\ref{alg_map_isaa}, PN tracking accuracy significantly improves around regions with high fluctuations, which verifies the superiority of the proposed iterative refinement approach.
%and indicates that the conventional SAA approach used in OFDM communications (e.g., \cite{OFDM_Joint_PHN_TSP_2006,OFDM_Joint_PHN_TVT_2009}), which corresponds to the first iteration of MAP-ISAA, would perform poorly in OFDM monostatic sensing scenarios.

%%%%%%%%%%%%%%%%%%%%%%%%%%%%%%%%%%%%%%%%%%
% RMSE iterations for PLL at SNR = 20 dB
\begin{figure}
	\centering
    %\vspace{-0.2in}
	\includegraphics[width=0.95\linewidth]{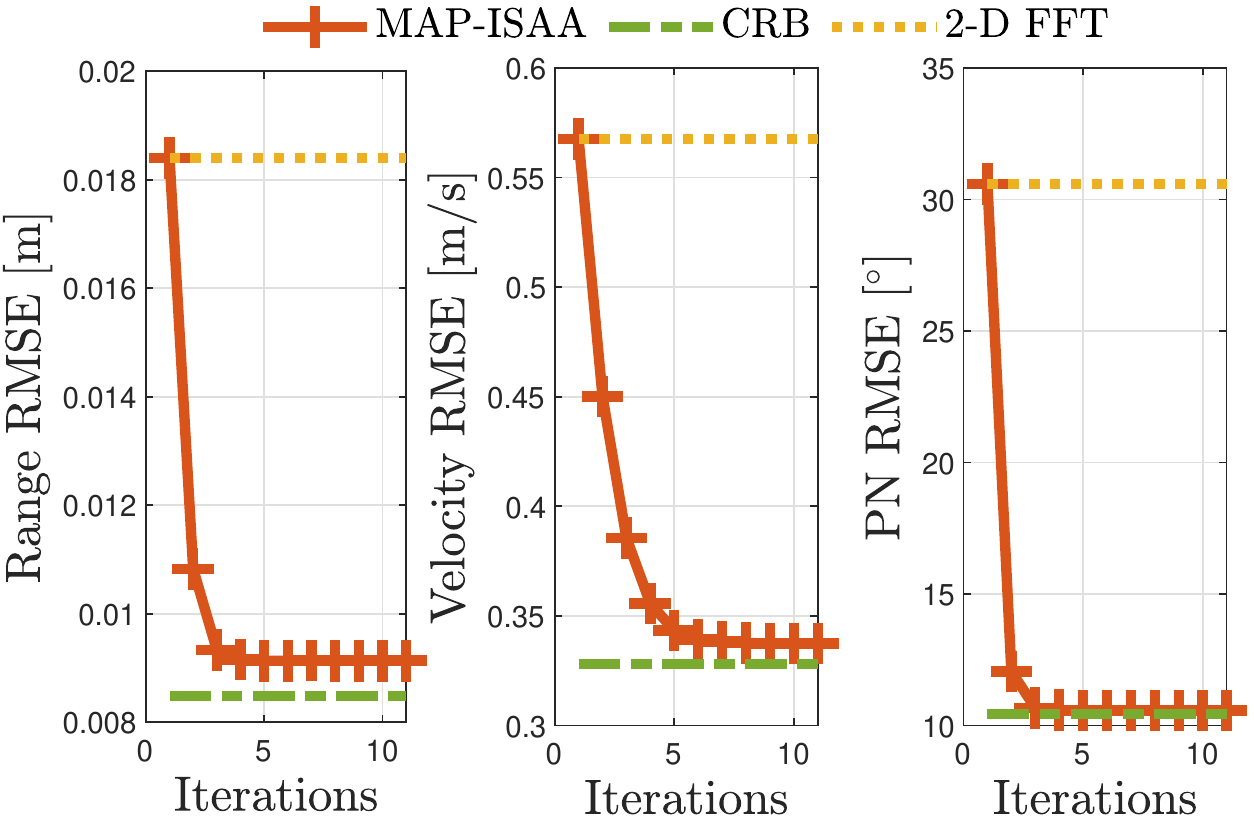}
	\vspace{-0.05in}
	\caption{Evolution of range, velocity and PN RMSEs through successive iterations of the proposed algorithm at $\snr = 10 \, \rm{dB}$ for PLL with $\fdb = 200 \, \rm{kHz}$ and $\floop = 1 \, \rm{MHz}$. Initialized at the FFT output, the proposed approach can quickly converge to the corresponding CRBs in few iterations, recovering the performance loss incurred by the PN.}
	\label{fig_iterations_rmse_PLL_10dB}
	\vspace{-0.1in}
\end{figure}
%%%%%%%%%%%%%%%%%%%%%%%%%%%%%%%%%%%%%%%%%%

%%%%%%%%%%%%%%%%%%%%%%%%%%%%%%%%%%%%%%%%%%
% PN tracking for FRO and PLL at SNR = 10 dB
\begin{figure}
        \begin{center}
        %\vspace{-0.22in}
        \subfigure[]{
			 \label{fig_pn_tracking_FRO}
			 \includegraphics[width=0.45\textwidth]{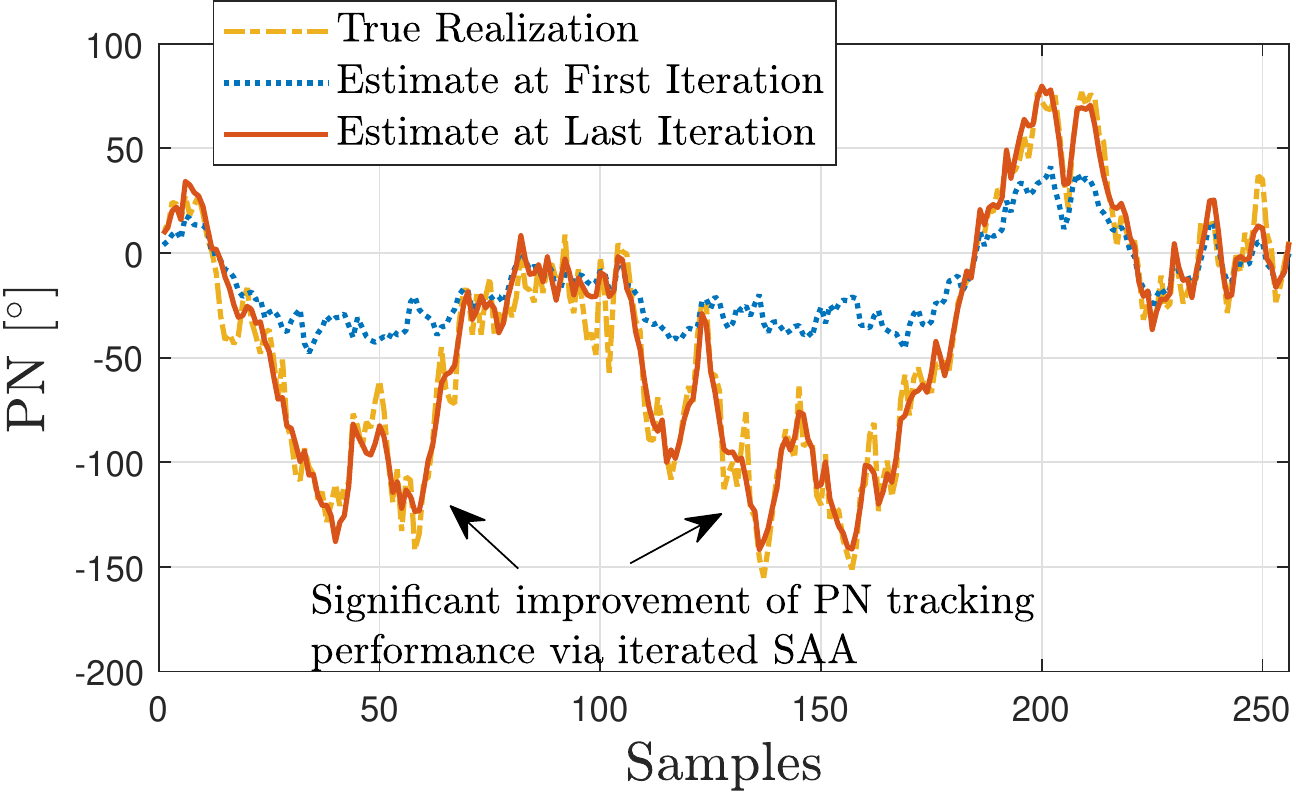}
			 %\vspace{-1in}
			 %\vspace{-5em}
			 %\hspace{-0.5in}
		}
        %\hfill 
        \subfigure[]{
			 \label{fig_pn_tracking_PLL}
			 \includegraphics[width=0.45\textwidth]{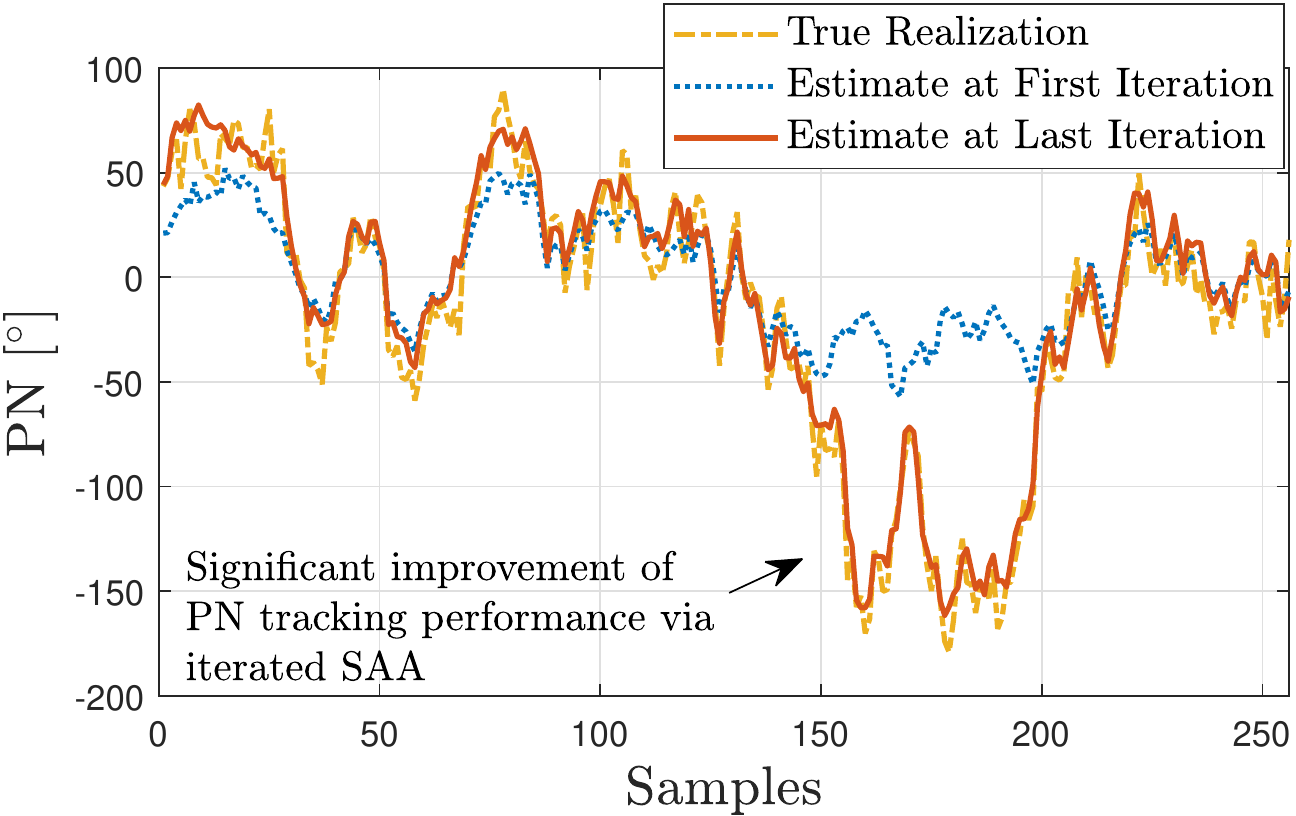}
		}

		\end{center}
		\vspace{-0.2in}
        \caption{True PN realization and PN estimates at the first and last iteration of the proposed algorithm at $\snr = 10 \, \rm{dB}$ for \subref{fig_pn_tracking_FRO} FRO with $\fdb = 150 \, \rm{kHz}$ considering a target with range $R = 100 \, \rm{m}$, and \subref{fig_pn_tracking_PLL} PLL with $\fdb = 150 \, \rm{kHz}$ and $\floop = 100 \, \rm{kHz}$ considering a target with range $R = 130 \, \rm{m}$. The proposed iterated SAA approach provides substantial improvements in PN trajectory tracking performance especially when large fluctuations occur.} 
        \label{fig_pn_tracking}
        %\vspace{-0.12in}
        \vspace{-0.15in}
\end{figure}
%%%%%%%%%%%%%%%%%%%%%%%%%%%%%%%%%%%%%%%%%%

%%%%%%%%%%%%%%%%%%%%%%%%%%%%%%%%%%%%%%%%%%%%%%%%%%%%%%%%
\subsection{PN Exploitation Capability of the Proposed Algorithm}\label{sec_pn_exp}
In this part, we demonstrate the PN exploitation capability of the proposed method in Algorithm~\ref{alg_resolve_amb} by considering a target located at $R = 1000 \, \rm{m}$, which is beyond the maximum unambiguous range of $768 \, \rm{m}$ according to Table~\ref{tab_parameters}. In Fig.~\ref{fig_range_rmse_pn_exploitation_pll}, we plot the range RMSE vs. SNR for PLL with $\fdb = 20 \, \rm{kHz}$ and $\floop = 1 \, \rm{MHz}$ by searching for the first two ambiguity intervals in \eqref{eq_set_delay}, i.e., $K=1$. In addition, Fig.~\ref{fig_covarianceProfilePNExploitation_SNR_20dB} shows the PN covariances at $\snr = 20 \, \rm{dB}$, obtained via the estimate $\rrhat$ in \eqref{eq_rhat} and via the analytical evaluation $[\boldR(\tau)]_{0,:}$ in \eqref{eq_parametric_Frob3} at the ambiguous and true ranges. It can be seen that the proposed PN exploitation approach can successfully resolve the range ambiguity starting from $\snr = 15 \, \rm{dB}$ (when the PN estimate becomes sufficiently accurate\rev{\footnote{\label{fn_pn_ambiguity}\rev{Additional simulations with different $\fdb$ values for PLL offer an intriguing insight into the relation between $\fdb$ and the degree of accuracy improvement via PN exploitation: higher $\fdb$ (which means lower oscillator quality) provides improved resolvability of range ambiguity at low SNRs through more pronounced notches in the PN covariance profile at the ambiguous and true ranges, while leading to lower accuracy at high SNRs due to performance saturation (i.e., when the ambiguity is already resolved).}}}) and attain the corresponding CRB, whereas neither of the FFT benchmarks can correctly identify the true target range due to intrinsic range ambiguity in $\bb(\tau)$, in compliance with the discussions in Sec.~\ref{sec_extension}. More specifically, the FFT processing in \eqref{eq_fft_benc}, either using PN-impaired observations $\boldY$ in \eqref{eq_ym_all_multi} or PN-free observations $\boldYfr$ in \eqref{eq_ym_special}, can extract the range information only from $\bb(\tau)$, which causes ambiguity for $\tau > T$ (see \eqref{eq_steer_delay}). On the other hand, to estimate the true (unambiguous) range of the target, Algorithm~\ref{alg_resolve_amb} can effectively exploit the ranging information revealed by the PN statistics $\boldR(\tau)$ in \eqref{eq_bxi_stat}, which does not introduce any ambiguity in estimating $\tau$ as noticed from \eqref{eq_exp_xi2}--\eqref{eq_rtau_entries}. Surprisingly, the proposed approach under the impact of PN can significantly outperform the FFT benchmark that uses ideal, PN-free observations, indicating its excellent PN exploitation performance.

%Remarkably, through Fig.~\ref{fig_range_rmse_pn_exploitation_pll}, it is observed that the proposed method can turn the detrimental effect of PN into something beneficial for sensing.

%%%%%%%%%%%%%%%%%%%%%%%%%%%%%%%%%%%%%%%%%%
% PN exploitation
\begin{figure}
        \begin{center}
        %\vspace{-0.22in}
        \subfigure[]{
			 \label{fig_range_rmse_pn_exploitation_pll}
			 \includegraphics[width=0.45\textwidth]{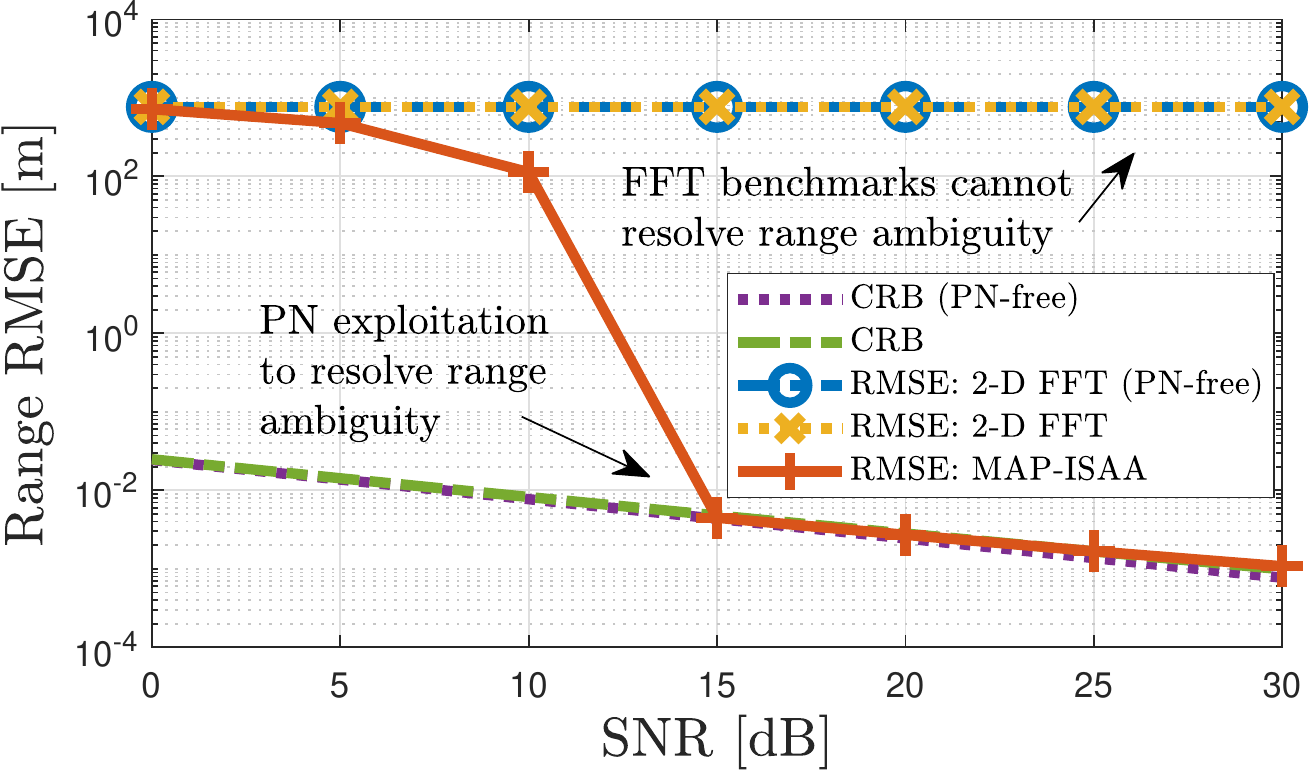}
			 %\vspace{-1in}
			 %\vspace{-5em}
			 %\hspace{-0.5in}
		}
        %\hfill 
        \subfigure[]{
			 \label{fig_covarianceProfilePNExploitation_SNR_20dB}
			 \includegraphics[width=0.45\textwidth]{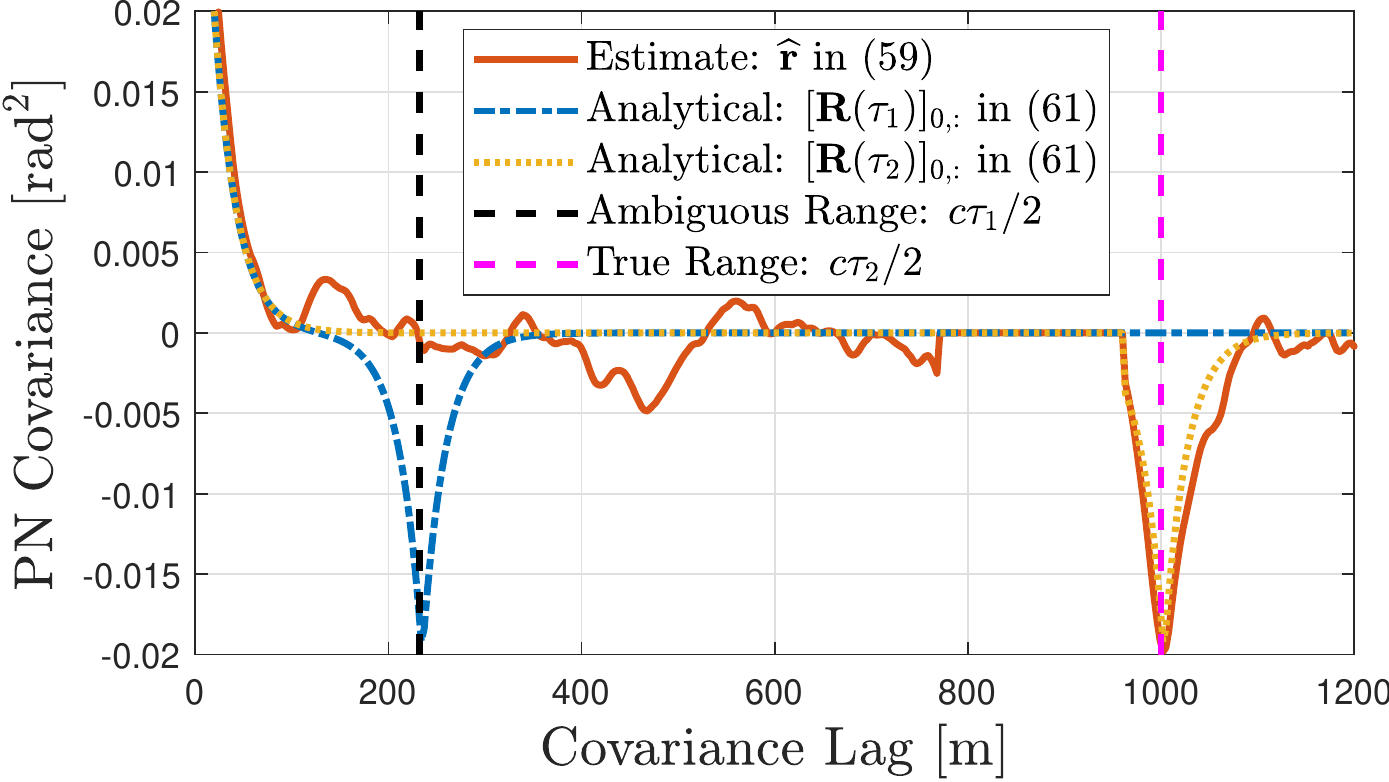}
		}

		\end{center}
		\vspace{-0.2in}
        \caption{PN exploitation to resolve range ambiguity for PLL with $\fdb = 20 \, \rm{kHz}$ and $\floop = 1 \, \rm{MHz}$, where target is located at $R = 1000 \, \rm{m}$. \subref{fig_range_rmse_pn_exploitation_pll} Range RMSE with respect to SNR. \subref{fig_covarianceProfilePNExploitation_SNR_20dB} PN covariance at $\snr = 20 \, \rm{dB}$, corresponding to the estimate $\rrhat$ in \eqref{eq_rhat} and the analytical expression $[\boldR(\tau)]_{0,:}$ in \eqref{eq_parametric_Frob3} evaluated at $\tau_1 = \tauhat$ and $\tau_2 = \tauhat + T$, where $\tauhat$ is the ambiguous delay estimate. The proposed approach in Algorithm~\ref{alg_resolve_amb} can correctly identify the true target range by exploiting PN statistics ($\rrhat$ matches with $[\boldR(\tau_2)]_{0,:}$), while the FFT benchmark fails to resolve range ambiguity irrespective of whether the PN is present or absent.} 
        \label{fig_pn_exploitation}
        %\vspace{-0.12in}
        \vspace{-0.05in}
\end{figure}
%%%%%%%%%%%%%%%%%%%%%%%%%%%%%%%%%%%%%%%%%%

%%%%%%%%%%%%%%%%%%%%%%%%%%%%%%%%%%%%%%%%%%%%%%%%%%%%%%%%
%%%%%%%%%%%%%%%%%%%%%%%%%%%%%%%%%%%%%%%%%%%%%%%%%%%%%%%%
\section{Concluding Remarks}\label{sec_conc}
In this work, we have investigated the problem of monostatic radar sensing in OFDM JRC systems under the effect of PN. Starting from an explicit derivation of PN statistics in the OFDM radar receiver, we have proposed a novel algorithm for joint estimation of delay, Doppler and PN by devising an iterated small angle approximation approach that solves the hybrid ML/MAP optimization problem through alternating updates of delay-Doppler and PN estimates. In addition, we have developed a PN exploitation algorithm that uses the statistics of the PN estimates to resolve range ambiguity. To assess the performance of the proposed algorithms, comprehensive simulations have been carried out over a broad range of operating conditions, leading to the following key findings:
%regarding oscillator type (FRO and PLL), oscillator quality ($3 \, \rm{dB}$ bandwidth and loop bandwidth), SNR and target range. The simulation results indicate the following key findings:
\begin{itemize}
    \item \textit{Range vs. Velocity Accuracy under PN:} \rev{The proposed MAP-ISAA algorithm significantly outperforms the FFT benchmark in ranging accuracy, while providing slight improvements in velocity accuracy. This} results from strong (weak) correlation of PN in fast-time (slow-time). 
    \item \textit{Impact of Oscillator Type - FRO vs. PLL:} FROs and PLLs lead to very similar performance in range estimation, whereas PLLs can provide much higher velocity accuracy than FROs. This is due to the absence of PN correlation in slow-time for FROs.
    \item \rev{\textit{Scenarios with Significant Impact of PN on Sensing:} The effect of PN on sensing performance becomes more significant at higher SNRs, for oscillators with larger $\fdb$ and smaller $\floop$ values, and farther targets. In such scenarios, PN should be considered in algorithm design.}
    \item \textit{Iterated Small Angle Approximation:} The proposed algorithm enjoys fast convergence to the hybrid CRB and provides superior PN tracking performance via iterated approximations. 
    \item \textit{PN Exploitation:} Delay-dependency of PN statistics can be effectively exploited to resolve range ambiguity above a certain SNR level, thereby converting the detrimental effect of PN into something beneficial for radar sensing.
\end{itemize}
As future research, we plan to investigate extensions of the proposed framework to MIMO architectures and \rev{the multi-target case.}
\vspace{-0.05in}

%%%%%%%%%%%%%%%%%%%%%%%%%%%%%%%%%%%%%%%%%%%%%%%%%%%%%%%%
%%%%%%%%%%%%%%%%%%%%%%%%%%%%%%%%%%%%%%%%%%%%%%%%%%%%%%%%
%%%%%%%%%%%%%%%%%%%%%%%%%%%%%%%%%%%%%%%%%%%%%%%%%%%%%%%%
%%%%%%%%%%%%%%%%%%%%%%%%%%%%%%%%%%%%%%%%%%%%%%%%%%%%%%%%
\begin{appendices}

\section{Correlation Function of Differential Phase Noise in \eqref{eq_exp_xi2}}\label{app_corr_dpn}
   Let the correlation function of $\phi(t)$ in \eqref{eq_pn_stat} be defined as
   \begin{align}
        \rphii(t_1,t_2) &\triangleq \E \left[  \phi(t_1)  \phi(t_2) \right] ~.
   \end{align}
   Then, it follows from \eqref{eq_dpn} that 
\begin{align}\label{eq_exp_xi}
    \E \left[  \xi(t_1, \tau)  \xi(t_2, \tau) \right] & = \rphii(t_1,t_2) - \rphii(t_1, t_2 - \tau) \\ \nonumber
    &- \rphii(t_1-\tau, t_2) + \rphii(t_1-\tau, t_2-\tau) ~.
\end{align}
From \eqref{eq_dpn} and \eqref{eq_dpn_stat}, we have
\begin{align}
    \sigmatau = \E \left[  \xi(t, \tau)^2 \right] = \sigmaphi(t) + \sigmaphi(t-\tau) - 2 \rphii(t, t-\tau) ~,
\end{align}
which yields 
\begin{align}\label{eq_rphii}
    \rphii(t_1, t_2) = \frac{ \sigmaphi(t_1) + \sigmaphi(t_2) - \sigmaxi(t_1-t_2)  }{2} ~.
\end{align}
Finally, inserting \eqref{eq_rphii} into \eqref{eq_exp_xi} yields
\begin{align}\label{eq_exp_xi3}
    &\E \left[  \xi(t_1, \tau)  \xi(t_2, \tau) \right] \\ \nonumber &= \frac{ \sigmaxi(t_1-t_2+\tau) + \sigmaxi(t_1-t_2-\tau) }{2} - \sigmaxi(t_1-t_2) ~.
\end{align}
Defining $\deltat \triangleq t_1 - t_2$ in \eqref{eq_exp_xi3} and considering $\sigmaxi(\tau) = \sigmaxi(-\tau)$, we obtain the result in \eqref{eq_exp_xi2}.

%%%%%%%%%%%%%%%%%%%%%%%%%%%%%%
%%%%%%%%%%%%%%%%%%%%%%%%%%%%%%
\section{Block-Diagonal Structure of $\boldR(\tau)$ for FROs}\label{app_blk_diag}
   
   Using \eqref{eq_sigmatauwn} in \eqref{eq_exp_xi2}, we obtain the correlation function for FROs as
\begin{align} \label{eq_rxii_fro}
     \rxii(\deltat, \tau) &=  4 \pi \fdb \max \left( \tau - \abs{\deltat}, 0 \right) ~.
\end{align}
Since $\abs{\deltat^{(m)}_{n_1 n_2}} \geq \Tcp$ for $m>0$ in \eqref{eq_ind_blocks}, we have $\rxii(\deltat^{(m)}_{n_1 n_2}, \tau) = 0$ for $m>0$ and $\tau \leq \Tcp$ from \eqref{eq_rxii_fro}, leading to $\boldR_m(\tau) = \boldzero$ for $m>0$, which yields \eqref{eq_R_blkdiag}.

%%%%%%%%%%%%%%%%%%%%%%%%%%%%%%
%%%%%%%%%%%%%%%%%%%%%%%%%%%%%%%%%%%%%%%%%
\section{Obtaining \eqref{eq_nll_original_saa2_2} From \eqref{eq_nll_original_saa2}} \label{app_obtain}
In this part, we provide details on how to reach \eqref{eq_nll_original_saa2_2} from \eqref{eq_nll_original_saa2}. To this end, we present the following auxiliary lemmas.

%%%%% Lemma S1
\begin{lemma}\label{lemma_s1}
For any $\xx \in \realset{N}{1}$ and any Hermitian matrix $\boldA \in \complexset{N}{N}$, the following equality holds:
\begin{align} \label{eq_xAx_lemma}
    \xx^T \boldA \xx = \xx^T \realp{\boldA} \xx ~.
\end{align}
\end{lemma}
\begin{proof}
Since $\xx^T \boldA \xx$ is scalar, we can write
\begin{align}\label{eq_lemma_s1}
    &\xx^T \boldA \xx = (\xx^T \boldA \xx)^T = \xx^T \boldA^T \xx
    \labelrel={eq_hermit_s1} \xx^T \boldA^\conj \xx ~,
\end{align}
where \eqref{eq_hermit_s1} is due to $ \boldA$ being Hermitian, i.e., $ \boldA^H =  \boldA$. Re-writing the leftmost and rightmost terms in \eqref{eq_lemma_s1}, we obtain
\begin{align} \label{eq_xax}
    &\xx^T \boldA \xx = \xx^T \left( \realp{\boldA} + j \imp{\boldA} \right) \xx
    \\ \nonumber
    &=\xx^T \boldA^\conj \xx = \xx^T \left( \realp{\boldA} - j \imp{\boldA} \right) \xx ~,
\end{align}
which leads to
\begin{align} \label{eq_ximx}
    \xx^T \imp{\boldA} \xx = 0 ~.
\end{align}
Substituting \eqref{eq_ximx} into \eqref{eq_xax} yields the result in \eqref{eq_xAx_lemma}. $\blacksquare$
\end{proof}

%%%%% Lemma S2
\begin{lemma}\label{lemma_s2}
For any $\xx \in \realset{N}{1}$, $\yy \in \realset{N}{1}$ and any Hermitian matrix $\boldA \in \complexset{N}{N}$, the following equality holds:
\begin{align} \label{eq_xAx_lemma}
    j(\xx^T \boldA \yy - \yy^T \boldA \xx) = 2 \yy^T \imp{\boldA} \xx ~.
\end{align}
\end{lemma}
\begin{proof}
Since $\xx^T \boldA \yy$ is scalar and $\boldA$ is Hermitian, we can write
\begin{align*}
    &j(\xx^T \boldA \yy - \yy^T \boldA \xx)
    = j(\yy^T \boldA^T \xx - \yy^T \boldA \xx)
    \\ 
    &= j(\yy^T \boldA^\conj \xx - \yy^T \boldA \xx)
    \\
    &= j \yy^T (\boldA^\conj - \boldA) \xx
    \\
    &= -j \yy^T 2j \imp{\boldA} \xx
    \\
    &= 2 \yy^T \imp{\boldA} \xx ~,
\end{align*} 
which completes the proof. $\blacksquare$
\end{proof}

%%%%%%%%%%%%%%%%%%%%%%%%%%%%
Opening up the first term in \eqref{eq_nll_original_saa2}, we have
\begin{align} \label{eq_lemma_s10}
    & (\boldone - j \bxidelta)^T \boldD^{\rev{(i)}}(\tau, \nu)  (\boldone + j \bxidelta)
    \\ \label{eq_lemma_s11}
    &= \boldone^T \boldD^{\rev{(i)}}(\tau, \nu)\boldone + \bxidelta^T \boldD^{\rev{(i)}}(\tau, \nu)\bxidelta 
    \\ \nonumber
    &~~ + j \left( \boldone^T \boldD^{\rev{(i)}}(\tau, \nu)\bxidelta -  \bxidelta^T \boldD^{\rev{(i)}}(\tau, \nu) \boldone \right)
    \\  \label{eq_lemma_s12}
    &= \boldone^T \realp{\boldD^{\rev{(i)}}(\tau, \nu)} \boldone + \bxidelta^T \realp{\boldD^{\rev{(i)}}(\tau, \nu)} \bxidelta
    \\ \nonumber &~~+ 2 \bxidelta^T \imp{\boldD(\tau, \nu} \boldone ~,
\end{align}
where the transition from \eqref{eq_lemma_s11} to \eqref{eq_lemma_s12} is via Lemma~\ref{lemma_s1} and Lemma~\ref{lemma_s2}. Plugging \eqref{eq_lemma_s10}--\eqref{eq_lemma_s12} into \eqref{eq_nll_original_saa2} yields \eqref{eq_nll_original_saa2_2}.

%%%%%%%%%%%%%%%%%%%%%%%%%%%%%%
%%%%%%%%%%%%%%%%%%%%%%%%%%%%%%%%%%%%%%%%%
\section{Delay-Doppler Estimation in \eqref{eq_hybrid_ml_map_doppler3_fixed}}\label{app_delayDoppler}
Using \eqref{eq_xi}, \eqref{eq_qq_norm} and \eqref{eq_hybrid_ml_map_doppler4} in \eqref{eq_hybrid_ml_map_doppler3_fixed}, we obtain
\begin{align} \nonumber
    & \llr(\tau, \nu, \bxihat) \\ \nonumber &= \frac{1}{\sigma^2} \left( e^{j \bxihat} \odot \yy \right)^H  \projnull{\qq(\tau, \nu)} \left( e^{j \bxihat} \odot \yy \right)  +  \bxihat^T \boldR(\tau)^{-1} \bxihat \\ \nonumber  &~~+ \log \det \boldR(\tau)
    \\ \nonumber
    &= \frac{1}{\sigma^2} \left( e^{j \bxihat} \odot \yy \right)^H  \left( \Imatrix  - \frac{ \qq(\tau, \nu) \qq^H(\tau, \nu)}{\norm{\boldX}_F^2} \right) \left( e^{j \bxihat} \odot \yy \right)    \\ \nonumber  &~~+ \bxihat^T \boldR(\tau)^{-1} \bxihat  + \log \det \boldR(\tau)
    \\ \nonumber
    &= \frac{\norm{\yy}^2}{\sigma^2} - \frac{ \absbig{ \qq^H(\tau, \nu) \left( e^{j \bxihat} \odot \yy \right)}^2 }{ \sigma^2 \norm{\boldX}_F^2  } + \bxihat^T \boldR(\tau)^{-1} \bxihat \\ \label{eq_delayDoppler_min}   &~~  + \log \det \boldR(\tau) ~.
\end{align}
Based on \eqref{eq_delayDoppler_min}, the problem \eqref{eq_hybrid_ml_map_doppler3_fixed} is equivalent to
\begin{align}\nonumber
    (\tauhat, \nuhat)  &= \arg \max_{ \tau, \nu} ~ \frac{ \absbig{ \qq^H(\tau, \nu) \left( e^{j \bxihat} \odot \yy \right)}^2 }{ \sigma^2 \norm{\boldX}_F^2  } - \bxihat^T \boldR(\tau)^{-1} \bxihat \\ \label{eq_hybrid_ml_map_doppler3_fixed_eq} &~~- \log \det \boldR(\tau)  ~.
\end{align}
Using \eqref{eq_qtaunu} in the first term of \eqref{eq_hybrid_ml_map_doppler3_fixed_eq} and defining $\Wbhat \triangleq \veccinv{ e^{- j \bxihat}}$, we have
\begin{align}\label{eq_map_first}
      & \absbig{ \qq^H(\tau, \nu) \left( e^{j \bxihat} \odot \yy \right)}^2 \\ \nonumber
      &= \absbig{ \trace{ \left[\FF_N^{H} \Big(\boldX \odot \bb(\tau) \cc^{H}(\nu) \Big) \right]^H \left[ \Wbhat^\conj \odot \boldY \right]  } }^2 
      \\ \nonumber
      &= \absbig{ \trace{  \Big(\boldX \odot \bb(\tau) \cc^{H}(\nu) \Big)^H \FF_N \left[ \Wbhat^\conj \odot \boldY \right]  } }^2
      \\ \nonumber
      &= \absbig{ \trace{  \Big(\bb(\tau) \cc^{H}(\nu) \Big)^H \Big( \boldX^\conj \odot \FF_N \left[ \Wbhat^\conj \odot \boldY \right] \Big)  } }^2
      \\ \nonumber
      &= \absbig{ \trace{   \cc(\nu) \bb^H(\tau) \Big( \boldX^\conj \odot \FF_N \left[ \Wbhat^\conj \odot \boldY \right] \Big)  } }^2
      \\ \nonumber
      &= \absbig{ \trace{   \bb^H(\tau) \Big( \boldX^\conj \odot \FF_N \left[ \Wbhat^\conj \odot \boldY \right] \Big) \cc(\nu)   } }^2
      \\ \nonumber
      &= \absbig{ \bb^H(\tau) \left( \boldX^\conj \odot \FF_N \big( \Wbhat^\conj \odot \boldY \big)  \right) \cc(\nu) }^2~,
\end{align}
which, after being inserted into \eqref{eq_hybrid_ml_map_doppler3_fixed_eq}, results in \eqref{eq_llrmax}.

%%%%%%%%%%%%%%%%%%%%%%%%%%%%%%%%%%%%%%%%%
%%%%%%%%%%%%%%%%%%%%%%%%%%%%%%%%%%%%%%%%%
\section{\rev{Convergence Analysis of Algorithm~\ref{alg_map_isaa}}}\label{sec_conv_analysis}
\rev{In this part, we provide the convergence analysis of the proposed ISAA algorithm in Algorithm~\ref{alg_map_isaa} to identify the conditions under which the algorithm converges to a stationary point of the hybrid ML/MAP optimization problem in \eqref{eq_hybrid_ml_map_doppler3}.}

\subsection{\rev{Formulation of Subproblems to Solve \eqref{eq_hybrid_ml_map_doppler3}}}\label{sec_subproblems}
\rev{In Algorithm~1, we decompose the original problem \eqref{eq_hybrid_ml_map_doppler3} into two subproblems, each optimizing over either delay-Doppler or PN while keeping the other variable fixed in an alternating fashion. Let $(\tauhat^{(i)}, \nuhat^{(i)}, \bxihat^{(i)})$ denote the delay, Doppler and PN estimates at the beginning of the $\thn{i}$ iteration of the alternating optimization procedure.}

\subsubsection{\rev{Subproblem for PN Estimation}}
\rev{The subproblem of \eqref{eq_hybrid_ml_map_doppler3} for PN estimation at the $\thn{i}$ iteration can be cast as
\begin{align}\label{eq_subprob_pn_supp}
    \bxihat^{(i+1)}  = \arg \min_{\bxi} ~ \llr(\tauhat^{(i)}, \nuhat^{(i)}, \bxi) ~,
\end{align}
where $\llr(\tau, \nu, \bxi)$ is the hybrid ML/MAP objective function in \eqref{eq_hybrid_ml_map_doppler4}. In the proposed ISAA approach, we tackle \eqref{eq_subprob_pn_supp} by solving the equivalent problem
\begin{align}\label{eq_hybrid_ml_map_doppler_res_supp}
    \bxideltahat  = \arg \min_{\bxidelta} ~ \llr(\tauhat^{(i)}, \nuhat^{(i)}, \bxidelta + \bxihat^{(i)}) ~,
\end{align}
which estimates the residual PN $\bxidelta = \bxi - \bxihat^{(i)}$ in \eqref{eq_bxidelta} instead of the actual PN $\bxi$ given the PN estimate $\bxihat^{(i)}$ from the previous iteration. In the PN update step in Line~\ref{alg_pn_step} of Algorithm~\ref{alg_map_isaa}, the problem \eqref{eq_subprob_pn_supp}, or, equivalently \eqref{eq_hybrid_ml_map_doppler_res_supp}, is solved in closed-form by applying the SAA for $\bxidelta$:
\begin{align}\label{eq_saa_supp}
    e^{-j\bxidelta} \approx \boldone - j \bxidelta ~.
\end{align}}

\subsubsection{\rev{Subproblem for Delay-Doppler Estimation}}
\rev{The subproblem of \eqref{eq_hybrid_ml_map_doppler3} for delay-Doppler estimation at the $\thn{i}$ iteration is given by
\begin{align}\label{eq_hybrid_ml_map_doppler3_fixed_supp}
    (\tauhat^{(i+1)}, \nuhat^{(i+1)})  = \arg \min_{ \tau, \nu} ~ \llr(\tau, \nu, \bxihat^{(i+1)}) ~.
\end{align}
The delay-Doppler update in Line~\ref{alg_dd_step} of Algorithm~\ref{alg_map_isaa} finds an optimal solution of the problem \eqref{eq_hybrid_ml_map_doppler3_fixed_supp}.}

\subsection{\rev{Convergence Result}}
\rev{Based on the subproblem definitions in Sec.~\ref{sec_subproblems}, we can present the convergence result for Algorithm~\ref{alg_map_isaa}.
\begin{proposition}\label{prop_conv}
    Assume that the SAA in \eqref{eq_saa_supp} holds (i.e., the residual PN is small) and there exists a unique delay-Doppler pair that minimizes the objective in \eqref{eq_hybrid_ml_map_doppler3_fixed_supp}. Then, Algorithm~\ref{alg_map_isaa}, which solves \eqref{eq_subprob_pn_supp} and \eqref{eq_hybrid_ml_map_doppler3_fixed_supp} in an alternating fashion, converges to a stationary point of the hybrid ML/MAP optimization problem in \eqref{eq_hybrid_ml_map_doppler3}.
\end{proposition}
\begin{proof}
    When the SAA in \eqref{eq_saa_supp} holds, the closed-form solution given by \eqref{eq_pn_res_est_doppler} and \eqref{eq_bxi_update} provides the optimal solution of the PN estimation subproblem in \eqref{eq_subprob_pn_supp}, or, equivalently \eqref{eq_hybrid_ml_map_doppler_res_supp}, and is therefore a unique solution (being a closed-form one). Then, relying on the assumption that the solution to \eqref{eq_hybrid_ml_map_doppler3_fixed_supp} is unique, it follows from \cite[Proposition~1]{aubry2018new} that Algorithm~\ref{alg_map_isaa} converges to a stationary point of \eqref{eq_hybrid_ml_map_doppler3}.
\end{proof}}

%%%%%%%%%%%%%%%%%%%%%%%%%%%%%%
%%%%%%%%%%%%%%%%%%%%%%%%%%%%%%%%%%%%%%%%%
\section{Complexity Analysis of Algorithm~\ref{alg_map_isaa}}\label{app_comp_alg_map_isaa}

In this section, we analyze the per-iteration complexity of Algorithm~\ref{alg_map_isaa}. We first focus on the update of PN estimates via \eqref{eq_pn_res_est_doppler} and \eqref{eq_bxi_update}, and then study the complexity of updating delay-Doppler estimates via \eqref{eq_llrmax} and \eqref{eq_llrtilde}.

\subsection{Complexity of PN Estimation in \eqref{eq_pn_res_est_doppler}--\eqref{eq_bxi_update}}
For convenience, we repeat the residual PN estimate in \eqref{eq_pn_res_est_doppler} here:
\begin{align} \nonumber
    \bxideltahat(\tau, \nu)
    &= - \boldR(\tau) \Big(  \realp{\boldD^{\rev{(i)}}(\tau, \nu)} \boldR(\tau) + \sigma^2  \Imatrix \Big)^{-1} 
    \\ \label{eq_pn_res_est_doppler_reg_supp}
    &~~~~~\times \left( \imp{\boldD^{\rev{(i)}}(\tau, \nu)} \boldone  + \sigma^2 \boldR(\tau)^{-1} \bxihat^{(i)} \right) ~.
\end{align}
In the following, the complexity of \eqref{eq_pn_res_est_doppler_reg_supp} is analyzed in four steps.

\subsubsection{Complexity of $\boldR(\tau)^{-1} \bxihat^{(i)}$}
For efficient computation of $\boldR(\tau)^{-1} \bxihat^{(i)}$, the conjugate gradient (CG) method can be employed, similar to \cite{OFDM_Joint_PHN_TSP_2006}. Note that computing
\begin{align} \label{eq_xx_rtau_cg}
    \xx = \boldR(\tau)^{-1} \bxihat^{(i)}
\end{align}
is equivalent to solving the following linear system of equations for $\xx$ \cite{OFDM_Joint_PHN_TSP_2006}:
\begin{align} \label{eq_cg_boldr}
    \boldR(\tau) \xx = \bxihat^{(i)} ~.
\end{align}
The CG method provides an iterative procedure to solve \eqref{eq_cg_boldr}, where at each iteration, the major complexity results from a matrix-vector multiplication \cite{OFDM_Joint_PHN_TSP_2006}
\begin{align} \label{eq_matvec}
    \bomega = \boldR(\tau) \bkappa
\end{align}
for some $\bkappa \in \realset{NM}{1}$, where $\boldR(\tau) \in \realset{NM}{NM}$ is the PN covariance matrix given in \eqref{eq_toep_block}. 

Let us define $\bomega = [\bomega_0^T \, \ldots \, \bomega_{M-1}^T]^T$ and $\bkappa = [\bkappa_0^T \, \ldots \, \bkappa_{M-1}^T]^T$, where $\bomega_m \in \realset{N}{1}$ and $\bkappa_m \in \realset{N}{1}$ for each $m$. Then, using the Toeplitz-block Toeplitz structure of $\boldR(\tau)$ in \eqref{eq_toep_block}, we can re-write \eqref{eq_matvec} as
\begin{align} \label{eq_omegam}
    \bomega_m = \sum_{\ell = 0}^{m-1} \boldR_{m-\ell}^T(\tau) \bkappa_{\ell} + \sum_{\ell = 
    m}^{M-1} \boldR_{\ell-m}(\tau) \bkappa_{\ell}
\end{align}
for $m = 0, \ldots, M-1$, where $\boldR_m(\tau) \in \realset{N}{N}$ is a Toeplitz matrix. Using the circulant approximation of Toeplitz matrices for large $N$ \cite{Toeplitz_Circulant_88}, $\boldR_m(\tau)$ can be approximated as
\begin{align}\label{eq_circ_approx_supp}
    \boldR_m(\tau) \approx \boldRcirc_m(\tau) = \FF_N \Lambdab_m(\tau) \FF_N^H ~,
\end{align}
where $\FF_N \in \complexset{N}{N}$ is the unitary DFT matrix, $\boldRcirc_m(\tau) \in \realset{N}{N}$ is the best circulant approximation to $\boldR_m(\tau)$ that minimizes $\norm{\boldR_m(\tau) - \boldRcirc_m(\tau)}_F$ \cite{Toeplitz_Circulant_88}, and 
\begin{align} \label{eq_eigs_rcirc}
    \Lambdab_m(\tau) = \diag{\sqrt{N} \FF_N^H \cc_m(\tau) } \in \complexset{N}{N} ~,
\end{align}
with $\cc_m(\tau)$ denoting the first column of $\boldRcirc_m(\tau)$.

Based on the approximation in \eqref{eq_circ_approx_supp}, each summand in the second term of \eqref{eq_omegam} can be expressed as
\begin{align}
    \boldR_{\ell-m}(\tau) \bkappa_{\ell} \approx \FF_N \Lambdab_{\ell-m}(\tau) \FF_N^H \bkappa_{\ell} ~,
\end{align}
which can be computed efficiently using FFT and IFFT, leading to $\mathcal{O}(N \log N)$ complexity. Similarly, for the summands in the first term of \eqref{eq_omegam}, we must compute
\begin{align}
    \boldR_{m-\ell}^T(\tau) \bkappa_{\ell} \approx \FF_N^H \Lambdab_{m-\ell}(\tau) \FF_N \bkappa_{\ell} ~,
\end{align}
which results in $\mathcal{O}(N \log N)$ complexity. In practice, as seen from the correlation characteristics of PN in Fig.~\ref{fig_PN_covariance_PLL_FRO}, the PN correlation vanishes after a certain time lag, meaning that only a few blocks in $\boldR(\tau)$, say $M_0$, are dominant in the computation of the right-hand side of \eqref{eq_omegam} (i.e., $\boldR_0(\tau), \ldots, \boldR_{M_0-1}(\tau)$), where $M_0 = 1$ for FROs due to the block-diagonal structure derived in Lemma~\ref{lemma_fro_blkdiag} and $1 \leq M_0 \ll M$ for PLLs (typically, $M_0 \leq 3$). Hence, the complexity of \eqref{eq_omegam} is given by
\begin{align}
    \mathcal{O}(M_0 N \log N) ~,
\end{align}
which yields the following complexity for \eqref{eq_matvec}:
\begin{align} \label{eq_comp_rkappa}
    \mathcal{O}(M_0 M N \log N) ~.
\end{align}
Assuming that the CG method requires $I$ iterations to converge, the complexity of evaluating $\boldR(\tau)^{-1} \bxihat^{(i)}$ in \eqref{eq_xx_rtau_cg} can be expressed as
\begin{align} \label{eq_comp_step1}
    \mathcal{O}(I M_0 M N \log N) ~.
\end{align}

\subsubsection{Complexity of $\imp{\boldD^{\rev{(i)}}(\tau, \nu)} \boldone$}
Using the definition of $\boldD^{\rev{(i)}}(\tau, \nu)$ in \eqref{eq_gammanutau}, we can write
\begin{align} \label{eq_boldDtau}
    & \boldD^{\rev{(i)}}(\tau, \nu)
    \\ \nonumber
    &= \left( \diag{\yy}^H \projnull{ \qq(\tau, \nu)}  \diag{\yy} \right) \odot \left( e^{-j  \bxihat^{(i)}  } (e^{j  \bxihat^{(i)}  })^T \right)
    \\ \nonumber
    &= \left[ \diag{\yy}^H \diag{\yy} - \left(\yy^\conj \yy^T \odot \frac{\qq(\tau, \nu) \qq^H(\tau, \nu)}{\norm{\qq(\tau, \nu)}^2} \right) \right]
    \\ \nonumber &~~\odot \left( e^{-j \bxihat^{(i)}} (e^{j \bxihat^{(i)}})^T \right)
    \\ \nonumber
    &= \diag{\abs{\yy}^2} 
    \\ \nonumber
    &~~- \frac{1}{\norm{\qq(\tau, \nu)}^2} \left(\yy^\conj \odot \qq(\tau, \nu) \odot e^{-j \bxihat^{(i)}} \right) \left(\yy^\conj \odot \qq(\tau, \nu) \odot e^{-j \bxihat^{(i)}} \right)^H
    \\ \nonumber
    &= \diag{\abs{\yy}^2} - \frac{1}{\norm{\boldX}_F^2}
     \left(\yy^\conj \odot \boldXihat^{(i)} \qq(\tau, \nu)  \right) \left(\yy^\conj \odot \boldXihat^{(i)} \qq(\tau, \nu)  \right)^H
\end{align}
where $\boldXihat^{(i)} \triangleq \diag{e^{-j \bxihat^{(i)}}}$. Using \eqref{eq_boldDtau}, the real and imaginary parts of $\boldD^{\rev{(i)}}(\tau, \nu)$ are given by
\begin{align}\label{eq_reald}
    &\realp{\boldD^{\rev{(i)}}(\tau, \nu)} = \diag{\abs{\yy}^2} 
    \\ \nonumber
    &~~- \frac{1}{\norm{\boldX}_F^2}
     \realp{\left(\yy^\conj \odot \boldXihat^{(i)} \qq(\tau, \nu)  \right) \left(\yy^\conj \odot \boldXihat^{(i)} \qq(\tau, \nu)  \right)^H} ~,
     \\ \label{eq_imagd}
      &\imp{\boldD^{\rev{(i)}}(\tau, \nu)} \\ \nonumber
      &~~= - \frac{1}{\norm{\boldX}_F^2}
     \imp{\left(\yy^\conj \odot \boldXihat^{(i)} \qq(\tau, \nu)  \right) \left(\yy^\conj \odot \boldXihat^{(i)} \qq(\tau, \nu)  \right)^H} ~.
\end{align}

Using \eqref{eq_imagd}, we can compute $\imp{\boldD^{\rev{(i)}}(\tau, \nu)} \boldone$ as
\begin{align}
    & \imp{\boldD^{\rev{(i)}}(\tau, \nu)} \boldone 
    \\ \nonumber
    &= - \frac{1}{\norm{\boldX}_F^2}
     \imp{\left(\yy^\conj \odot \boldXihat^{(i)} \qq(\tau, \nu)  \right) \left(\yy^\conj \odot \boldXihat^{(i)} \qq(\tau, \nu)  \right)^H \boldone} ~,
\end{align}
which requires
\begin{align}\label{eq_comp_step2}
    \mathcal{O}(NM)
\end{align}
operations.

\subsubsection{Complexity of $\Big(  \realp{\boldD^{\rev{(i)}}(\tau, \nu)} \boldR(\tau) + \sigma^2  \Imatrix \Big)^{-1} \varpib $}
Now that we have computed $\boldR(\tau)^{-1} \bxihat^{(i)}$ and $\imp{\boldD^{\rev{(i)}}(\tau, \nu)} \boldone$ in \eqref{eq_pn_res_est_doppler_reg_supp}, we will analyze the complexity of evaluating
\begin{align} \label{eq_varpib}
    \Big(  \realp{\boldD^{\rev{(i)}}(\tau, \nu)} \boldR(\tau) + \sigma^2  \Imatrix \Big)^{-1} \varpib ~,
\end{align}
where 
\begin{align} \label{eq_varpib_def}
    \varpib = \imp{\boldD^{\rev{(i)}}(\tau, \nu)} \boldone  + \sigma^2 \boldR(\tau)^{-1} \bxihat^{(i)} ~.
\end{align}
Similar to the above analysis, the CG method can be employed to evaluate \eqref{eq_varpib}, where each iteration involves the computation of
\begin{align} \label{eq_varpib2}
    \Big(  \realp{\boldD^{\rev{(i)}}(\tau, \nu)} \boldR(\tau) + \sigma^2  \Imatrix \Big) \bkappa 
\end{align}
for some $\bkappa \in \realset{NM}{1}$. Opening up the terms in \eqref{eq_varpib2}, we have
\begin{align} \label{eq_varpib3}
    \realp{\boldD^{\rev{(i)}}(\tau, \nu)} \boldR(\tau)\bkappa + \sigma^2 \bkappa ~.
\end{align}
As previously derived in \eqref{eq_comp_rkappa}, evaluation of $\boldR(\tau)\bkappa$ in \eqref{eq_varpib3} has the complexity of 
\begin{align} \label{eq_order_comp}
    \mathcal{O}(M_0 M N \log N) ~.
\end{align}
Next, defining $\bomega = \boldR(\tau)\bkappa$, we compute $\realp{\boldD^{\rev{(i)}}(\tau, \nu)} \bomega$ in \eqref{eq_varpib3} using \eqref{eq_reald} as
\begin{align}
    &\realp{\boldD^{\rev{(i)}}(\tau, \nu)} \bomega
    \\ \nonumber
    &= \diag{\abs{\yy}^2} \bomega
    \\ \nonumber
    &~~- \frac{1}{\norm{\boldX}_F^2}
     \realp{\left(\yy^\conj \odot \boldXihat^{(i)} \qq(\tau, \nu)  \right) \left(\yy^\conj \odot \boldXihat^{(i)} \qq(\tau, \nu)  \right)^H \bomega} ~,
\end{align}
which leads to a complexity of $\mathcal{O}(NM)$. Combining this with \eqref{eq_order_comp}, the complexity of computing the expression in \eqref{eq_varpib2} becomes
\begin{align} \label{eq_order_comp2}
    \mathcal{O}(M_0 M N \log N)~. 
\end{align}
Finally, assuming $I$ iterations for the CG method to converge, \eqref{eq_varpib} has a computational complexity of
\begin{align} \label{eq_comp_step3}
    \mathcal{O}(I M_0 M N \log N) ~.
\end{align}

\subsubsection{Final Complexity Result for \eqref{eq_pn_res_est_doppler_reg_supp}}
After evaluating \eqref{eq_varpib}, the final step to obtain \eqref{eq_pn_res_est_doppler_reg_supp} is to compute
\begin{align} \label{eq_varsigmab}
     \boldR(\tau)\varsigmab ~,
\end{align}
where
\begin{align}
    \varsigmab = \Big(  \realp{\boldD^{\rev{(i)}}(\tau, \nu)} \boldR(\tau) + \sigma^2  \Imatrix \Big)^{-1} \varpib ~,
\end{align}
with $\varpib$ being defined in \eqref{eq_varpib_def}. From \eqref{eq_comp_rkappa}, evaluation of \eqref{eq_varsigmab} involves
\begin{align} \label{eq_comp_step4}
    \mathcal{O}(M_0 M N \log N) 
\end{align}
operations. Combining \eqref{eq_comp_step1}, \eqref{eq_comp_step2}, \eqref{eq_comp_step3} and \eqref{eq_comp_step4}, the computational complexity of  \eqref{eq_pn_res_est_doppler_reg_supp} is obtained as
\begin{align} \label{eq_comp_step5}
    \mathcal{O}(I M_0 M N \log N) ~.
\end{align}

%%%%%%%%%%%%%%%%%%%%%%%%%%%%%%
\subsection{Complexity of Delay-Doppler Estimation in \eqref{eq_llrmax}--\eqref{eq_llrtilde}}
To evaluate \eqref{eq_llrtilde} in practice, it is enough to compute the first term as the remaining terms are negligibly small and can be discarded. We notice that the first term in \eqref{eq_llrtilde} (up to a multiplicative constant)
\begin{align}
    & \absbig{ \bb^H(\tau) \left( \boldX^\conj \odot \FF_N \big( \Wbhat^\conj \odot \boldY \big)  \right) \cc(\nu) }^2
\end{align}
can be evaluated over a delay-Doppler grid via
\begin{align}\label{eq_fft_delayDoppler}
    \absbig{ \FF_N^H \left( \boldX^\conj \odot \FF_N \big( \Wbhat^\conj \odot \boldY \big)  \right) \FF_M }^2 ~,
\end{align}
utilizing the fact that $\bb(\tau)$ in \eqref{eq_steer_delay} and $\cc(\nu)$ in \eqref{eq_steer_doppler} correspond to DFT matrix columns for uniformly sampled delay-Doppler values. The complexity of \eqref{eq_fft_delayDoppler} is thus given by
\begin{align} \label{eq_order_fft}
    \mathcal{O}(NM) + \mathcal{O}(M N \log N) + \mathcal{O}(N M \log M) ~,
\end{align}
where $ \mathcal{O}(NM)$ is due to element-wise multiplication of $N \times M$ matrices, $\mathcal{O}(M N \log N) $ results from $M$ times $N$-point FFT/IFFT, and $\mathcal{O}(N M \log M) $ comes from $N$ times $M$-point FFT. Re-writing \eqref{eq_order_fft}, the complexity of \eqref{eq_fft_delayDoppler} can be expressed as
\begin{align} \label{eq_order_fft2}
     \mathcal{O}(NM \log(NM)) ~.
\end{align}

%%%%%%%%%%%%%%%%%%%%%%%%%%%%%%
\subsection{Complexity of Algorithm~\ref{alg_map_isaa}}
Using the complexity of PN estimation in \eqref{eq_comp_step5} and that of delay-Doppler estimation in \eqref{eq_order_fft2}, the per-iteration complexity of Algorithm~\ref{alg_map_isaa} can be written as
\begin{align} \label{eq_order_algo}
     \mathcal{O}\big( MN  \left( \log M + (I M_0 +1) \log N \right) \big) ~.
\end{align}

% %%%%%%%%%%%%%%%%%%%%%%%%%%%%%%
% %%%%%%%%%%%%%%%%%%%%%%%%%%%%%%%%%%%%%%%%%

\section{\rev{Theoretical Ground for Formulation of PN Exploitation Problem in \eqref{eq_parametric_Frob}}}\label{sec_cov_match}
\rev{In this section, we provide the theoretical motivation behind the parametric covariance matrix reconstruction problem in \eqref{eq_parametric_Frob}, which is formulated to resolve range ambiguity via PN exploitation. The goal of PN exploitation is to estimate the unknown delay $\tau$ from the PN estimate $\bxihat \in \realset{NM}{1}$ obtained at the output of Algorithm~\ref{alg_map_isaa}. Intuitively, this can lead to unambiguous estimates of $\tau$ since no ambiguity exists in PN covariance $\boldR(\tau)$ with respect to $\tau$.}

\vspace{-0.2in}
\subsection{\rev{PN Observation Model and ML Estimator of Delay}}
\rev{The observation model for the above estimation problem can be written as
\begin{align}\label{eq_bxihat_obs}
    \bxihat = \bxi + \nn ~,
\end{align}
where $\bxi \in \realset{NM}{1}$ denotes the true PN vector with $\bxi \sim \mtN(\boldzero, \boldR(\tau)) $ and $\nn \in \realset{NM}{1}$ is the estimation noise, independent of $\bxi$, whose statistics are given by $\nn \sim \mtN(\boldzero, \boldSigma(\tau)  ) $. Here, $\boldSigma(\tau) \in \realset{NM}{NM}$ represents the covariance matrix of the PN estimation error in Algorithm~\ref{alg_map_isaa} and can be set to the CRB matrix on PN estimation for the hybrid ML/MAP problem \eqref{eq_hybrid_ml_map_doppler3}. Accordingly, the statistics of $\bxihat$ in \eqref{eq_bxihat_obs} are 
\begin{align}\label{eq_bxihat_stat}
    \bxihat \sim \mtN(\boldzero, \boldR(\tau) + \boldSigma(\tau)  ) ~.
\end{align}
The ML estimator of $\tau$ from the PN estimates in \eqref{eq_bxihat_stat} can be obtained as \cite{ML_radar_array_98}
\begin{align}\label{eq_ml_tau_bxi}
    \tauhat  = \arg \min_{ \tau} &~ \Big\{ \bxihat^T \big( \boldR(\tau) + \boldSigma(\tau)\big)^{-1} \bxihat
    \\ \nonumber
    &~~+ \log \det \big( \boldR(\tau) + \boldSigma(\tau) \big) \Big\}   ~.
\end{align}}

\rev{We identify several challenges in solving \eqref{eq_ml_tau_bxi}. First, the inverse of the covariance matrix of $\bxihat$ needs to be calculated, which leads to a high computational burden. Second, only a single sample of PN estimate is available to solve \eqref{eq_ml_tau_bxi}, i.e., the sample covariance matrix $\bxihat \bxihat^T$ is rank-one, leading to poor estimates of $\tau$. Additionally, the formulation in \eqref{eq_ml_tau_bxi} does not allow us to exploit redundancy in $\boldR(\tau)$ stemming from its Toeplitz-block Toeplitz structure in \eqref{eq_toep_block}.}

\vspace{-0.1in}
\subsection{\rev{Covariance Matching Approach}}
\rev{To overcome the aforementioned challenges, we propose a covariance matching (i.e., parametric covariance matrix reconstruction) approach that fits the sample covariance to the analytical model using a weighted least squares (WLS) formulation \cite{ottersten1998covariance}:
\begin{align} \nonumber
    \tauhat  = \arg \min_{ \tau} &~ \Big\{  \big[ \veccs{\boldRhat} - \veccs{\boldR(\tau) + \boldSigma(\tau)} \big]^T \big( \boldRhat^T \otimes \boldRhat \big)^{-1} 
    \\ \label{eq_wls_cov_match}
    &~~ \times   \big[ \veccs{\boldRhat} - \veccs{\boldR(\tau) + \boldSigma(\tau)} \big] \Big\} ~,
\end{align}
where $\boldRhat = \bxihat \bxihat^T $
% \begin{align}\label{eq_samp_cov_r}
%     \boldRhat = \bxihat \bxihat^T 
% \end{align}
is the sample covariance matrix. By the extended invariance principle (EXIP), the proposed covariance matching based estimator \eqref{eq_wls_cov_match} is asymptotically equivalent to the original ML estimator \eqref{eq_ml_tau_bxi} at high SNRs \cite{ottersten1998covariance}. Since $\boldRhat$ is rank-one, $\big( \boldRhat^T \otimes \boldRhat \big)$ is also rank-one and thus not invertible. Hence, we resort to the LS reformulation of \eqref{eq_wls_cov_match}:
\begin{align} \label{eq_wls_cov_match2}
    \tauhat  = \arg \min_{ \tau} &~ \norm{ \veccs{\boldRhat} - \veccs{\boldR(\tau) + \boldSigma(\tau)} }^2 ~.
\end{align}}

\vspace{-0.3in}
\subsection{\rev{High-SNR Approximation}}
\rev{The Toeplitz-block Toeplitz structure of $\boldR(\tau)$ cannot be exploited in \eqref{eq_wls_cov_match2} since $\boldSigma(\tau)$ can have a generic covariance structure. Moreover, the sample covariance $\boldRhat$ is constructed using a single observation $\bxihat$ and thus constitutes a very inaccurate estimate of the true covariance $\boldR(\tau) + \boldSigma(\tau)$. To deal with these issues, we propose to make the approximation
\begin{align}
    \boldR(\tau) + \boldSigma(\tau) \approx \boldR(\tau) ~,
\end{align}
which is valid at high SNRs since the CRB matrix elements are inversely proportional to SNR (i.e., PN estimation becomes sufficiently accurate at high SNRs such that $\boldR(\tau) + \boldSigma(\tau)$ is dominated by the PN covariance $\boldR(\tau)$). By virtue of this high-SNR approximation, the problem \eqref{eq_wls_cov_match2} becomes
\begin{align} \label{eq_wls_cov_match3}
    \tauhat  = \arg \min_{ \tau} &~ \norm{ \veccs{\boldRhat} - \veccs{\boldR(\tau)} }^2 ~,
\end{align}
which coincides with the formulation in \eqref{eq_parametric_Frob}. As seen from \eqref{eq_parametric_Frob2} and \eqref{eq_rhat_row}, the Toeplitz-block Toeplitz structure of $\boldR(\tau)$ can now be effectively exploited in \eqref{eq_wls_cov_match3} to estimate delay from a single observation of the PN vector.}

\end{appendices}

%%%%%%%%%%%%%%%%%%%%%%%%%%%%%%%%%%%%%%%%%%%%%%%%%%%%%%%%
%%%%%%%%%%%%%%%%%%%%%%%%%%%%%%%%%%%%%%%%%%%%%%%%%%%%%%%%
\bibliographystyle{IEEEtran}
\bibliography{main}
%\bibliography{IEEEabrv,ofdm_dfrc}

\end{document}